\documentclass[a4paper]{article}

\usepackage{amsmath,amssymb,amsthm,amsfonts}   
\numberwithin{equation}{section}   
\usepackage{physics}   
\usepackage[english]{babel}   
\usepackage[margin=2.54cm]{geometry}   
\usepackage{csquotes}   
\usepackage{mathtools}   
\usepackage{relsize}   
\usepackage{indentfirst}  
\usepackage{thm-restate}   
\usepackage{mathrsfs}   
\usepackage[dvipsnames]{xcolor}   
\usepackage{paracol}    
\usepackage{enumitem}    
\usepackage{soul}       
\usepackage{aligned-overset}    
\usepackage[thinc]{esdiff}  
\usepackage{upgreek}    
\usepackage{datetime}   
\usepackage{tensor}     
\usepackage{caption}    
\usepackage{subcaption} 
\usepackage{appendix}   
\allowdisplaybreaks

\usepackage[dvipsnames]{xcolor}     
\usepackage{hyperref}	            
\newcommand\myshade{85}         
\colorlet{mylinkcolor}{NavyBlue}  
\colorlet{mycitecolor}{YellowOrange}    
\colorlet{myurlcolor}{Aquamarine}       
\hypersetup{            
    linkcolor  = mylinkcolor!\myshade!black,   
    citecolor  = mycitecolor!\myshade!black,   
    urlcolor   = myurlcolor!\myshade!black, 
    colorlinks = true,      
    hypertexnames=false,
}               
\usepackage{cleveref}                       
\usepackage[                                
style = alphabetic,                         
maxbibnames=9,maxcitenames=9,               
maxalphanames=4,minalphanames=4,            
doi=true,isbn=true,giveninits=true,         
backend=biber,                              
url=false,arxiv=true]{biblatex}             
\theoremstyle{plain}   

\crefname{thm}{theorem}{theorems}
\theoremstyle{plain}   
\newtheorem{assumption}{Assumption}     
\crefname{assumption}{assumption}{assumptions}

\newtheorem{lemma}{Lemma}[section]   
\crefname{lemma}{lemma}{lemmas}
\newtheorem{theorem}{Theorem}[section]  

\crefname{bigcor}{corollary}{corollaries}
\Crefname{bigcor}{Corollary}{Corollaries}

\newtheorem{definition}{Definition}[section]    %
\newtheorem{prop}[lemma]{Proposition}   
\newtheorem{cor}[lemma]{Corollary}   
\newtheorem{remark}[lemma]{Remark}   
\usepackage{fancyhdr}   
\newtheoremstyle{examplestyle}%
  {3pt}{3pt}
  {\normalfont}
  {}
  {\bfseries}
  {.}
  {0.5em}
  {}
\theoremstyle{examplestyle}
\newtheorem{example}{Example}

\crefname{condition}{condition}{conditions}
\Crefname{condition}{Condition}{Conditions}

\usepackage{xparse} 
\ExplSyntaxOn       
\NewDocumentCommand{\definealphabet}{mmmm}{
  \int_step_inline:nnn { `#3 } { `#4 }{
    \cs_new_protected:cpx { #1 \char_generate:nn { ##1 }{ 11 } }{
      \exp_not:N #2 { \char_generate:nn { ##1 } { 11 } }}}}              
\ExplSyntaxOff      
\definealphabet{mb}{\mathbb}{A}{Z}      
\definealphabet{mc}{\mathcal}{A}{Z}     
\definealphabet{scr}{\mathscr}{A}{Z}    
\definealphabet{mf}{\mathfrak}{A}{Z}    

\DeclareMathOperator{\states}{\mbS_d}   

\DeclarePairedDelimiterX{\corr}[2]{\mathrm{corr}(}{)}{#1,#2}     
\DeclarePairedDelimiterX{\cov}[2]{\mathrm{cov}(}{)}{#1,#2}          
\DeclarePairedDelimiterX{\dis}[2]{\mathrm{d}(}{)}{#1, #2} 

\DeclareMathOperator{\dummyborel}{\mathscr{B}}                  
\newcommand\borel[1]{\dummyborel\left(#1\right)}    
\DeclareMathOperator{\pr}{\mbP}           
\newcommand{\adj}{^\dagger}         
\DeclareMathOperator{\outcomes}{\mcA^\mbN}      
\DeclarePairedDelimiterX{\inner}[2]{\langle}{\rangle}{#1| #2}   
\newcommand{\s}{\rho_{\mathsf{ss}}}     
\newcommand{\dee}{\mathrm{d}}           

\newcommand{\esp}{{\hyperlink{esp}{ESP}}}

\renewcommand{\tr}[1]{\operatorname{Tr}\left(#1\right)} 

\DeclareMathOperator{\str}{tr} 

\DeclarePairedDelimiterX{\cnum}[1]{\mathrm{c}(}{)}{#1}  
\DeclarePairedDelimiterX{\diam}[1]{\text{diam}(}{)}{#1} 
\newcommand{\proj}{{\mathlarger{\boldsymbol{\cdot}}}} 
\let\oldker\ker         
\renewcommand{\ker}[1]{\oldker\left(#1\right)}  

\title{Central Limit Theorems for Outcome Records in Disordered Quantum Trajectories}
\usepackage{authblk}

\author[]{Lubashan Pathirana\thanks{lpk@math.ku.dk}}
\affil[]{Department of Mathematical Sciences and QMATH, University of Copenhagen, Denmark.}
\date{}

\DeclareMathOperator{\Cov}{Cov}

\newcommand{\TV}{d_{\mathrm{TV}}}

\newcommand{\unorm}[1]{\left\lVert#1\right\rVert}
\usepackage{longtable,booktabs,array,tabularx}
\begin{document}
\pagenumbering{arabic}
\lhead{\thepage}
\maketitle
\vspace{-1cm}


\begin{abstract}
    We prove annealed functional central limit theorems for finite pattern counts in the measurement record of discrete-time quantum trajectories, with the instrument applied at each step determined by an invertible, probability-preserving base dynamical system.
    When the base is ergodic, under summable strong-mixing coefficients of the instrument process and a summable uniform annealed trace-norm forgetting rate for the associated non-selective channel cocycle, we establish a joint functional CLT for bounded vector-valued functions of finite outcome blocks under the annealed law determined by the dynamically stationary state.
    We then extend this limit to every measurable random initial state, yielding a universal functional CLT with unchanged stationary centering and asymptotic covariance.
    We also provide practical sufficient criteria ensuring the existence and uniqueness of the dynamically stationary state and the required annealed trace-norm forgetting.
    We illustrate the results through a broad family of examples, including disordered walk-type models generated by finite group actions, measurement followed by preparation, and instruments with reset components.
    The results apply to general disordered quantum instruments and are not restricted to the perfect-measurement regime; they complement the law of large numbers established in \cite{qtlln} for the same disordered setting and provide a disordered counterpart of the homogeneous CLT in \cite{Attal_2014}.
\end{abstract}




\section{Introduction}
\label{section:intro}


The statistics of repeated measurement outcomes of a quantum system are modeled by a \emph{quantum instrument} \(\mcV\) \cite{Davies_1970}.
Let \(\mcH=\mbC^d\) be a finite-dimensional Hilbert space and let \(\mcA\) be a finite outcome set.
We write \(\mcB(\mcH)\) for the algebra of linear operators on \(\mcH\) and \(\states\) for the state space of \(\mcH\).
A quantum instrument is a family \(\mcV:=\{\mcT_a\}_{a\in\mcA}\) of completely positive, trace-non-increasing maps \(\mcT_a:\mcB(\mcH)\to\mcB(\mcH)\) such that
\[
    \Phi:=\sum_{a\in\mcA}\mcT_a
\]
is trace-preserving.
Thus \(\Phi\) is a completely positive and trace-preserving (CPTP) map, or a \emph{quantum channel} \cite{Watrous_2018}.

We call this the \emph{perfect-measurement} regime when each instrument component has a single-Kraus representation, \(\mcT_a(\rho)=V_a\rho V_a\adj\).
More generally, a recorded outcome may combine several microscopic alternatives, in which case
\[
    \mcT_a(\rho)=\sum_{j=1}^{r(a)}V_{a,j}\rho V_{a,j}\adj.
\]
We refer to this more general setting as the \emph{imperfect-measurement} regime.
It arises naturally from mixed probe preparations, degenerate readout, detector misidentification, or coupling to unmonitored degrees of freedom; see, e.g., \cite{tristant_2025,somaraju2012design}.
For example, suppose that \((V_j)_j\) describes an ideal measurement and that \(\eta_{i,j}\) is the probability of reporting \(i\) when the microscopic outcome is \(j\), so that \(\eta_{i,j}\ge 0\) and \(\sum_{i}\eta_{i,j} = 1\) for each $j$.
The corresponding reported-outcome instrument is then
\[
    \mcT_i(\rho)=\sum_j\eta_{i,j}V_j\rho V_j\adj.
\]

Given an initial state \(\rho_0\in\states\) and an outcome string \((a_1,\ldots,a_n)\in\mcA^n\), the corresponding posterior state after \(n\) repeated measurements is
\begin{equation}
    \rho_n
    :=
    \frac{(\mcT_{a_n}\circ\cdots\circ\mcT_{a_1})(\rho_0)}
         {\tr{(\mcT_{a_n}\circ\cdots\circ\mcT_{a_1})(\rho_0)}},
\end{equation}
whenever the denominator is nonzero. By Born's rule \cite{Born1926}, the probability of observing the finite outcome string \((a_1,\ldots,a_n)\) is
\begin{equation}
    \tr{(\mcT_{a_n}\circ\cdots\circ\mcT_{a_1})(\rho_0)}.
\end{equation}

These finite-dimensional distributions define, by the Kolmogorov extension theorem, a unique probability measure, $\mbQ_{\rho_0}$, on the outcome space \((\mcA^{\mbN},\Sigma)\), where \(\Sigma\) is the product \(\sigma\)-algebra, such that for every cylinder set
\[
    \{a_1\}\times\cdots\times\{a_n\}\times\mcA\times\mcA\times\cdots
\]
one has
\begin{equation}
    \mbQ_{\rho_0}\!\left(
        \{a_1\}\times\cdots\times\{a_n\}\times\mcA\times\mcA\times\cdots
    \right)
    =
    \tr{(\mcT_{a_n}\circ\cdots\circ\mcT_{a_1})(\rho_0)}.
\end{equation}
Thus \(\mbQ_{\rho_0}\) is the law of the measurement record associated with the initial state \(\rho_0\), governed by the instrument $\mcV$.

Let
\[
    A_n:\mcA^\mbN\to\mcA,
    \qquad
    A_n(\bar a)=a_n,
\]
be the canonical coordinate maps. For \(L\in\mbN\), we also write
\[
    A_{1:L}:=(A_1,\dots,A_L):\mcA^\mbN\to\mcA^L
\]
for the block-coordinate map. The associated posterior-state process \((\rho_n)_{n\ge0}\) is then defined \(\mbQ_{\rho_0}\)-almost surely by
\[
    \rho_n(\bar a)
    :=
    \frac{
        (\mcT_{A_n(\bar a)}\circ\cdots\circ\mcT_{A_1(\bar a)})(\rho_0)
    }{
        \tr{
            (\mcT_{A_n(\bar a)}\circ\cdots\circ\mcT_{A_1(\bar a)})(\rho_0)
        }
    }.
\]
Thus \((A_n)_{n\ge1}\) is the \emph{outcome process}, while \((\rho_n)_{n\ge0}\) is the corresponding \emph{quantum trajectory}, with the initial state $\rho_0$.


\subsection{Disordered Quantum Trajectories}
\label{section:Disordered_Traj}


    We now allow the instrument to depend on an external disorder parameter. 
    Let \((\Omega,\mcF,\pr,\theta)\) be a probability space equipped with a \(\pr\)-preserving transformation \(\theta:\Omega\to\Omega\). 
    For each \(\omega\in\Omega\), let \(\mcV_\omega=\{\mcT_{a;\omega}\}_{a\in\mcA}\) be a quantum instrument on \(\mcB(\mcH)\), so that each \(\mcT_{a;\omega}\) is completely positive and trace-non-increasing and
    \begin{equation}
    \label{eq:non_selective_channel}
        \Phi_\omega:=\sum_{a\in\mcA}\mcT_{a;\omega}
    \end{equation}
    is trace-preserving for \(\pr\)-almost every \(\omega\). 
    We assume that for each \(a\in\mcA\), the map \(\omega\mapsto \mcT_{a;\omega}\) is measurable as a map from \((\Omega,\mcF)\) into the finite-dimensional vector space \(\mcL(\mcB(\mcH))\), equipped with its Borel \(\sigma\)-algebra.
    Equivalently, for each \(\rho\in\mcB(\mcH)\), the map \(\omega\mapsto \mcT_{a;\omega}(\rho)\) is measurable.
    We call \((\Omega,\mcF,\pr,\theta)\) the \emph{base system}, and a point \(\omega\in\Omega\) a \emph{disorder realization}.

    Fix a measurable initial state \(\rho_0:\Omega\to\states\).
    Our convention is that \(\rho_0(\omega)\) is the state at time \(0\), while the \(k\)-th measurement uses the instrument \(\mcV_{\theta^k(\omega)}\).
    The construction above yields a probability measure \(\mbQ_{\rho_0(\omega);\omega}\) on \((\mcA^\mbN,\Sigma)\), characterized by
    \[
        \mbQ_{\rho_0(\omega);\omega}
        (A_1=a_1,\ldots,A_n=a_n)
        =
        \tr{
            \left(
                \mcT_{a_n;\theta^n(\omega)}
                \circ\cdots\circ
                \mcT_{a_1;\theta(\omega)}
            \right)\!\left(\rho_0(\omega)\right)
        }
    \]
    for every \(n\in\mbN\) and every \((a_1,\ldots,a_n)\in\mcA^n\).
    We call \(\mbQ_{\rho_0(\omega);\omega}\) the \emph{quenched outcome law} at disorder realization \(\omega\).

    The map
    \[
        \Omega\ni\omega
        \longmapsto
        \mbQ_{\rho_0(\omega);\omega}
        \in\mcP(\mcA^\mbN)
    \]
    is measurable; see \Cref{lem:Q-random-initial-state}. 
    We refer to \cite{traj,qtlln} for a detailed proof of the quenched path-space construction.
    Although those works use the perfect-measurement formulation, the construction depends only on the instrument axioms and therefore extends unchanged to general quantum instruments.

    For $\bar  a \in \mcA^\mbN$, set \(\rho_0^\omega(\bar a):=\rho_0(\omega)\) and for \(n\ge1\), define the quenched posterior state \(\mbQ_{\rho_0(\omega);\omega}\)-almost surely by
    \[
        \rho_n^\omega(\bar a)
        :=
        \frac{
            \left(
                \mcT_{A_n(\bar a);\theta^n(\omega)}
                \circ\cdots\circ
                \mcT_{A_1(\bar a);\theta(\omega)}
            \right)\!\left(\rho_0(\omega)\right)
        }{
            \tr{
                \left(
                    \mcT_{A_n(\bar a);\theta^n(\omega)}
                    \circ\cdots\circ
                    \mcT_{A_1(\bar a);\theta(\omega)}
                \right)\!\left(\rho_0(\omega)\right)
            }
        }.
    \]
    Here, when 
    \[
        \tr{
                \left(
                    \mcT_{A_n(\bar a);\theta^n(\omega)}
                    \circ\cdots\circ
                    \mcT_{A_1(\bar a);\theta(\omega)}
                \right)\!\left(\rho_0(\omega)\right)
            }
            =0,
    \]
    we define the posterior as a fixed reference state, and that gives a jointly measurable version in $(\omega,\bar a)$: each finite word map is measurable, and division is performed only where its trace is positive (see \Cref{section:appn}, \Cref{eq:appendix-projective-action} and \Cref{lem:proj-action-jointly-measurable}).
    Let
    \[
        \mcH_0:=\{\varnothing,\mcA^\mbN\},
        \qquad
        \mcH_n:=\sigma(A_1,\ldots,A_n),
        \quad n\ge1.
    \]
    Then, for every \(a\in\mcA\) and \(n\ge0\),
    \[
        \mbQ_{\rho_0(\omega);\omega}
        (A_{n+1}=a\mid\mcH_n)
        =
        \tr{\mcT_{a;\theta^{n+1}(\omega)}(\rho_n^\omega)},
        \qquad
        \mbQ_{\rho_0(\omega);\omega}\text{-a.s.}
    \]
    Moreover,
    \[
        \rho_{n+1}^\omega
        =
        \frac{
            \mcT_{a;\theta^{n+1}(\omega)}(\rho_n^\omega)
        }{
            \tr{\mcT_{a;\theta^{n+1}(\omega)}(\rho_n^\omega)}
        },
        \qquad
        \mbQ_{\rho_0(\omega);\omega}\text{-a.s. on }\{A_{n+1}=a\}.
    \]
    Thus, for $\mbP$-almost every fixed \(\omega\), the posterior-state process is a time-inhomogeneous Markov chain whose transitions are marked by the observed outcomes:
    \[
        \rho_n^\omega
        \ \xrightarrow{\;A_{n+1}=a\;}\
        \frac{\mcT_{a;\theta^{n+1}(\omega)}(\rho_n^\omega)}
        {\operatorname{Tr}\!\left(
        \mcT_{a;\theta^{n+1}(\omega)}(\rho_n^\omega)
        \right)}.
    \]

    For a measurable initial state \(\rho_0:\Omega\to\states\), the measurable family of quenched laws \(\mbQ_{\rho_0(\omega);\omega}\)  defines the joint \emph{annealed law}
    \begin{equation}
    \label{eq:annealed}
        \overline{\mbQ}_{\rho_0}(d\omega,d\bar a)
        :=
        \pr(d\omega)\,
        \mbQ_{\rho_0(\omega);\omega}(d\bar a)
    \end{equation}
    on \((\Omega\times\mcA^\mbN,\mcF\otimes\Sigma)\).
    Let \(\varsigma\) denote the left shift on \(\mcA^\mbN\).
    Define the skew shift
    \[
        \tau(\omega,\bar a):=(\theta\omega,\varsigma\bar a).
    \]
    Under \(\overline{\mbQ}_{\rho_0}\), the coordinate process \((A_n)_{n\ge1}\) is the measurement record in the disordered environment.

    \paragraph{Outcome Frequencies.}
        Fix \(m\in\mbN\) and a pattern \(\pmb b=(b_1,\dots,b_m)\in\mcA^m\).
        Define the cylinder set
        \begin{equation}
        \label{eq:cylinder_for_b}
            C_{\pmb b}
            :=
            \{\bar a\in\mcA^\mbN:a_1=b_1,\dots,a_m=b_m\}.
        \end{equation}
        The number of occurrences of \(\pmb b\) whose starting position lies in \(\{1,\dots,n\}\) is
        \[
            N_n^{\pmb b}(\bar a)
            :=
            \sum_{k=1}^n
            \mathbf 1_{C_{\pmb b}}\left(\varsigma^{k-1}\bar a\right).
        \]
        We call \((N_n^{\pmb b})_{n\ge1}\) the \emph{pattern-counting process} associated with \(\pmb b\), and \(N_n^{\pmb b}/n\) its \emph{empirical pattern frequency}.

        \medskip

        The goal of this paper is to establish a central limit theorem for the centered fluctuations of the empirical pattern frequencies, complementing the law of large numbers obtained in \cite{qtlln}.
        The LLN in \cite{qtlln} assumes that the base system is invertible and ergodic and that the associated non-selective channel cocycle admits a unique dynamically stationary state.
        Although \cite{qtlln} is formulated in the perfect-measurement regime,  the construction of the quenched and annealed laws and the corresponding ergodic arguments rely only on the instrument axioms: each component \(\mcT_{a;\omega}\) is completely positive and trace-non-increasing, and \(\sum_{a\in\mcA}\mcT_{a;\omega}\) is trace-preserving. 
        Thus, the LLN extends verbatim to the imperfect-measurement regime considered here.
  
    \paragraph{LLN and Dynamically Stationary States.}
        We now recall the definition of a dynamically stationary state and the LLN result from \cite{qtlln} below. 
        A random state \(\s:\Omega\to\states\) is called \emph{dynamically stationary} if
        \[
            \Phi_{\theta(\omega)}\!\left(\s(\omega)\right)
            =
            \s(\theta(\omega))
            \qquad
            \text{for \(\pr\)-almost every \(\omega\)}.
        \]
        Here $\Phi$ is the quantum channel associated with the instrument given in \Cref{eq:non_selective_channel}.
        By iteration, this is equivalent to
        \[
            \left(
                \Phi_{\theta^n(\omega)}
                \circ\cdots\circ
                \Phi_{\theta(\omega)}
            \right)\!\left(\s(\omega)\right)
            =
            \s(\theta^n(\omega))
        \]
        for every \(n\in\mbN\) and \(\pr\)-almost every \(\omega\).
        In \cite{qtlln}, the following Law of Large Numbers for \(N_n^{\pmb b}\) was proved under the existence and uniqueness of the dynamically stationary state.

    Under \Cref{assumption1} and the standing measurability assumptions on the channel cocycle, dynamically stationary states always exist; this follows from the vector-valued ergodic theorem of Beck and Schwartz \cite{Beck_1957}.
    The following LLN result is obtained in \cite{qtlln} under the additional assumption that a unique dynamically stationary state exists.
        
    \begin{theorem}[LLN for Outcome Frequencies \cite{qtlln}]
    \label{thm:LLN}
        Assume that $\left(\Omega,\mcF,\pr,\theta\right)$ is an invertible, ergodic, \(\pr\)-preserving dynamical system, and let \(\omega\mapsto\mcV_\omega\) be a random instrument on a finite outcome set \(\mcA\).
        Suppose there exists a \emph{unique} dynamically stationary state \({\s}(\omega)\) for the sequence of measurements \(\left(\mcV_{\theta^{1}(\omega)},\mcV_{\theta^{2}(\omega)},\ldots\right)\).
        Then, for every \(m\in\mbN\) and every \(\pmb b\in\mcA^m\), the following hold:
        \begin{enumerate}
            \item The skew shift \(\tau\) is a \(\overline\mbQ_{\s}\)-measure-preserving ergodic transformation. 
            Furthermore, for any $\tau$-invariant set $\Gamma\in \mcF\otimes\Sigma$ (i.e. $\tau^{-1}(\Gamma) = \Gamma$) we have $\overline\mbQ_{\vartheta}(\Gamma) = \overline\mbQ_{\s}(\Gamma) \in \{0,1\}$, for any initial state $\vartheta:\Omega\to \states$.
            \item For every (random) initial state \(\vartheta\), for \(\pr\)-almost every \(\omega\in\Omega\), and for \(\mbQ_{\vartheta(\omega);\omega}\)-almost every \(\bar a\in\outcomes\),
            \[
                \lim_{n\to\infty} \frac{N_n^{\pmb{b}}(\bar a)}{n}
                \;=\;
                \int_\Omega \mbQ_{{\s}(\omega);\omega}(C_{\pmb{b}})\,\dee\pr(\omega),
            \]
            where $C_{\pmb{b}}$ is as defined in \eqref{eq:cylinder_for_b}.
        \end{enumerate}
    \end{theorem}

    We are now ready to state the assumptions and the main results of this paper. 
    
 
\subsection{Assumptions and Main Results}
\label{section:assumption_and_main}


    We first state the three assumptions used for the universal functional CLT.
    The first concerns the base dynamics, the second controls the temporal dependence of the stationary instrument process, and the third controls the loss of dependence on the initial quantum state under the associated non-selective channel cocycle.

    \begin{assumption}
    \label{assumption1}
        The base system \((\Omega,\mcF,\pr,\theta)\) is invertible, ergodic, and \(\pr\)-preserving.
    \end{assumption}

    Under \Cref{assumption1}, for \(k\in\mbZ\), define the past and future \(\sigma\)-algebras of the instrument process by
    \begin{equation}
    \label{eq:past_and_future_sigma_instruments}
        \mcF_k
        :=
        \sigma\!\left(
            \omega\mapsto\mcV_{\theta^j(\omega)}:\,j\le k
        \right),
        \qquad
        \mcF^k
        :=
        \sigma\!\left(
            \omega\mapsto\mcV_{\theta^j(\omega)}:\,j\ge k
        \right).
    \end{equation}

    For sub-\(\sigma\)-algebras \(\mcG,\mcH\subseteq\mcF\), define
    \begin{equation}
    \label{eq:rho_mixing_cooef}
        \rho_{\pr}(\mcG,\mcH)
        :=
        \sup\left\{
            |\mathrm{Corr}_{\pr}(U,V)|:\,
            U\in L^2(\mcG),\,
            V\in L^2(\mcH),\,
            \mathrm{Var}_{\pr}(U)>0,\,
            \mathrm{Var}_{\pr}(V)>0
        \right\},
    \end{equation}
    and
    \begin{equation}
    \label{eq:alpha_mixing_cooef}
        \alpha_{\pr}(\mcG,\mcH)
        :=
        \sup\left\{
            |\pr(G\cap H)-\pr(G)\pr(H)|:\,
            G\in\mcG,\,
            H\in\mcH
        \right\}.
    \end{equation}
    These sigma-algebras record only the instrument maps. Mixing of this process does not require mixing of $\theta$ on all of $\mcF$; additional environmental variables need not be mixing.

    The corresponding mixing coefficients (collectively called the mixing profile) of the instrument process are
    \begin{equation}
    \label{eq:rho_n__and_alpha_n_disorder}
        \rho_{\pr}(n)
        :=
        \sup_{k\in\mbZ}
        \rho_{\pr}(\mcF_k,\mcF^{k+n}),
        \qquad
        \alpha_{\pr}(n)
        :=
        \sup_{k\in\mbZ}
        \alpha_{\pr}(\mcF_k,\mcF^{k+n}).
    \end{equation}
    They satisfy \(\alpha_{\pr}(n)\le\frac14\rho_{\pr}(n)\) for each $n\in\mbN$, see \cite{bradley2007introduction}.

    \begin{assumption}
    \label{assumption_mixing_base}
        The mixing coefficients satisfy
        \[
            \sum_{n=1}^{\infty}
            \alpha_{\pr}(n)
            <
            \infty.
        \]
    \end{assumption}

    Our final standing assumption concerns annealed loss of memory for the non-selective channel cocycle.
    First, for the measurement sequence \((\mcV_{\theta^k(\omega)})_{k\in \mbN}\), write 
    \[
        \Phi_{k;\omega}:=\sum_{a\in\mcA}\mcT_{a;\theta^k(\omega)},
        \qquad k\ge1,
    \]
    and for \(k,n\ge1\),
    \[
        \Phi^{(n)}_{k;\omega}
        :=
        \Phi_{k+n-1;\omega}\circ\cdots\circ\Phi_{k;\omega},
        \qquad
        \Phi^{(n)}_\omega:=\Phi^{(n)}_{1;\omega}.
    \]

    With $\Phi^{(0)}_\omega=\mathrm{id}$, these products satisfy
    \[
        \Phi^{(n+m)}_\omega
        =\Phi^{(m)}_{\theta^n\omega}\circ\Phi^{(n)}_\omega,
        \qquad n,m\ge0.
    \]
    Thus $(\Phi^{(n)}_\omega)_{n\ge0}$ is a measurable forward cocycle over $\theta$. 
    We call this the \emph{non-selective channel cocycle}, since $\Phi_\omega=\sum_{a\in\mcA}\mcT_{a;\omega}$ describes the evolution when the measurement outcome is ignored.
    For $B\subseteq\mcA$, the operation
    \[
        \mcT_{B;\omega}:=\sum_{a\in B}\mcT_{a;\omega}
    \]
    describes the unnormalized state update associated with observing an outcome in $B$.
    Selecting a particular outcome, or a specified subset of outcomes, gives a \emph{selective} operation; summing over all outcomes recovers the non-selective channel.
    
    \begin{assumption}
    \label{assumption:forgetting}
        There exist a dynamically stationary state \(\s:\Omega\to\states\) and a deterministic sequence \((r_n)_{n\ge1}\subseteq[0,\infty)\) such that
        \[
            \sum_{n=1}^{\infty}
            r_n
            <
            \infty
        \]
        and, for every measurable initial state \(\vartheta:\Omega\to\states\),
        \[
            \beta_n(\vartheta)
            :=
            \mbE_{\pr}\!\left[
                \norm{
                    \Phi^{(n)}_\omega(\vartheta(\omega))
                    -
                    \s(\theta^n\omega)
                }_1
            \right]
            \le
            r_n,
            \qquad n\in\mbN.
        \]
    \end{assumption}

    \begin{remark}
        A convenient sufficient condition for the estimate in \Cref{assumption:forgetting} is
        \[
            \mbE_{\pr}\!\left[
                \sup_{\rho\in\states}
                \norm{
                    \Phi^{(n)}_\omega(\rho)
                    -
                    \s(\theta^n\omega)
                }_1
            \right]
            \le
            r_n,
            \qquad n\in\mbN.
        \]
        The dynamically stationary state in \Cref{assumption:forgetting} is necessarily unique up to \(\pr\)-almost-sure equality.
        Indeed, if \(\widetilde\s\) is another dynamically stationary state, then \(\pr\)-invariance of \(\theta\) gives
        \[
            \mbE_{\pr}\!\left[
                \norm{\widetilde\s(\omega)-\s(\omega)}_1
            \right]
            =
            \mbE_{\pr}\!\left[
                \norm{
                    \widetilde\s(\theta^n\omega)-\s(\theta^n\omega)
                }_1
            \right]
            =
            \beta_n(\widetilde\s)
            \le
            r_n.
        \]
        Since the summability condition implies \(r_n\to0\), it follows that \(\widetilde\s=\s\) \(\pr\)-almost surely.
    \end{remark}
    
    \Cref{assumption1,assumption_mixing_base} concern the classical mechanism driving the instrument sequence: the former describes the environmental dynamics, while the latter controls the temporal dependence of the selected instruments.
    For example, if the stochastic driving process of the instruments is i.i.d., both assumptions hold, with $\alpha_{\pr}(n)=0$ for every $n\ge1$, where we take the base system to be the two-sided path space of a stationary classical stochastic process.
    Both assumptions also hold when the classical driving process is a stationary, irreducible, aperiodic finite-state Markov chain, since the instrument process inherits the exponential mixing of the driving chain.

    The random setting described here also subsumes the deterministic homogeneous setting by taking a one-point base $\Omega = \{\omega\}$ with an atomic probability $\pr = \delta_\omega$ and $\theta$ the identity map. 
    This base automatically satisfies \Cref{assumption1}, and as the instrument process is constant, we also have that $\Cref{assumption_mixing_base}$ holds with $\alpha_{\pr}(n) = 0$ for each $n$. 
    Therefore, one only needs to make the additional assumption in \Cref{assumption:forgetting}.  

    \Cref{assumption:forgetting} concerns the loss of dependence on the initial quantum state under the non-selective channel cocycle and requires additional conditions on the instruments.
    We establish two structural sufficient criteria in \Cref{subsec:esp_forgetting,subsec:doeblin_forgetting}.
    In \Cref{subsec:esp_forgetting}, under \Cref{assumption1}, eventual strict positivity of the non-selective channels yields a unique, strictly positive dynamically stationary state.
    If, in addition, $\rho_{\pr}(n)\to0$, then the expected contraction coefficients decay faster than every prescribed polynomial, implying \Cref{assumption:forgetting}.
    This additional mixing condition holds for both stochastic driving processes described above.
    In \Cref{subsec:doeblin_forgetting}, we give a sufficient criterion for the imperfect-measurement regime, formulated as an augmented block Doeblin condition.
    Under \Cref{assumption1}, this condition yields a unique dynamically stationary state and uniform exponential trace-norm forgetting, and hence implies \Cref{assumption:forgetting}, without requiring eventual strict positivity or $\rho_{\pr}(n)\to0$.

    We also note that the mixing and forgetting conditions in \Cref{assumption_mixing_base,assumption:forgetting} are not necessary for an outcome CLT.
    \Cref{app:periodic-examples} gives periodic instruments whose outcome records satisfy functional CLTs even though the original instrument fails both conditions.

    \paragraph{The CLT under the dynamically stationary state.}
        We first consider the annealed law $\overline\mbQ_{\s}$.
        For the functional limit, $D([0,1],\mbR^q)$ denotes the space of right-continuous paths with left limits, equipped with the Skorokhod $J_1$ topology; see \cite{Billingsley_1999}.
        The Skorokhod $J_1$ topology is characterized by $x_n\to x$ if there exist strictly increasing continuous bijections $\lambda_n:[0,1]\to[0,1]$ such that
        \[
             \sup_{0\le t\le1}|\lambda_n(t)-t|\longrightarrow0,
             \qquad
             \sup_{0\le t\le1}\unorm{x_n(\lambda_n(t))-x(t)}\longrightarrow0.
        \]
        For deterministic paths with a continuous limit $x$, convergence in the $J_1$ topology is equivalent to uniform convergence on $[0,1]$.
        The theorem below asserts convergence in distribution in this path space, with a limiting process $\Sigma_f^{1/2}W$ whose paths are continuous almost surely.

        Above and throughout, the unsubscripted norm $\unorm{\,\cdot\,}$ on finite-dimensional real vector spaces denotes the Euclidean norm. 

        \begin{restatable}[Annealed functional CLT for \(\s\)]{thm}{annealedclt}
        \label{thm:aclt}
            Suppose \Cref{assumption1,assumption_mixing_base,assumption:forgetting} hold. 
            Let $\s$ be the unique dynamically stationary state and $(r_n)$ be as in \Cref{assumption:forgetting}.
            Let $q,\ell\in\mbN$ and let $f:\mcA^\ell\to\mbR^q$ be bounded. 
            Take
            \[
                X_j=f(A_{j:j+\ell-1})-\mbE_{\overline\mbQ_{\s}}f(A_{1:\ell}).
            \]
            Then the covariance series
            \begin{equation}
            \label{eq:universal-covariance}
                \Sigma_f=
                    \mbE_{\overline\mbQ_{\s}}[X_1X_1^{\mathsf T}]
                    +\sum_{h\ge1}\left(
                         \mbE_{\overline\mbQ_{\s}}[X_1X_{1+h}^{\mathsf T}]
                         +
                         \mbE_{\overline\mbQ_{\s}}[X_{1+h}X_1^{\mathsf T}]
                    \right)
            \end{equation}
            converges absolutely and defines a positive semidefinite matrix. 
            Furthermore,
            \[
                \left(
                    \frac1{\sqrt n}
                    \sum_{j=1}^{\lfloor nt\rfloor}X_j
                \right)_{0\le t\le1}
                \underset{\overline\mbQ_{\s}}
                {\overset{\mathrm d}{\longrightarrow}}
                \left(\Sigma_f^{1/2}W(t)\right)_{0\le t\le1}
            \]
            in $D([0,1],\mbR^q)$ with the $J_1$ topology,  where $W$ is standard $q$-dimensional Brownian motion and an empty sum is understood to be zero.
            In addition,
            \[
                \frac1n
                \sum_{i,j=1}^n
                \mbE_{\overline\mbQ_{\s}}
                    [X_iX_j^{\mathsf T}]
                \longrightarrow\Sigma_f,
                \qquad
                \frac1n
                \Cov_{\overline\mbQ_{\s}}
                    \left(\sum_{j=1}^nX_j\right)
                \longrightarrow\Sigma_f.
            \]
        \end{restatable}
        
    \begin{remark}
        In the theorem above, if $\Sigma_f$ is positive definite, the Gaussian limit is nondegenerate. Singular covariance matrices are also allowed. 
        In particular, if $\Sigma_f=0$, the limiting process is identically zero, and 
        \[
            \frac1{\sqrt n}\sum_{j=1}^nX_j
            \longrightarrow0
            \qquad\text{in }L^2(\overline\mbQ_{\s}). 
        \]
    \end{remark}

    \paragraph{Universal CLT.}
        To extend \Cref{thm:aclt} beyond the dynamically stationary initial state, we compare the corresponding annealed outcome laws.
        Summing over an initial segment of outcomes replaces that segment by the non-selective channel product. 
        Its forgetting estimate therefore compares the remaining outcome law with the stationary law.

        \begin{restatable}[Universal functional CLT]{thm}{univannealedclt}
        \label{thm:universal-functional}
            Suppose \Cref{assumption1,assumption_mixing_base,assumption:forgetting} hold, and let $q,\ell,f,X_j$ and $\Sigma_f$ be as in \Cref{thm:aclt}.
            Then, for every measurable random initial state $\vartheta:\Omega\to\states$,
            \[
                \left(
                    \frac1{\sqrt n}
                    \sum_{j=1}^{\lfloor nt\rfloor}X_j
                \right)_{0\le t\le1}
                \underset{\overline\mbQ_\vartheta} {\overset{\mathrm d}{\longrightarrow}}
                \left(\Sigma_f^{1/2}W(t)\right)_{0\le t\le1}
            \]
            in $D([0,1],\mbR^q)$ with the $J_1$ topology,  where $W$ is standard $q$-dimensional Brownian motion and an empty sum is understood to be zero.
            In addition,
            \begin{equation}
            \label{second_moment_universal_clt}
                \frac1n
                \sum_{i,j=1}^n
                \mbE_{\overline\mbQ_\vartheta}
                    [X_iX_j^{\mathsf T}]
                \longrightarrow\Sigma_f,
                \qquad
                \frac1n
                \Cov_{\overline\mbQ_\vartheta}
                    \left(\sum_{j=1}^nX_j\right)
                \longrightarrow\Sigma_f.
            \end{equation}
            The centering and asymptotic covariance are the same as in \Cref{thm:aclt} and do not depend on the initial state $\vartheta$.
        \end{restatable}

    Under this theorem, taking $q=1$, $\ell=m$ and $f(u)=\mathbf1_{\{u=\pmb b\}}$ in \Cref{thm:universal-functional}, and evaluating the functional limit at $t=1$, gives the following central limit theorem for pattern frequencies.
    The zero-variance $L^2$ conclusion follows from the second-moment convergence in the same theorem.
    
    \begin{restatable}[Universal annealed CLT for pattern counts]{bigcor}{universalcltA}
    \label{thm:universal-clt-A}
        Suppose \Cref{assumption1,assumption_mixing_base,assumption:forgetting} hold, and let \(\s\) be the unique dynamically stationary state. Fix \(m\in\mbN\) and \(\pmb b\in\mcA^m\), and define
        \[
            \mu_{\pmb b}
            :=
            \overline{\mbQ}_{\s}(C_{\pmb b})
            =
            \int_\Omega \mbQ_{\s(\omega);\omega}(C_{\pmb b})\,\dee\pr(\omega).
        \]
        Then the covariance series
        \[
            \Sigma_{\pmb b}^2
            :=
            \mbE_{\overline\mbQ_{\s}}\!\left[(\mathbf 1_{C_{\pmb b}}-\mu_{\pmb b})^2\right]
            +
            2\sum_{k=1}^\infty
            \mbE_{\overline\mbQ_{\s}}\!\left[
                (\mathbf 1_{C_{\pmb b}}-\mu_{\pmb b})
                ( \mathbf 1_{C_{\pmb b}}\circ\varsigma^{k} -\mu_{\pmb b})
            \right]
        \]
        converges absolutely and defines a finite asymptotic variance \(\Sigma_{\pmb b}^2\in[0,\infty)\). 
        Moreover, for every measurable random initial state \(\vartheta:\Omega\to\states\),
        \[
            \frac{1}{\sqrt n}\sum_{k=1}^n
            \left(\mathbf 1_{C_{\pmb b}}\circ\varsigma^{k-1} -\mu_{\pmb b}\right)
            \underset{\overline\mbQ_{\vartheta}}{\overset{\mathrm d}{\longrightarrow}}
            \mcN(0,\Sigma_{\pmb b}^2),
            \qquad n\to\infty.
        \]
        If \(\Sigma_{\pmb b}^2=0\), then the normalized sums converge to \(0\) in \(L^2(\overline\mbQ_{\vartheta})\), hence also in probability.
    \end{restatable}


\subsection{Related Literature}
\label{section:lit}


    The literature on quantum trajectories naturally separates into a continuous-time theory and a discrete-time repeated-measurement theory. 
    For continuous-time quantum trajectories, one may consult the general stochastic Schr\"odinger equation and continuous-measurement framework developed in \cite{Dav76,Gis84,Car93,Hol01,BGM04,BG09,WM09}. 
    Discrete-time repeated interactions and repeated measurements provide a mathematically precise route to these continuous-time models; see, for example, \cite{AP06,Pel08,BBB13,NP09}. 
    Within the discrete-time finite-outcome setting, early ergodic results for measurement records and trajectories, that is, for the posterior-state process, appear in \cite{kummerer2003ergodic,KM04}. 
    
    For homogeneous repeated measurements, fluctuation and invariant-measure results were developed further in \cite{Attal_2014,BFPP19,BFP23,BPS24,Benoist_2025}. 
    Closely related path-space questions for repeated measurements were also studied from the viewpoint of entropy production in \cite{BJPP18,Benoist_2021}, while repeated quantum non-demolition measurements and their continuous-time limits were analyzed in \cite{BBB13}. 
    More recently, imperfect-measurement versions of the trajectory theory were developed in \cite{tristant_2025}. 
    
    The present work belongs to the disordered branch of the subject. 
    Section~1.2 of \cite{qtlln} places this model within the broader literature on disordered open quantum dynamical systems and repeated interactions, including \cite{BJM08,PP09,Bru+09,BJM10,NP12,Bur+13,BJM14,MS21,BJP22,MS22,PS23,Raqu_pas_2025,ES24,traj}. 
    Among these works, \cite{traj} introduced the specific model of disordered quantum trajectories considered in \cite{qtlln}. 
    Our contribution differs in emphasis from the invariant-measure theory of the posterior-state process: we focus instead on fluctuation theory for finite block counts in the measurement record under disorder.
    
    More precisely, our theorem complements the law of large numbers established in \cite{qtlln}. It may be viewed as a disordered analogue, at the level of the measurement record, of the central limit theorems obtained in the homogeneous setting in \cite{Attal_2014}.
    
    For comparison, one may also approach limit theorems in random environments using techniques from classical random dynamical systems, such as those in \cite{Kifer_1998,HAFOUTA_2019,Hafouta_2023}.
    However, the main difficulty is translating the assumptions in these results into assumptions based on the classical stochastic driving process of the instruments and their structural properties. 


\section{Examples}
\label{section:examples}


We now give examples that verify the assumption above.
The i.i.d.\ and finite-state Markov driving processes discussed above satisfy \Cref{assumption1,assumption_mixing_base}.
We therefore focus on structural properties of the instruments that ensure \Cref{assumption:forgetting}, and throughout this section we assume that \Cref{assumption1,assumption_mixing_base} are satisfied. 

We write $\mcR_\tau(X)=\operatorname{Tr}(X)\tau$ for replacement by a state $\tau\in\states$.
This is called a \emph{replacement channel}.
For a channel $\Psi$, define its trace-norm contraction coefficient by
\begin{equation}
\label{eq:kappa}
     \kappa(\Psi)
     =
     \sup_{\substack{D=D^*,\str D=0\\D\ne0}}
     \frac{\unorm{\Psi(D)}_1}{\unorm D_1}
     =
     \frac12\sup_{\rho,\eta\in\states}\unorm{\Psi(\rho)-\Psi(\eta)}_1.
\end{equation}
The state-diameter identity and submultiplicativity follow from \cite[Eq.~(34)]{HiaiRuskai2016} and \cite{Replacement}: $0\le\kappa\le1$ and $\kappa(\Psi\Gamma)\le\kappa(\Psi)\kappa(\Gamma)$.
The coefficient is measurable, since in finite dimension it is a supremum over a fixed compact unit sphere, equivalently over a countable dense subset.

\begin{prop}
\label{prop:unified:common}
    Suppose \Cref{assumption1} holds and put $e_n=\mbE\kappa(\Phi^{(n)}_\omega)$, where $\kappa$ is defined in \eqref{eq:kappa}.
    If $e_n\to0$, there is a unique measurable dynamically stationary state $s$, and for every measurable initial state $\vartheta$,
    \begin{equation}
    \label{eq:unified:forget}
     \mbE_{\pr}
        \norm{\Phi^{(n)}_\omega\vartheta(\omega)-s(\theta^n\omega)}_1
        \le r_n:=2e_n.
    \end{equation}
    In particular, $\sum_{n\ge1}e_n<\infty$ verifies \Cref{assumption:forgetting}. If the instrument also satisfies \Cref{assumption_mixing_base}, then \Cref{thm:universal-functional} applies.
\end{prop}
    \begin{proof}
        Work on a common invariant set on which all channel identities hold. 
        The images
        \[
         K_n(\omega):=\Phi^{(n)}_{\theta^{-n}\omega}(\states)
         =\Phi_\omega\Phi_{\theta^{-1}\omega}\cdots
           \Phi_{\theta^{-(n-1)}\omega}(\states)
        \]
        are nonempty nested compact sets, since each $\phi_{\theta^k(\omega)}$ preserves quantum states. 
        Their trace-norm diameters, $2\kappa(\Phi^{(n)}_{\theta^{-n}\omega})$, are non-increasing in $n$, due to the submultiplicative property of $\kappa$ and as $0\le \kappa\le 1$. 
        But $\mbE_{\pr}[\operatorname{diam}K_n]=2e_n\to0$, so the limiting diameter is zero $\mbP$-almost surely.
        There, we used that $\theta$ is $\mbP$-invariant to get that $\mbE_{\pr}[\operatorname{diam}K_n]=2e_n$. 
        Thus for $\mbP$-almost every $\omega$ there is some $s(\omega)\in\states$ such that 
        \[
            \bigcap_nK_n(\omega)=\{s(\omega)\}.
        \] 
        For any fixed state $\rho_*$, the sequence $\Phi^{(n)}_{\theta^{-n}\omega}(\rho_*)$ converges to $s(\omega)$, proving measurability of the state $\omega\mapsto s(\omega)$. 
        The identity $K_{n+1}(\theta\omega)=\Phi_{\theta\omega}K_n(\omega)$, or the corresponding identity for these convergent sequences, gives $s(\theta\omega)=\Phi_{\theta\omega}s(\omega)$. Every other dynamically stationary state belongs to all the sets $K_n$ almost surely and therefore equals $s$. Finally, stationarity and \eqref{eq:kappa} give the simultaneous bound
        \[
         \sup_{\rho\in\states}\|\Phi^{(n)}_\omega(\rho)-s(\theta^n\omega)\|_1
         \le2\kappa(\Phi^{(n)}_\omega).
        \]
        Taking expectations proves \eqref{eq:unified:forget}.
    \end{proof}
    
The following lemma identifies projections that preserve the outcome laws, allowing the preceding contraction criterion to be applied to the projected instrument.

\begin{lemma}
\label{lem:unified:reduction}
    Let $\mcP$ be a fixed CPTP projection such that
    \begin{equation}
    \label{eq:unified:reduction}
         \mcP\mcT_{a;\omega}
         =\mcP\mcT_{a;\omega}\mcP
    \end{equation}
    for $\pr$-almost every $\omega$ and every outcome $a\in\mcA$.
    Then
    \[
        \widehat{\mcT}_{a;\omega}
        :=\mcP\mcT_{a;\omega}
    \]
    defines an instrument,  which we call the \emph{reduced instrument} associated with $\mcP$.
    For a random initial state $\vartheta:\Omega\to \states$ denote the quenched outcome laws by $\widehat{\mbQ}_{\rho;\omega}$.
    For every measurable initial state $\vartheta:\Omega\to\states$, 
    \begin{equation}
    \label{eq:unified:reduced-law}
         \widehat{\mbQ}_{(\mcP\vartheta)(\omega);\omega}
         =
         \mbQ_{\vartheta(\omega);\omega}
         \qquad\text{for $\pr$-almost every $\omega$.}
    \end{equation}
\end{lemma}

    \begin{proof}
        Intersect the full-measure sets on which the instrument and reduction identities hold over all integer shifts. 
        Work on the resulting invariant set of full $\pr$-measure.
        Each $\widehat{\mcT}_{a;\omega}$ is completely positive,
        and
        \[
         \sum_a\widehat{\mcT}_{a;\omega}
         =\mcP\Phi_\omega
        \]
        is trace preserving. Hence $\widehat{\mcT}$ is an instrument.
        Since $\mcP$ is completely positive and trace preserving, $(\mcP\vartheta)(\omega)\in\states$.
        Moreover, $\mcP\vartheta$ is measurable because $\mcP$ is a fixed continuous linear map.
        Fix $n\ge1$ and outcomes $a_1,\ldots,a_n$, and write
        \[
             \mcT_j:=\mcT_{a_j;\theta^j\omega},
             \qquad
             \widehat{\mcT}_j:=\mcP\mcT_j.
        \]
        Repeated use of \eqref{eq:unified:reduction} gives
        \[
             \widehat{\mcT}_n\cdots\widehat{\mcT}_1
             \mcP
             =
             \mcP\mcT_n\cdots\mcT_1.
        \]
        Therefore, using preservation of trace by $\mcP$,
        \begin{align*}
         \widehat{\mbQ}_{(\mcP\vartheta)(\omega);\omega}
           (A_1=a_1,\ldots,A_n=a_n)
         &\quad=
         \operatorname{tr}\!\left[
           \widehat{\mcT}_n\cdots\widehat{\mcT}_1
           \left(\mcP(\vartheta(\omega))\right)
         \right]
         \\
         &\quad=
         \operatorname{tr}\!\left[
           \mcP\mcT_n\cdots\mcT_1
           \left(\vartheta(\omega)\right)
         \right]
         \\
         &\quad=
         \operatorname{tr}\!\left[
           \mcT_n\cdots\mcT_1
           \left(\vartheta(\omega)\right)
         \right]
         \\
         &\quad=
         \mbQ_{\vartheta(\omega);\omega}
           (A_1=a_1,\ldots,A_n=a_n).
        \end{align*}
        Equality on all outcome cylinders proves \eqref{eq:unified:reduced-law}. Integrating against $\pr(d\omega)$ gives equality of the corresponding annealed laws.
        The final assertion follows by applying \Cref{prop:unified:common} to $\widehat{\mcT}$.
    \end{proof}

\begin{remark}
    For the reduced instrument above, the past and future sigma-algebras in \eqref{eq:past_and_future_sigma_instruments} are generated by
    \[
        \widehat{\mcV}_{\theta^j\omega}
        =
        \left(\widehat{\mcT}_{a;\theta^j\omega}\right)_{a\in\mcA}.
    \]
    These are sub-sigma-algebras of the corresponding sigma-algebras for the original instrument. 
    Hence, \Cref{assumption_mixing_base} for the original instrument implies it for the reduction, although the reduction may satisfy this assumption even when the original instrument does not.
    Under \Cref{assumption1}, the functional CLT for the original outcome process can therefore be obtained by verifying \Cref{assumption_mixing_base,assumption:forgetting} for the reduced instrument, with \Cref{prop:unified:common} providing a sufficient criterion for the latter.
    The transfer to the original outcome process follows from \eqref{eq:unified:reduced-law}.

    Forgetting and uniqueness of the dynamically stationary state for the reduced channel cocycle need not imply either property for the original channel cocycle. If an original dynamically stationary state $\s$ exists, then $\mcP \s$ is dynamically stationary for the reduced cocycle.
    Whenever the reduced dynamically stationary state $\widehat \s$  is unique, we therefore have $\mcP \s=\widehat{\s}$ almost surely, and the corresponding stationary quenched and annealed outcome laws agree.
\end{remark}

With the result above in place, we now give several classes of examples below. 


\subsection{Exact classical reduction of the outcomes}
\label{subsec:unified:classical}


    In this subsection, we consider instruments for which, at each fixed environment, the probability of every outcome word depends on the initial quantum state only through its populations in a specified orthogonal decomposition of the Hilbert space. 
    We first state the algebraic condition that gives this property and identify the instrument governing those populations.
    
    We then give three types of examples. Weighted permutation operators include basis transitions, cyclic and group walks, and amplitude damping.
    Imperfect classical readout allows several Kraus operators per recorded outcome. 
    State-independent emissions give the one-state case $m=1$, including the synchronizing and countdown models.
    Along the way, we distinguish population forgetting from full quantum-state forgetting and examine their relationship with eventual strict positivity and augmented Doeblin minorization.

    \paragraph{The reduced instrument.}
        Let $\mcH=\mbC^d$, and fix nonzero orthogonal projections  $P_1,\ldots,P_m$ on $\mcH$ such that $P_iP_j=0$ for $i\ne j$ and $\sum_{i=1}^mP_i=I$.
        Define
        \[
         R:\mcB(\mcH)\longrightarrow\mbC^m,
         \qquad R(X)=\left(\operatorname{Tr}(P_iX)\right)_{i=1}^m.
        \]
        Let $\Delta_m$ denote the set of probability vectors in $\mbR^m$:
        \[
         \Delta_m=\left\{x\in[0,1]^m:\sum_{i=1}^m x_i=1\right\}.
        \]
        Thus $R(\rho)\in\Delta_m$ for every quantum state $\rho$; its $i$-th entry is the population of the subspace $P_i\mcH$.
        Write $\mathbf1_m$ for the column vector of $m$ ones and write $u_i$ for the $i$-th standard coordinate vector in $\mbR^m$.

        Fix a finite outcome alphabet $\mcA$.
        For each outcome $a\in\mcA$, suppose that there is a measurable matrix $B_{a;\omega}\in[0,\infty)^{m\times m}$ such that, almost surely,
        \begin{equation}
        \label{eq:unified:intertwining}
             R\mcT_{a;\omega}(X)=B_{a;\omega}R(X)
             \quad\text{for every }X\in\mcB(\mcH),
             \qquad T_\omega:=\sum_{a\in\mcA}B_{a;\omega}.
        \end{equation}
        The vector $B_{a;\omega}x$ is the unnormalized population vector after outcome $a$ when the incoming population vector is $x\in\Delta_m$.
        The probability of this outcome is $\mathbf1_m^{\mathsf T}B_{a;\omega}x$; if it is positive, the normalized population vector is 
        \[
            \dfrac{B_{a;\omega}x}{(\mathbf1_m^{\mathsf T}B_{a;\omega}x)}.
        \]

        The matrix $T_\omega$ is column stochastic. Indeed, the states $\zeta_i:=P_i/\operatorname{Tr}P_i$ satisfy $R(\zeta_i)=u_i$, so trace preservation of $\Phi_\omega=\sum_a\mcT_{a;\omega}$ gives $\mathbf1_m^{\mathsf T}T_\omega u_i=1$ for every $i$.
        Consequently, $x\mapsto T_\omega x$ describes population evolution when the outcome is not recorded.

        Iterating \eqref{eq:unified:intertwining} gives, for every state $\rho$ and every word $(a_1,\ldots,a_n)\in\mcA^n$,
        \begin{equation}
        \label{eq:unified:classical-law}
             \mbQ_{\rho;\omega}(A_1=a_1,\ldots,A_n=a_n)
             =\mathbf1_m^{\mathsf T}B_{a_n;\theta^n\omega}\cdots
               B_{a_1;\theta\omega}R(\rho).
        \end{equation}
        All these identities hold on a common invariant full-measure set.
        In particular, two initial states with the same population vector $R(\rho)$ have the same outcome law in the same environment.
        The matrices along the environmental path still enter that law.

        To fit the quantum-instrument formalism to this finite classical description, define $J:\mbC^m\longrightarrow\mcB(\mcH)$ by
        \[
            Jx=\sum_{i=1}^m x_i\zeta_i,
        \]
        and take 
        \[
            \mcP:=JR,
                \qquad\text{ and }\qquad 
            \widehat{\mcT}_{a;\omega}:=JB_{a;\omega}R.
        \]
        Since $RJ$ is the identity, $\mcP$ is a CPTP projection.
        Each reduced branch is completely positive because
        \[
         \widehat{\mcT}_{a;\omega}(X)
         =\sum_{k,i=1}^m(B_{a;\omega})_{ki}
               \operatorname{Tr}(P_iX)\zeta_k.
        \]
        Column stochasticity of $T_\omega$ makes their sum trace preserving.
        Moreover, $\mcP\mcT_{a;\omega}
        =\widehat{\mcT}_{a;\omega}
        =\mcP\mcT_{a;\omega}\mcP$, so
        \Cref{lem:unified:reduction} applies.

    \paragraph{Forgetting of the population vector.}
        For a column-stochastic $m\times m$ matrix $T$, define its Dobrushin contraction coefficient by
        \[
             \delta(T)
             :=\frac12\max_{1\le i,j\le m}
                \sum_{k=1}^m|T_{ki}-T_{kj}|
             =1-\min_{1\le i,j\le m}\sum_{k=1}^m\min(T_{ki},T_{kj}).
        \]
        It satisfies $\|Tx-Ty\|_{\ell^1}\le\delta(T)\|x-y\|_{\ell^1}$ for $x,y\in\Delta_m$, where $\|v\|_{\ell^1}=\sum_i|v_i|$.
        For each $n\ge 1$, set
        \[
             T_\omega^{(n)}:=T_{\theta^n\omega}\cdots T_{\theta\omega},
             \qquad e_n^{\mathrm{red}}:=\mbE_{\pr}\delta(T_\omega^{(n)}),
             \qquad \text{ and }\qquad  r_n^{\mathrm{red}}:=2e_n^{\mathrm{red}},\qquad n\ge1.
        \]
        The superscript ``$\mathrm{red}$'' indicates that these quantities concern the reduced instrument. 
        Its $n$-step non-selective channel is $JT_\omega^{(n)}R$. 
        Since $R$ contracts the trace norm on Hermitian matrices, $\|Jv\|_1=\|v\|_{\ell^1}$ for real $v$, and  $R(\zeta_i)=u_i$, its trace-norm contraction coefficient is exactly $\delta(T_\omega^{(n)})$.
        
        If $e_n^{\mathrm{red}}\to0$, \Cref{prop:unified:common} therefore gives a unique measurable stationary population vector $\pi:\Omega\to\Delta_m$, satisfying $\pi(\theta\omega)=T_{\theta\omega}\pi(\omega)$, and
        \begin{equation}
        \label{eq:examples:population-forgetting}
             \mbE_{\pr}
             \left\|T_\omega^{(n)}x(\omega)-\pi(\theta^n\omega)\right\|_{\ell^1}
             \le r_n^{\mathrm{red}}
        \end{equation}
        for every measurable $x:\Omega\to\Delta_m$.
        The reduced stationary quantum state is $\widehat s=J\pi$.
        Thus $\sum_{n\ge1}e_n^{\mathrm{red}}<\infty$ verifies \Cref{assumption:forgetting} for the reduced instrument, with the sequence $r_n^{\mathrm{red}}$ just defined.
        
        For example, suppose that there are $L\ge1$ and $\varepsilon\in(0,1]$ such that, almost surely,
        \[
         \min_{1\le i,j\le m}\sum_{k=1}^m
               \min\left((T_\omega^{(L)})_{ki},(T_\omega^{(L)})_{kj}\right)
         \ge\varepsilon.
        \]
        This is a uniform lower bound on the overlap of every pair of columns in an $L$-step product. 
        Submultiplicativity yields
        \begin{equation}
        \label{eq:unified:classical-rate}
             r_n^{\mathrm{red}}
             \le2(1-\varepsilon)^{\lfloor n/L\rfloor}.
        \end{equation}

    For the outcome CLT, \Cref{assumption_mixing_base} is checked for the process $((B_{a;\theta^j\omega})_{a\in\mcA})_{j\in\mbZ}$.
    This is exactly the mixing requirement for the reduced quantum instrument.
    Together with summable reduced forgetting, it gives \Cref{thm:universal-functional} for the original outcomes from every measurable initial state $\vartheta$, by substituting $\rho=\vartheta(\omega)$ in \eqref{eq:unified:classical-law}. 

   Throughout the examples and special cases in this subsection, all applications of \Cref{thm:universal-functional} assume \Cref{assumption_mixing_base} for the reduced instrument process, in addition to the displayed forgetting bounds.
    This mixing requirement also applies when $r_n^{\mathrm{red}}=0$ for every $n\ge1$.
    Uniqueness of the stationary population vector $\pi$ concerns the reduced instrument and does not imply uniqueness of a stationary state for the original quantum instrument.

    \begin{example}[Weighted permutation operators.]
    \label{eg:weighted-permutations}
        Fix an orthonormal basis $(e_i)_{i\in [d]}$ of $\mbC^d$, where $[d]:=\{1,\ldots,d\}$ and $d\ge2$.
        For each outcome $a\in\mcA$, let $s_a:[d]\to [d]$ be a fixed permutation.
        Let $c_{a,i}:\Omega\to\mbC$ be measurable and satisfy $\sum_{a\in\mcA}|c_{a,i;\omega}|^2=1$ for each $i$, almost surely.
        Define
        \begin{equation}
        \label{eq:unified:permutations}
             V_{a;\omega}:=\sum_{i=1}^d c_{a,i;\omega}
                   \ket{e_{s_a(i)}}\bra{e_i},
             \qquad \mcT_{a;\omega}(X):=V_{a;\omega}XV_{a;\omega}^*.
        \end{equation}
        These maps form a perfect instrument, since
        \[
         \sum_a V_{a;\omega}^*V_{a;\omega}
         =\sum_{i=1}^d\left(\sum_a|c_{a,i;\omega}|^2\right)
               \ket{e_i}\bra{e_i}=I.
        \]
        Take $m=d$ and $P_i=\ket{e_i}\bra{e_i}$ in the reduction above, so $R(X)=(\inner{e_i}{Xe_i})_{i=1}^d$.
        Because $s_a$ is injective, each output diagonal entry receives a contribution from at most one input diagonal entry. 
        Therefore \eqref{eq:unified:intertwining} holds with
        \[
             (B_{a;\omega})_{ki}
             =|c_{a,i;\omega}|^2\mathbf1_{\{k=s_a(i)\}},
             \qquad i,k\in [d].
        \]
        Thus the phases of the coefficients and the off-diagonal entries of the initial state do not enter any outcome-word probability.
        The following special cases are concrete choices of these operators.
    \end{example}
    
        \paragraph{Special case: Basis-transition instruments}
        \label{par:basis-transition-instruments}
            Let $Q_\omega=(q_{ki;\omega})_{k,i=1}^d$ be a measurable column-stochastic matrix. 
            Use the outcome set $\mcA=[d]\times[d]$ and the operators
            \[
                 V_{(k,i);\omega}
                 =\sqrt{q_{ki;\omega}}\,\ket{e_k}\bra{e_i}.
            \]
            These operators have the weighted-permutation form \eqref{eq:unified:permutations}. 
            To see this, fix an outcome $(k,i)$. If $i\ne k$, choose $s_{(k,i)}:[d]\to[d]$ to be the transposition exchanging $i$ and $k$:
            \[
                 s_{(k,i)}(j)=
                     \begin{cases}
                      k, & j=i,\\
                      i, & j=k,\\
                      j, & j\notin\{i,k\}.
                     \end{cases}
            \]
            If $i=k$, choose the identity permutation.
            In either case, $s_{(k,i)}$ is a bijection satisfying   $s_{(k,i)}(i)=k$.
            For this outcome, choose the coefficients
            \[
                 c_{(k,i),j;\omega}
                 =
                 \begin{cases}
                  \sqrt{q_{ki;\omega}}, & j=i,\\
                  0, & j\ne i.
                 \end{cases}
            \]
            Thus only the term indexed by $j=i$ remains in the weighted-permutation sum:
            \[
                 \sum_{j=1}^d c_{(k,i),j;\omega}
                 \ket{e_{s_{(k,i)}(j)}}\bra{e_j}
                 =
                 \sqrt{q_{ki;\omega}}\,\ket{e_k}\bra{e_i}
                 =
                 V_{(k,i);\omega}.
            \]
            In particular, $V_{(k,i);\omega}$ sends $e_i$ to $\sqrt{q_{ki;\omega}}\,e_k$ and sends every other basis vector to zero.
    
            Let $u_i\in\mbR^d$ denote the $i$-th standard population-coordinate vector, while $e_i\in\mbC^d$ denotes the corresponding Hilbert-space basis vector.
            The reduced branch matrices and their sum are
            \[
                 B_{(k,i);\omega}
                 =q_{ki;\omega}u_ku_i^{\mathsf T},
                 \qquad
                 T_\omega
                 =\sum_{k,i=1}^d B_{(k,i);\omega}
                 =Q_\omega.
            \]
            In this example, each original branch already satisfies
            \[
                 \mcT_{(k,i);\omega}
                 =JB_{(k,i);\omega}R
                 =\widehat{\mcT}_{(k,i);\omega},
            \]
            so the reduction leaves the instrument unchanged.
            Consequently,
            \[
                 \Phi_\omega=JQ_\omega R,
                 \qquad
                 \Phi_\omega^{(n)}=JT_\omega^{(n)}R,
                 \qquad
                 \kappa(\Phi_\omega^{(n)})
                 =\delta(T_\omega^{(n)}),
                 \qquad n\ge1.
            \]
            Thus the following bounds concern full-state forgetting for the original instrument.
            We write
            \[
                 r_n:=2\mbE_{\pr}\delta(T_\omega^{(n)}).
            \]
            Whenever this sequence is summable, \Cref{prop:unified:common} verifies \Cref{assumption:forgetting} for the original instrument.
            
            For a measurable map $h:\Omega\to[0,1]$ and a fixed  probability vector $v\in\Delta_d$, the choice
            \[
                 Q_\omega=(1-h_\omega)I_d
                          +h_\omega v\mathbf1_d^{\mathsf T}
            \]
            has $\delta(Q_\omega)=1-h_\omega$.

            Within this family, we consider uniform noise, absorption, and random replacement:
            Uniform noise corresponds to $v=\mathbf1_d/d$, and absorption at the basis state $e_1$ corresponds to $v=u_1$.
            A bound $h_\omega\ge h_*>0$ gives
            \[
                r_n\le2(1-h_*)^n.
            \]
            If $(h_{\theta^j\omega})_{j\in\mbZ}$ is i.i.d.  and $\mbE_{\pr}h>0$, independence instead gives
            \[
             r_n\le2(1-\mbE_{\pr}h)^n.
            \]
            For random replacement, take a measurable $\ell:\Omega\to[d]$ and set \(Q_\omega=u_{\ell(\omega)}\mathbf1_d^{\mathsf T}\).
            Then $r_n=0$ for every $n\ge1$.
    
        \paragraph{Special case: Cyclic keep--switch instruments.}
        \label{par:cyclic-keep-switch}
            This construction is based on the two-state keep--switch instrument studied in \cite[Section~2.1.2]{Benoist_2021}.
            Here we allow environment-dependent parameters and extend the construction to a cyclic basis index set of size $d\ge2$.
            Take $\mcA=\{K,S\}$ and define the permutations
            \[
                 s_K(i)=i,
                 \qquad
                 s_S(i)=
                 \begin{cases}
                  i+1, & i<d,\\
                  1, & i=d.
                 \end{cases}
            \]
            Thus $K$ keeps the current basis index, while $S$ advances it by one position around the cycle.
            Choose measurable $p_i:\Omega\to[0,1]$ and coefficients satisfying
            \[
                 |c_{K,i;\omega}|^2=p_{i;\omega},
                 \qquad
                 |c_{S,i;\omega}|^2=1-p_{i;\omega}.
            \]
            Suppose that $p_{i;\omega}\in[\eta,1-\eta]$ for every $i$, for $\pr$-almost every $\omega$, where $0<\eta\le1/2$.
            Then every entry of every $(d-1)$-step population transition product is at least $\eta^{d-1}$.
            Indeed, to reach any target from any starting index, make the required number of cyclic moves and fill the remaining steps with stays. 
            Each step has probability at least $\eta$. 
            Thus the column-overlap bound holds with
            \[
                 L=d-1,
                 \qquad
                 \varepsilon=d\eta^{d-1},
            \]
            and \eqref{eq:unified:classical-rate} applies.
    
            For $d=2$, the total population matrix and its contraction
            coefficient are
            \[
                 T_\omega=
                 \begin{pmatrix}
                  p_{1;\omega}&1-p_{2;\omega}\\
                  1-p_{1;\omega}&p_{2;\omega}
                 \end{pmatrix},
                 \qquad
                 \delta(T_\omega)=|p_{1;\omega}+p_{2;\omega}-1|.
            \]
            The parameters $p_1$ and $p_2$ may be chosen separately.
            If $p_{2;\omega}=1-p_{1;\omega}$ almost surely, both columns coincide and population forgetting occurs in one step.
            More generally, if the parameter pairs $(p_{1;\theta^j\omega},p_{2;\theta^j\omega})$ are i.i.d. over time and
            \[
                 q:=\mbE_{\pr}|p_1+p_2-1|<1,
            \]
            then
            \[
             r_n^{\mathrm{red}}\le2q^n.
            \]
            This does not require independence of $p_1$ and $p_2$ at a fixed time.
    
        \paragraph{Special case: Finite-group instruments.}
        \label{par:finite-group-instruments}
            A further specialization of \eqref{eq:unified:permutations} is obtained from a finite group $G$ of size $d\ge2$.
            Index an orthonormal basis by $g\in G$, and choose an outcome set $\mcA=S\subseteq G$ containing the identity and generating $G$ as a semigroup.
            For each $a\in S$, define
            \[
                s_a(g)=ag.
            \]
            This is a permutation of $G$, with inverse $g\mapsto a^{-1}g$.
            Choose measurable coefficients satisfying, for $\pr$-almost
            every $\omega$,
            \[
                 |c_{a,g;\omega}|^2=p_{g;\omega}^a\ge\eta>0,
                 \qquad
                 \sum_{a\in S}p_{g;\omega}^a=1
                 \quad\text{for every }g\in G.
            \]
            The corresponding operators are
            \[
             V_{a;\omega}
             =\sum_{g\in G}c_{a,g;\omega}
                   \ket{e_{ag}}\bra{e_g}.
            \]
            
            Choose $L\ge1$ so that every element of $G$ can be written as a product of at most $L$ elements of $S$.
            For any starting element $g$ and target element $h$, represent $hg^{-1}$ by such a product.
            Padding with identity elements gives a path of exactly $L$ steps from $g$ to $h$.
            Each step has probability at least $\eta$, so every entry of every $L$-step population transition product is at least  $\eta^L$.
            Hence the column-overlap bound holds with
            \[
                \varepsilon=|G|\eta^L,
            \]
            and \eqref{eq:unified:classical-rate} applies.
            The preceding cyclic keep--switch construction is recovered by taking the additive group $G=\mbZ/d\mbZ$, the outcome set $S=\{0,1\}$, and identifying $K$ with $0$ and $S$ with $1$.
            
        \paragraph{Special Case: amplitude damping.}
        \label{par:amplitude-damping}
            On $\mbC^2$ with the fixed basis $(e_1,e_2)$, let $p,\gamma:\Omega\to[0,1]$ be measurable.
            With environmental dependence suppressed in the following formulas, take the outcome set $\mcA=\{0,1,2,3\}$ and
            \[
                \begin{aligned}
                 V_0&=\sqrt p\,\operatorname{diag}(1,\sqrt{1-\gamma}),
                 &V_1&=\sqrt{p\gamma}\,\ket{e_1}\bra{e_2},\\
                 V_2&=\sqrt{1-p}\,\operatorname{diag}(\sqrt{1-\gamma},1),
                 &V_3&=\sqrt{(1-p)\gamma}\,\ket{e_2}\bra{e_1}.
                \end{aligned}
            \]
            These are weighted permutation operators, with the identity permutation for outcomes $0,2$ and the transposition of $1,2$ for outcomes $1,3$.
            Their total population matrix is
            \[
                 T_\omega=\begin{pmatrix}
                 1-(1-p_\omega)\gamma_\omega&p_\omega\gamma_\omega\\
                 (1-p_\omega)\gamma_\omega&1-p_\omega\gamma_\omega
                 \end{pmatrix},
                 \qquad \delta(T_\omega)=1-\gamma_\omega.
            \]
            Thus $\gamma_\omega\ge\gamma_*>0$ gives $r_n^{\mathrm{red}}\le2(1-\gamma_*)^n$.
            If the parameter pairs $(p_{\theta^j\omega},\gamma_{\theta^j\omega})$ are i.i.d. and $\mbE_{\pr}\gamma>0$, then $r_n^{\mathrm{red}}\le2(1-\mbE_{\pr}\gamma)^n$ and the reduced instrument also satisfies \Cref{assumption_mixing_base}.
            
            This channel also contracts the full quantum state.
            For a traceless Hermitian matrix
            \[
                X=\left(\begin{matrix}x&w\\\overline w&-x\end{matrix}\right),
            \]
            where $x\in\mbR$ and $w\in\mbC$, one has $\|X\|_1=2\sqrt{x^2+|w|^2}$.
            The channel multiplies $x$ by $1-\gamma$ and $w$ by $\sqrt{1-\gamma}$, so $\kappa(\Phi_\omega)=\sqrt{1-\gamma_\omega}$.
            Thus, reduction is unnecessary for full-state forgetting here.
            By \Cref{prop:unified:common}, one may take $r_n=2(1-\gamma_*)^{n/2}$ under the uniform lower bound, and $r_n=2(\mbE_{\pr}\sqrt{1-\gamma})^n$ under the i.i.d.\ assumption.
            Ordinary amplitude damping is the case $p\equiv1$, after deleting the identically zero outcomes $2,3$; its stationary state is $\ket{e_1}\bra{e_1}$ under these forgetting conditions.
            
            Under $\gamma_\omega\ge\gamma_*>0$, the second population obeys
            $z'=(1-\gamma)z+(1-p)\gamma$, and the stationary state is
            \begin{equation}
            \label{eq:GAD-stationary-population}
             \begin{aligned}
                 z(\omega)&=\sum_{h=0}^{\infty}
                 (1-p(\theta^{-h}\omega))\gamma(\theta^{-h}\omega)
                 \prod_{j=0}^{h-1}(1-\gamma(\theta^{-j}\omega)),\\
                 s(\omega)&=\operatorname{diag}(1-z(\omega),z(\omega)).
             \end{aligned}
            \end{equation}
            The empty product is one. Iterating the population update backwards gives this series; the residual initial population is bounded by $(1-\gamma_*)^n$ after $n$ iterations.

    \begin{example}[Imperfect classical readout.]
    \label{eg:imperfect-classical-readout}
        Let $P_\omega=(P_\omega(i,j))_{i,j=1}^d$ be a measurable row-stochastic matrix: $P_\omega(i,j)\ge0$ and  $\sum_{j=1}^dP_\omega(i,j)=1$ for every $i$.
        Use the outcome set $\mcA = [d]$, and define
        \[
            \mcT_{j;\omega}(X)
            =\sum_{i=1}^d P_\omega(i,j)\inner{e_i}{Xe_i}
                                 \ket{e_j}\bra{e_j}.
        \]
        The Kraus operators of branch $j$ are $\sqrt{P_\omega(i,j)}\ket{e_j}\bra{e_i}$, indexed by $i$.
        This instrument records only the destination $j$ of a basis transition.
        For the rank-one basis reduction, the branch matrices and their sum are
        \[
             B_{j;\omega}
             =u_j\left(P_\omega(1,j),\ldots,P_\omega(d,j)\right),
             \qquad
             T_\omega=P_\omega^{\mathsf T}.
        \]
        Each original branch already satisfies
        \[
             \mcT_{j;\omega}
             =JB_{j;\omega}R
             =\widehat{\mcT}_{j;\omega},
        \]
        so the reduction leaves the instrument unchanged.
        Consequently,
        \[
             \Phi_\omega=JP_\omega^{\mathsf T}R,
             \qquad
             \kappa(\Phi_\omega^{(n)})
             =\delta(T_\omega^{(n)}),
             \qquad n\ge1.
        \]
        If there exists $\varepsilon\in(0,1]$ such that
        \[
             \min_{i,i'\in[d]}
             \sum_{j=1}^d
             \min\left(P_\omega(i,j),P_\omega(i',j)\right)
             \ge\varepsilon
        \]
        for $\pr$-almost every $\omega$, then \Cref{prop:unified:common} verifies \Cref{assumption:forgetting} for the original instrument,   with
        \[
            r_n=2(1-\varepsilon)^n.
        \]
    \end{example}

    \begin{example}[State-independent emissions]
    \label{eg:state-independent-emissions}
        We consider the single-block case of the reduction above: $m=1$ and $P_1=I$.
        Then $R(X)=\operatorname{Tr}X$, and every quantum state has the same reduced population $R(\rho)=1$.
        The quantum dimension $d$ remains arbitrary.
        Let $(p_{a;\omega})_{a\in\mcA}$ be a measurable probability vector and let $\Psi_{a;\omega}$ be measurable quantum channels.
        Define
        \[
            \mcT_{a;\omega}
            =p_{a;\omega}\Psi_{a;\omega}.
        \]
        Since each $\Psi_{a;\omega}$ preserves trace,
        \[
             R\mcT_{a;\omega}(X)
             =p_{a;\omega}R(X).
        \]
        Thus the reduced branch matrices and their sum are the $1\times1$ matrices
        \[
             B_{a;\omega}=[p_{a;\omega}],
             \qquad
             T_\omega=[1].
        \]
        For $\pr$-almost every $\omega$, the outcome law satisfies
        \[
             \mbQ_{\rho;\omega}(A_1=a_1,\ldots,A_n=a_n)
             =\prod_{j=1}^n p_{a_j;\theta^j\omega}.
        \]
        Hence the outcomes are independent conditional on the environment, and their law is independent of the initial quantum state. 
        The sole reduced population is $\pi=1$, and $r_n^{\mathrm{red}}=0$ for every $n\ge1$.
        Consequently, the outcome CLT follows under \Cref{assumption_mixing_base} for the probability-vector process.
        This includes unitary channels $\Psi_{a;\omega}(X)=U_{a;\omega}XU_{a;\omega}^*$ and replacement channels $\Psi_{a;\omega}=\mcR_{\tau_{a;\omega}}$, with arbitrary measurable unitaries and states.
    \end{example}

    \paragraph{Special case: Synchronizing basis maps.}
    \label{par:synchronizing-basis-maps}
        Within the preceding family of state-independent emissions, we now give a sufficient condition for full-state forgetting.
        Fix an orthonormal basis $(e_i)_{i\in[d]}$ of $\mbC^d$ and fixed maps $f_a:[d]\to[d]$, which need not be bijective.
        Use the channels
        \[
             \Psi_a(X)
             =\sum_{i=1}^d\inner{e_i}{Xe_i}
                          \ket{e_{f_a(i)}}\bra{e_{f_a(i)}},
             \qquad
             \mcT_{a;\omega}=p_{a;\omega}\Psi_a.
        \]
        A word $u=(a_1,\ldots,a_L)\in\mcA^L$ is synchronizing if  there exists $k_*\in[d]$ such that
        \[
             f_{a_L}\circ\cdots\circ f_{a_1}(i)=k_*
             \qquad\text{for every }i\in[d].
        \]
        Its selective product then satisfies
        \[
         \mcT_{a_L;\theta^L\omega}\cdots
         \mcT_{a_1;\theta\omega}(X)
         =\varepsilon_\omega\operatorname{Tr}(X)
            \ket{e_{k_*}}\bra{e_{k_*}},
         \qquad
         \varepsilon_\omega
         :=\prod_{j=1}^L p_{a_j;\theta^j\omega}.
        \]
        Summing over the other words gives a completely positive remainder with trace $(1-\varepsilon_\omega)\operatorname{Tr}(X)$.
        Consequently,
        \[
         \kappa(\Phi_\omega^{(L)})\le1-\varepsilon_\omega,
        \]
        including when $\varepsilon_\omega=1$.
        A uniform lower bound $\varepsilon_\omega\ge\varepsilon_*>0$ for $\pr$-almost every $\omega$ therefore gives full-state forgetting, by \Cref{prop:unified:common}, with
        \[
             r_n=2(1-\varepsilon_*)^{\lfloor n/L\rfloor},
             \qquad n\ge1,
        \]
        where the zeroth power is interpreted as one.

        Since the channels $\Psi_a$ are fixed, the original instrument process inherits the mixing assumption from the probability-vector process. 
        Under the additional lower bound, the outcome CLT therefore also follows by applying the theorem directly to the original instrument.
        This lower bound is not required for the outcome CLT obtained through the scalar reduction.

    \paragraph{Special case: A countdown example.}
    \label{par:countdown}
        Fix $L_0\ge1$ and take $\mcH=\mbC^{L_0+1}$ with orthonormal basis $(e_i)_{i=0}^{L_0}$.
        Specialize the preceding synchronizing-map construction to the index set $\{0,\ldots,L_0\}$ and outcome alphabet $\mcA=\{a,b\}$.
        Set
        \[
             f_a(i)=\max(i-1,0),
             \qquad
             f_b(i)=i,
             \qquad
             p_a\in[0,1],
             \qquad
             p_b=1-p_a,
        \]
        with constant probabilities.
        Thus
        \[
             \mcT_c(X)
             =p_c\sum_{i=0}^{L_0}\inner{e_i}{Xe_i}
                   \ket{e_{f_c(i)}}\bra{e_{f_c(i)}},
             \qquad c\in\{a,b\}.
        \]
        Outcome $a$ decreases the basis index by one until it reaches zero, while outcome $b$ leaves the index unchanged.
        The outcomes are i.i.d.\ with probabilities $(p_a,p_b)$, independently of the initial quantum state.
        The outcome CLT therefore holds for every $p_a\in[0,1]$ through the scalar reduction, including the deterministic cases $p_a=0$ and $p_a=1$.
        
        If $p_a>0$, the word $a^{L_0}$ is synchronizing in the sense that it sends every basis index to zero.
        Its selective map is
        \[
             \mcT_a^{L_0}(X)
             =p_a^{L_0}\operatorname{Tr}(X)
                   \ket{e_0}\bra{e_0}.
        \]
        Consequently, the preceding full-state forgetting criterion applies with
        \[
             r_n=2(1-p_a^{L_0})^{\lfloor n/L_0\rfloor},
         \qquad n\ge1,
        \]
        where the zeroth power is interpreted as one.
        The unique stationary quantum state is $\s=\ket{e_0}\bra{e_0}$.
        Since the instrument is constant, the outcome CLT also follows directly from \Cref{thm:universal-functional} under this positive-probability condition.


\subsection{Measurement followed by preparation}
\label{subsec:unified:prepare}


The following examples give direct full-state forgetting bounds for
the original instruments, so no reduction is required.
Throughout this subsection and \Cref{subsec:unified:replacement}, \Cref{assumption1} is understood, and every application of \Cref{thm:universal-functional} assumes \Cref{assumption_mixing_base} for the instrument process whose outcomes are recorded, in addition to the displayed forgetting bounds.
All rates $r_n$ below concern full-state forgetting.

\begin{example}[Measurement followed by preparation.]
\label{eg:measurement-preparation}
    On $\mbC^d$, $d\ge2$, let $(F_{a;\omega})_{a\in\mcA}$  be a measurable positive operator-valued measure (POVM):
    $F_{a;\omega}\ge0$ and $\sum_{a\in\mcA}F_{a;\omega}=I_d$
    for $\pr$-almost every $\omega$.
    Let $\tau_{a;\omega}$ be measurable quantum states and define
    \begin{equation}
    \label{eq:unified:prepare}
         \mcT_{a;\omega}(X)
         =\operatorname{Tr}(F_{a;\omega}X)\tau_{a;\omega}.
    \end{equation}
    This family includes both perfect and imperfect instruments.
    
    For a POVM $F=(F_a)_a$ and a family of states $\tau=(\tau_a)_a$, put
    \[
     \mu(F)
     :=\sup_{\substack{X=X^*,\ \operatorname{Tr}X=0\\X\ne0}}
           \frac{\sum_a|\operatorname{Tr}(F_aX)|}{\|X\|_1},
     \qquad
     c(\tau):=\frac12\max_{a,b}\|\tau_a-\tau_b\|_1.
    \]
    Both coefficients belong to $[0,1]$, and
    \begin{equation}
    \label{eq:unified:prepare-bound}
     \kappa(\Phi_\omega)\le\mu(F_\omega)c(\tau_\omega).
    \end{equation}
    Indeed, for Hermitian $X$, $\sum_a|\operatorname{Tr}(F_aX)|\le\operatorname{Tr}|X|$.
    For a nonzero real vector $x=(x_a)_a$ with $\sum_ax_a=0$, write $x=b(u-v)$, where $u,v$ are probability vectors and  $b=\frac12\sum_a|x_a|$. Then
    \[
         \left\|\sum_ax_a\tau_a\right\|_1
         =b\left\|\sum_{a,a'}u_av_{a'}(\tau_a-\tau_{a'})\right\|_1
         \le c(\tau)\sum_a|x_a|.
    \]
    The same inequality is immediate when $x=0$.
    Taking $x_a=\operatorname{Tr}(F_aX)$ proves \eqref{eq:unified:prepare-bound}.
    Consequently, if $\mu(F_\omega)c(\tau_\omega)\le q<1$ for $\pr$-almost every $\omega$, then \Cref{prop:unified:common} gives a unique dynamically stationary  quantum state and verifies \Cref{assumption:forgetting} with
    \[
     r_n=2q^n,\qquad n\ge1.
    \]
\end{example}

    \paragraph{Special case: Tetrahedral measurements.}
    \label{par:tetrahedral-measurements}
        On $\mbC^2$, use $\mcA=[4]$ and let
        \[
            \begin{aligned}
             n_1&=(1,1,1)^{\mathsf T}/\sqrt3,
             &n_2&=(1,-1,-1)^{\mathsf T}/\sqrt3,\\
             n_3&=(-1,1,-1)^{\mathsf T}/\sqrt3,
             &n_4&=(-1,-1,1)^{\mathsf T}/\sqrt3.
        \end{aligned}
        \]
        For measurable rotations $O_\omega\in SO(3)$, take
        \[
             F_{i;\omega}
             =\frac{I_2+(O_\omega n_i)\cdot\boldsymbol\sigma}{4},
             \qquad
             \boldsymbol\sigma=(\sigma_x,\sigma_y,\sigma_z),
        \]
        where $\sigma_x,\sigma_y,\sigma_z$ are the Pauli matrices, and choose arbitrary measurable preparation states in \eqref{eq:unified:prepare}.
        Since $\sum_in_i=0$ and  $\sum_in_in_i^{\mathsf T}=\frac43I_3$, these effects form a rank-one POVM.
        For $X=v\cdot\boldsymbol\sigma/2$, $v\in\mbR^3$, one has $\|X\|_1=\|v\|_2$ and
            \[
             \sum_i|\operatorname{Tr}(F_{i;\omega}X)|
             =\frac14\sum_i|(O_\omega n_i)\cdot v|
             \le\frac14\sqrt{4\sum_i((O_\omega n_i)\cdot v)^2}
             =\frac{\|v\|_2}{\sqrt3}.
        \]
        Equality holds for $v=O_\omega(1,0,0)^{\mathsf T}$,  so $\mu(F_\omega)=1/\sqrt3$.
        Thus, for every choice of preparation states,
        \[
             r_n=2\,3^{-n/2},\qquad n\ge1.
        \]
        Pure preparation states make the instrument perfect; a mixed qubit preparation gives the corresponding branch  Kraus rank two.     
    
    \paragraph{Special case: Noisy measurements.}
    \label{par:noisy-measurements}
        Let $G_\omega=(G_{a;\omega})_{a\in\mcA}$ be a measurable POVM, let $(q_a)_{a\in\mcA}$ be a fixed probability vector, and let $\varepsilon:\Omega\to[0,1]$ be measurable.
        Set
        \[
             F_{a;\omega}
             =\varepsilon_\omega q_aI_d
                   +(1-\varepsilon_\omega)G_{a;\omega}.
        \]
        For traceless $X$, the first term contributes zero, so $\mu(F_\omega)=(1-\varepsilon_\omega)\mu(G_\omega)$.
        Thus a bound $\varepsilon_\omega\ge\varepsilon_*>0$ for $\pr$-almost every $\omega$ gives
        \[
             r_n=2(1-\varepsilon_*)^n
        \]
        for arbitrary measurable preparation states.
    
    \paragraph{Special case: Preparations with a common component.}
    \label{par:preparations-common-component}
        Let $b:\Omega\to[0,1]$ be measurable, and let $\zeta_\omega$ and $\sigma_{a;\omega}$ be measurable states.
        For an arbitrary measurable POVM, choose
        \[
             \tau_{a;\omega}
             =b_\omega\zeta_\omega+(1-b_\omega)\sigma_{a;\omega}.
        \]
        Then $c(\tau_\omega)=(1-b_\omega)c(\sigma_\omega)\le1-b_\omega$.
        Consequently, a bound $b_\omega\ge b_*>0$ for $\pr$-almost every $\omega$ gives
        \[
            r_n=2(1-b_*)^n,
        \]
        including for singular or noncommuting effects.

    
\subsection{Replacement in channel products}
\label{subsec:unified:replacement}


    Here full-state forgetting follows from a replacement component in the channel or in a finite product of channels.
    The mixing requirement stated in \Cref{subsec:unified:prepare} applies to all outcome CLTs below.

    \paragraph{A replacement criterion.}
        Suppose that, for some $L\ge1$, $\varepsilon\in(0,1]$, and measurable state $\tau_\omega$,
        \begin{equation}
        \label{eq:unified:replacement}
             \Phi_\omega^{(L)}(\rho)\ge\varepsilon\tau_\omega
             \qquad\text{for every state }\rho
        \end{equation}
        for $\pr$-almost every $\omega$.
        If $\varepsilon<1$, the map
        \[
         \frac{\Phi_\omega^{(L)}
               -\varepsilon\mcR_{\tau_\omega}}{1-\varepsilon}
        \]
        is positive and trace preserving, and therefore contracts the trace norm on Hermitian matrices.
        The replacement term vanishes on traceless matrices, giving $\kappa(\Phi_\omega^{(L)})\le1-\varepsilon$.
        If $\varepsilon=1$, the channel sends every state to $\tau_\omega$, so the same bound holds.
        By submultiplicativity and \Cref{prop:unified:common}, there is a unique dynamically stationary quantum state, and \Cref{assumption:forgetting} holds with
        \[
            r_n=2(1-\varepsilon)^{\lfloor n/L\rfloor},
             \qquad n\ge1.
        \]

        A sufficient condition to guarantee \cref{eq:unified:replacement} is a completely positive replacement component of mass at least $\varepsilon$ in the block channel.
        For example, suppose that a collection of distinct words $w\in\mcW\subseteq\mcA^L$ has selective maps
        \[
             X\longmapsto\operatorname{Tr}(F_{w;\omega}X)\tau_{w;\omega},
             \qquad
             F_{w;\omega}\ge c_{w;\omega}I_d,
             \qquad c_{w;\omega}\ge0.
        \]
        If $c_\omega:=\sum_{w\in\mcW}c_{w;\omega}\ge\varepsilon$, their sum contains the completely positive replacement component $c_\omega\mcR_{\tau_\omega}$, where
        \[
             \tau_\omega
             =\frac{1}{c_\omega}
               \sum_{w\in\mcW}c_{w;\omega}\tau_{w;\omega}.
        \]
        Here $c_\omega\le1$ because the selected word maps have total trace at most one on every state.
        This proves \eqref{eq:unified:replacement}.
    
    \begin{example}[Reset instruments.]
    \label{eg:reset-instrument-family}
        Let $(\mcJ_{a;\omega})_{a\in\mcA}$ be a measurable instrument, let $\tau_{a;\omega}$ be measurable states, and let $c_{a;\omega}\ge0$ be measurable with $h_\omega:=\sum_ac_{a;\omega}\le1$.
        Define
        \begin{equation}
        \label{eq:unified:reset-family}
             \mcT_{a;\omega}(X)
             =(1-h_\omega)\mcJ_{a;\omega}(X)
                   +c_{a;\omega}\operatorname{Tr}(X)\tau_{a;\omega}.
        \end{equation}
        Each branch is completely positive, and their sum preserves trace, so these maps form an instrument.
        Writing $\Lambda_\omega=\sum_a\mcJ_{a;\omega}$, one has
            \[
             \Phi_\omega
             =(1-h_\omega)\Lambda_\omega
                   +\sum_ac_{a;\omega}\mcR_{\tau_{a;\omega}},
             \qquad
             \kappa(\Phi_\omega)
             \le(1-h_\omega)\kappa(\Lambda_\omega)
             \le1-h_\omega.
        \]
        Thus $h_\omega\ge h_*>0$ for $\pr$-almost every $\omega$ gives
        \[
            r_n=2(1-h_*)^n.
        \]
        Alternatively, if $(h_{\theta^j\omega})_{j\in\mbZ}$ is
        i.i.d.\ and $\mbE_{\pr}h>0$, then
        \[
             \mbE_{\pr}\kappa(\Phi_\omega^{(n)})
             \le\mbE_{\pr}\prod_{j=1}^n(1-h_{\theta^j\omega})
             =(1-\mbE_{\pr}h)^n,
             \qquad
             r_n=2(1-\mbE_{\pr}h)^n.
        \]
        Independence of the reset strengths suffices for this forgetting bound.
    \end{example}

    The following choices give elementary special cases of \eqref{eq:unified:reset-family}. 
    In each case, the preceding forgetting bounds apply with the indicated total reset probability $h_\omega$.
    All data in the table below are measurable, and unspecified residual branches are chosen so that $(\mcJ_{a;\omega})_a$ is an instrument.

    \begin{table}[htbp]
    \centering
    \small
    \renewcommand{\arraystretch}{1.25}
    \caption{Elementary special cases of reset instruments.}
    \label{tab:reset-special-cases}
    \begin{tabularx}{\linewidth}{
        @{}>{\raggedright\arraybackslash}p{.24\linewidth}
        >{\raggedright\arraybackslash}X@{}
    }
    \toprule
    Special case & Choices in the reset family \\
    \midrule
    
    One or several visible reset labels
    &
    Choose $\mcA_{\mathrm r}\subseteq\mcA$.
    Set $\mcJ_{r;\omega}=0$ for
    $r\in\mcA_{\mathrm r}$ and $c_{a;\omega}=0$ outside
    $\mcA_{\mathrm r}$, with arbitrary preparation states
    $\tau_{r;\omega}$.
    Then
    $\mcT_{r;\omega}
     =c_{r;\omega}\mcR_{\tau_{r;\omega}}$
    and
    $h_\omega=\sum_{r\in\mcA_{\mathrm r}}c_{r;\omega}$.
    \\
    \addlinespace
    
    Random-unitary evolution with a visible reset
    &
    Use $\mcA=G\cup\{*\}$, where $G$ is a finite label set
    and $*\notin G$.
    Take
    $\mcJ_{g;\omega}(X)
     =q_{g;\omega}U_{g;\omega}XU_{g;\omega}^*$,
    where $(q_{g;\omega})_{g\in G}$ is a probability vector
    and $U_{g;\omega}$ are unitaries.
    Set
    $\mcJ_{*;\omega}=0$,
    $c_{*;\omega}=p_\omega$,
    $c_{g;\omega}=0$, and
    $\tau_{*;\omega}=\sigma_{\theta\omega}$,
    where $p_\omega\in[0,1]$ and $\sigma_\omega$ is a state.
    Here $h_\omega=p_\omega$.
    \\
    \addlinespace
    
    Thermal replacement
    &
    Use a visible reset label $*$ and take
    $\mcJ_{*;\omega}=0$,
    $c_{*;\omega}=\alpha_\omega\in[0,1]$, and
    $c_{a;\omega}=0$ for $a\ne *$.
    Choose
    $\tau_{*;\omega}
     =e^{-\beta H_{\theta\omega}}/
       \operatorname{Tr}(e^{-\beta H_{\theta\omega}})$,
    where $H_\omega$ is Hermitian and $\beta\ge0$.
    Here $h_\omega=\alpha_\omega$.
    Fixed $H_\omega=H$ gives a fixed Gibbs preparation.
    \\
    \bottomrule
    \end{tabularx}
    \end{table}
    
    Observing a visible reset label $r$ of positive probability gives posterior state $\tau_{r;\omega}$.
    The general forgetting bounds use the total reset probability $h_\omega$; individual reset probabilities may vanish.
    The general reset family also permits thermal replacement and a residual contribution to share the same recorded label.

    We give one more special case below. 

    \paragraph{Special case: Amplitude damping with replacement.}
    \label{par:amplitude-damping-replacement}
        Let $\mcA_{p_\omega,\gamma_\omega}$ denote the full amplitude-damping channel in \hyperref[par:amplitude-damping]{the amplitude-damping special case}.
        For a measurable $h:\Omega\to[0,1]$ and measurable state
            $\tau_\omega$, consider
            \[
             \Phi_\omega
             =(1-h_\omega)\mcA_{p_\omega,\gamma_\omega}
                   +h_\omega\mcR_{\tau_\omega}.
        \]
        On traceless matrices the replacement term vanishes, so
        \[
             a_\omega:=\kappa(\Phi_\omega)
             =(1-h_\omega)\sqrt{1-\gamma_\omega}.
        \]
        Consequently, $a_\omega\le a_*<1$ for $\pr$-almost every $\omega$ gives $r_n=2a_*^n$.
        If $(a_{\theta^j\omega})_{j\in\mbZ}$ is i.i.d.\ and $\mbE_{\pr}a<1$, then
        \[
            r_n=2(\mbE_{\pr}a)^n.
        \]
        In particular, an i.i.d.\ process of pairs $(h_{\theta^j\omega},\gamma_{\theta^j\omega})$ suffices for this product estimate; the two coordinates at a fixed time need not be independent.
        Damping can supply contraction even when $h_\omega=0$.
        A visible-reset realization uses $\mcA=\{0,1,2,3,*\}$ and
        \[
             \mcT_{j;\omega}(X)
             =(1-h_\omega)V_{j;\omega}XV_{j;\omega}^*,
             \quad j\in\{0,1,2,3\},
             \qquad
             \mcT_{*;\omega}(X)
             =h_\omega\operatorname{Tr}(X)\tau_\omega,
        \]
        where $V_{j;\omega}$ are the displayed amplitude-damping operators. 
        One may in particular take $\tau_\omega=\operatorname{diag}(p_\omega,1-p_\omega)$.


\section{Preliminaries and notation}
\label{section:proof_first}

    We start this section by collecting the notation and basic objects used throughout the article.
    We then give the contraction criteria and mixing estimates used in the proof of \Cref{thm:universal-functional}.

    Throughout, \(\mcH=\mbC^d\) and \(\mbM_d=\mcB(\mcH)\).
    We also take $d\ge 2$. 
    We use the trace norm and Hilbert--Schmidt inner product
    \[
        \norm{X}_1=\tr{|X|},
        \qquad
        \inner{A}{B}=\tr{A^\dagger B}.
    \]
    The adjoint \(\phi\adj\) of a linear map \(\phi:\mbM_d\to\mbM_d\) is taken with respect to this inner product.
    We write
    \[
        \states
        :=
        \{\rho\in\mbM_d:\rho\geq0,\ \tr{\rho}=1\},
        \qquad
        \states^{\mathrm o}
        :=
        \{\rho\in\states:\rho>0\}.
    \]
    A positive linear map is called \emph{strictly positive} if it maps every nonzero positive semidefinite matrix to a positive definite matrix.

    Let $(X,\mcF,\mu)$ be a probability space, where $\mcF$ is the $\sigma$-algebra of events and $\mu$ is a probability measure on $(X,\mcF)$. 
    A map $\theta:X\to X$ is called \emph{$\mu$-probability preserving} (or \emph{ppt} for short) if $\theta$ is $\mcF$-measurable and
    \[
        \mu(\theta^{-1}(E)) = \mu(E)
        \qquad\text{for all }E\in\mcF.
    \]
    When $\theta$ is bijective and $\theta^{-1}$ is also $\mcF$-measurable, we say that $\theta$ is an \emph{invertible} ppt. 
    A ppt $\theta$ is called \emph{ergodic} if for every $E\in\mcF$ with $\theta^{-1}(E)=E$ we have $\mu(E)\in\{0,1\}$. 
    An equivalent characterization of ergodicity of a $\mu$-ppt is that members $E$ of $\mcF$ with $\mu(\theta^{-1}(E) \triangle E) =0 $ are those with $\mu(E) \in \{0,1\}$. 
    Here $\triangle$ denotes the symmetric difference. 
    We refer the reader to \cite{erg_walter} for further characterizations of ergodicity.
    We denote by $\mbE_\mu[\,\cdot\,]$ the expectation with respect to $\mu$. 
    When there is no ambiguity, we write simply $\mbE[\,\cdot\,]$.

    Fix a reference state \(\rho_*\in\states^{\mathrm o}\).
    For a positive linear map \(\phi:\mbM_d\to\mbM_d\) and \(X\geq0\),
    define
    \begin{equation}
    \label{eq:projective_action}
        \phi\proj X
        :=
        \begin{cases}
            \displaystyle
            \frac{\phi(X)}{\norm{\phi(X)}_1},
                & \phi(X)\neq0,\\[1.2ex]
            \rho_*,
                & \phi(X)=0.
        \end{cases}
    \end{equation}

    \begin{definition}
    \label{dfn:cnum}
        For a positive linear map \(\phi:\mbM_d\to\mbM_d\), with \(\ker\phi\cap\states=\emptyset\), define
        \[
            \cnum{\phi}
            :=
            \sup_{\rho,\rho'\in\states}
            \mathrm d\left(\phi\proj\rho,\phi\proj\rho'\right),
        \]
        where
        \[
            \mathrm d(\rho,\rho')
            :=
            \frac{1-m(\rho,\rho')m(\rho',\rho)}
                 {1+m(\rho,\rho')m(\rho',\rho)},
            \qquad
            m(\rho,\rho')
            :=
            \sup\{\lambda\geq0:\lambda\rho'\leq\rho\}.
        \]
    \end{definition}

    The following shall be required in the analysis below.

    \begin{lemma}[{\cite[Lemmas~3.5, 3.8, and~3.10]{MS22}}]
    \label{lemma:cnum_properties}
        Let \(\phi:\mbM_d\to\mbM_d\) be positive and suppose that \(\ker\phi\cap\states=\ker{\phi\adj}\cap\states=\emptyset\).
        Then, for all \(\rho,\rho'\in\states\),
        \[
            \frac12
            \norm{\phi\proj\rho-\phi\proj\rho'}_1
            \leq
            \mathrm d\left(\phi\proj\rho,\phi\proj\rho'\right)
            \leq
            \cnum{\phi}\,\mathrm d(\rho,\rho')
            \leq
            \cnum{\phi}.
        \]
        Moreover, \(0\leq\cnum{\phi}\leq1\), with \(\cnum{\phi}<1\) whenever \(\phi\) is strictly positive.
        If \(\psi:\mbM_d\to\mbM_d\) is another positive map satisfying
        \(\ker{\psi}\cap\states=\ker{\psi\adj}\cap\states=\emptyset\), then \(\cnum{\psi\circ\phi}\leq\cnum{\psi}\cnum{\phi}\).
    \end{lemma}

    For \(n\in\mbN\) and \(\pmb a=(a_1,\ldots,a_n)\in\mcA^n\), set
    \[
        \mcT^{(n)}_{\pmb a;\omega}
        :=
        \mcT_{a_n;\theta^n(\omega)}
        \circ\cdots\circ
        \mcT_{a_1;\theta(\omega)}.
    \]
    For \(\pr\)-almost every \(\omega\in\Omega\), let \(\mbQ^*_\omega:\Sigma\to\mbM_d\) be the unique POVM satisfying
    \begin{equation}
    \label{eq:duality}
        \inner{\rho}{\mbQ^*_\omega(S)}
        =
        \mbQ_{\rho;\omega}(S),
        \qquad
        \rho\in\states,\quad S\in\Sigma.
    \end{equation}
    Existence of this POVM follows from \cite[Lemma~2.2]{qtlln}, and uniqueness follows from \eqref{eq:duality}.
    Although the result is stated there in the perfect measurement setting, the same construction applies to general quantum instruments. 
    We also recall the following identity from that work.

    \begin{lemma}[{\cite[Lemma~3.1]{qtlln}}]
    \label{lem:disordered_k_m_lemma_1_2}
        Let \(S\in\Sigma\), \(n\in\mbN\), and
        \(\pmb a\in\mcA^n\). Then, for \(\pr\)-almost every
        \(\omega\),
        \[
            \mbQ^*_\omega
            \left(
                \varsigma^{-n}(S)
                \cap
                A_{1:n}^{-1}(\{\pmb a\})
            \right)
            =
            \left(\mcT^{(n)}_{\pmb a;\omega}\right)\adj
            \left(
                \mbQ^*_{\theta^n(\omega)}(S)
            \right).
        \]
    \end{lemma}

    In the lemma above \(\varsigma:\mcA^{\mbN}\to\mcA^{\mbN}\) denotes the left shift. 

    We now introduce the mixing profiles associated with the joint instrument--outcome process. 
    For each \(k\in\mbN\), define
    \begin{equation}
        \label{eq:mixing+joint}
        \mcG_k
        :=
        \mcF_k\otimes\sigma(A_1,\ldots,A_k),
        \qquad
        \mcG^k
        :=
        \mcF^k\otimes\sigma(A_j:j\geq k),
    \end{equation}
    where $\mcF_k$ and $\mcF^k$ are the past and future sigma-algebras defined in \cref{eq:past_and_future_sigma_instruments} for the instrument sequence. These profiles use the specified instrument and outcome sigma-algebras; they do not assert mixing of the full skew system on $\mcF\otimes\Sigma$. 

    For a random initial state \(\vartheta:\Omega\to\states\), define
    \begin{align}
        \overline{\alpha}_{\overline{\mbQ}_{\vartheta}}(n)
        &:=
        \sup_{k\in\mbN}
        \alpha_{\overline{\mbQ}_{\vartheta}}
        \left(\mcG_k,\mcG^{k+n}\right),
        \label{eq:alpha_n_skew}
        \\
        \overline{\rho}_{\overline{\mbQ}_{\vartheta}}(n)
        &:=
        \sup_{k\in\mbN}
        \rho_{\overline{\mbQ}_{\vartheta}}
        \left(\mcG_k,\mcG^{k+n}\right),
        \label{eq:rho_n_skew}
    \end{align}
   where \(\alpha_{\overline{\mbQ}_{\vartheta}}\) and \(\rho_{\overline{\mbQ}_{\vartheta}}\) are defined by the same formulas as in \eqref{eq:alpha_mixing_cooef} and \eqref{eq:rho_mixing_cooef}, with \(\pr\) replaced by \(\overline{\mbQ}_{\vartheta}\).

    \smallskip

    We shall also use the following standard covariance inequalities for future use. 

    \begin{lemma}[{\cite[\S1.2 Thm.~3]{doukhan2012mixing}}, {\cite[Thm.~3.8]{bradley2007introduction}}]
    \label{lemma:covariance_bounds}
        Let $\mcA,\mcB$ be sub–$\sigma$–algebras in the probability space  $(X, \mcF, \mu)$. 
        Let $U\in L^2(\mcA)$ and $V\in L^2(\mcB)$. Then
        \[
            |\cov{U}{V}|\ \le\ \rho(\mcA,\mcB)\,\|U\|_{L^2}\,\|V\|_{L^2}\,.
        \]
    \end{lemma}
    
    \begin{lemma}[{\cite[Theorem~3.7]{bradley2007introduction}}]
    \label{lemma:correlation_bounds}
        Let $\mcA, \mcB$ be sub-$\sigma$-algebras in the probability space $(X, \mcF, \mu)$ and let $p,q\in (1,\infty]$ satisfy $1/p+1/q <1$. If $A\in L^p_{\mathrm{real}}(\mcA)$ and $B\in L^q_{\mathrm{real}}(\mcB)$, then 
        \begin{equation}
            \left|\mbE_{\mu}[AB] - \mbE_{\mu}[A]\mbE_{\mu}[B]\right|
            \le
            8 \alpha(\mcA,\mcB)^{1-1/p-1/q}\norm{A}_{L^p(\mu)}\norm{B}_{L^q(\mu)}
        \end{equation}
    \end{lemma}
    

\subsection{Eventual Strict Positivity and Annealed Forgetting}
\label{subsec:esp_forgetting}


    We first give a sufficient condition for the existence of a unique, strictly positive dynamically stationary state. 

    \begin{prop}
    \label{prop:esp_stationary_state}
        Suppose \Cref{assumption1} holds, and assume:
        \begin{description}
        \hypertarget{esp}{}
        \item[(\esp)]
            For \(\pr\)-almost every \(\omega\in\Omega\), there exists \(N_0(\omega)\in\mbN\) such that \(\Phi_\omega^{(n)}\) is strictly positive for every \(n\ge N_0(\omega)\).
        \end{description}
        Then there exists a unique, up to \(\pr\)-almost-sure equality, dynamically stationary state \(\s:\Omega\to\states\).
        Moreover, \(\s(\omega)\in\states^{\mathrm o}\) almost surely and, on a common event of full probability,
        \[
            \Phi_\omega^{(n)}\left(\s(\omega)\right)
            =\s(\theta^n\omega),
            \qquad n\in\mbN.
        \]
    \end{prop}
    \begin{proof}
        Condition~(\esp) is the eventual strict positivity assumption in \cite{MS22}. 
        Theorem~1 of \cite{MS22} then gives a measurable random state \(\s:\Omega\to\states^{\mathrm o}\) satisfying
        \[
            \Phi_{\theta\omega}\left(\s(\omega)\right)
            =\s(\theta\omega)
            \qquad\text{almost surely}.
        \]
        Here the projective action in that theorem agrees with the ordinary
        action because the maps are trace-preserving. 
        Recall from \cite[Definition~1 and Remark~2]{ES24} that a measurable random orthogonal projection \(P\) reduces the forward cocycle precisely when    
        \begin{equation}
        \label{eq:reducing_proj}
            \Phi_\omega^{(n)}\left(P(\omega)\right)
            \le
            \tr{\Phi_\omega^{(n)}(P(\omega))}
                P(\theta^n\omega)
            =
            \tr{P(\omega)}P(\theta^n\omega)
            \le
            d\,P(\theta^n\omega),
        \end{equation}
        $\mbP$-almost surely.
        In particular, the constant projection \(\mbI_d\) is reducing.

        Let \(P\) be a nonzero reducing projection, meaning that \(\pr(P\ne0_d)>0\), and put \(A:=\{\omega:P(\omega)\ne0_d\}\).
        If \(P(\theta\omega)=0_d\), then \eqref{eq:reducing_proj} with \(n=1\)
        and trace preservation imply
        \[
            \tr{P(\omega)}
            =
            \tr{\Phi_{\theta\omega}(P(\omega))}
            =
            0.
        \]
        Thus \(A\subseteq\theta^{-1}A\) modulo null sets.
        Since \(\theta\) preserves \(\pr\), these events have the same probability and hence coincide modulo null sets. 
        Ergodicity and \(\pr(A)>0\) give \(\pr(A)=1\).
        Consequently, by (\esp) and \eqref{eq:reducing_proj}, $\mbP$-almost surely
        \[
            0<\Phi_\omega^{(n)}\left(P(\omega)\right)
            \le d\,P(\theta^n\omega)
            \qquad\text{for all sufficiently large }n.
        \]
        A positive definite matrix cannot be supported on a proper subspace, so \(P(\theta^n\omega)=\mbI_d\) for all sufficiently large \(n\).
        Measure preservation and dominated convergence therefore yield
        \[
            \pr(P\ne\mbI_d)
            =
            \mbE_{\pr}\!\left[
                \mathbf1_{\{P\ne\mbI_d\}}\circ\theta^n
            \right]
            \longrightarrow0.
        \]
        Hence \(P=\mbI_d\) almost surely, and \(\mbI_d\) is the only nonzero
        reducing projection, in particular a minimal one.
        Therefore, by \cite[Theorem~1]{ES24}, there is a unique dynamically stationary state whose values are positive definite almost surely.
        The state \(\s\) constructed above has this property and thus must be this state. 
    \end{proof}

    We next obtain a uniform annealed forgetting estimate by adding \(\rho\)-mixing of the instrument process, see \Cref{eq:rho_n__and_alpha_n_disorder}. 

    \begin{prop}
    \label{prop:rho_esp_forgetting}
        Suppose \Cref{assumption1} and Condition~{\esp}  hold.
        Assume furthermore that \(\rho_{\pr}(n)\to0\). 
        Then, for every \(p\in\mbN\), there exists a deterministic constant \(C_p>0\) such that
        \[
            \mbE_{\pr}\!\left[\cnum{\Phi_\omega^{(n)}}\right]
            \le C_p n^{-p},
            \qquad n\in\mbN.
        \]
        Consequently, for every measurable random initial state \(\vartheta:\Omega\to\states\),
        \[
            \beta_n(\vartheta)\le2C_p n^{-p},
            \qquad n\in\mbN,
        \]
        with \(C_p\) independent of \(\vartheta\). 
        In particular, \Cref{assumption:forgetting} holds. 
        Moreover, for every \(\delta>0\), the deterministic bounds \(r_n\) in that assumption can be chosen to satisfy \(\sum_{n\ge1}r_n^{\delta/(2+\delta)}<\infty\).
    \end{prop}
    \begin{proof}
        Put
        \[
            c_n(\omega):=\cnum{\Phi_\omega^{(n)}},
            \qquad a_n:=\mbE_{\pr}[c_n].
        \]
        By \cite[Lemma~2.1]{MS22}, (\esp) implies both kernel conditions in \Cref{lemma:cnum_properties}. 
        Thus \(0\le c_n\le1\), and submultiplicativity shows that \((a_n)\) is nonincreasing.
        Furthermore, \cite[Lemma~3.11]{MS22} gives \(c_n\to0\) almost surely, so dominated convergence gives \(a_n\to0\).
        
        We now use the block argument of
        \cite[Lemma~4.6]{MPS} and \cite[Lemma~4.5]{asym}, applied here to the projective coefficient \(\cnum{\,\cdot\,}\).
        For \(r,s,t\ge1\), split the product into consecutive blocks of lengths \(r,s,t\). 
        Submultiplicativity and the bound on the middle factor give
        \[
            c_{r+s+t}(\omega)
            \le c_r(\omega)c_t(\theta^{r+s}\omega).
        \]
        The two factors on the right are measurable with respect to \(\mcF_r\) and \(\mcF^{r+s+1}\), respectively, because each coefficient is a measurable function of its channel block. 
        Then \Cref{lemma:covariance_bounds} gives 
        \begin{equation}
        \label{eq:esp_block_recursion}
            a_{r+s+t} \le a_ra_t+\rho_{\pr}(s+1)\sqrt{a_ra_t}
            \le  a_ra_t+\rho_{\pr}(s)\sqrt{a_ra_t}. 
        \end{equation}
        Here we have also used the monotonicity of \(\rho_{\pr}\).
        
        Fix \(p\in\mbN\). Choose \(M,N\ge1\) such that
        \[
            \rho_{\pr}(M)+a_N\le2^{-(p+1)}.
        \]
        Set \(n_0=N\) and \(n_{j+1}=2n_j+M\). 
        By \eqref{eq:esp_block_recursion} and monotonicity,
        \[
            a_{n_{j+1}}
            \le a_{n_j}\left(a_{n_j}+\rho_{\pr}(M)\right)
            \le2^{-(p+1)}a_{n_j}.
        \]
        Thus \(a_{n_j}\le a_N2^{-(p+1)j}\), while \(n_j=2^j(N+M)-M\). If \(n_j\le n<n_{j+1}\), then
        \[
            a_n\le a_{n_j}
            \le a_N2^{-(p+1)j}
            \le a_N[2(N+M)]^p n^{-p}.
        \]
        Increasing the constant to include \(1\le n<N\) proves
        \begin{equation}
        \label{eq:cnum_decay_here}
            \mbE_{\pr}\!\left[\cnum{\Phi_\omega^{(n)}}\right]
            \le C_p n^{-p},
            \qquad n\ge1.
        \end{equation}
        
        Finally, dynamical stationarity and \Cref{lemma:cnum_properties} give, pointwise on a common event of full probability and for every
        \(\rho\in\states\),
        \[
            \norm{\Phi_\omega^{(n)}(\rho)-\s(\theta^n\omega)}_1
            =\norm{\Phi_\omega^{(n)}(\rho)
                   -\Phi_\omega^{(n)}(\s(\omega))}_1
            \le2c_n(\omega).
        \]
        Taking \(\rho=\vartheta(\omega)\) and then $\mbP$-expectations proves \(\beta_n(\vartheta)\le2C_pn^{-p}\), uniformly over all measurable random initial states.
        Choosing \(p>1\) verifies \Cref{assumption:forgetting}. 
        More generally, for a given \(\delta>0\), choose \(p\delta/(2+\delta)>1\) and put \(r_n=2C_pn^{-p}\) to obtain the asserted power summability.
    \end{proof}

    With \Cref{prop:esp_stationary_state} we have the following immediate corollary. 
    
    \begin{cor}
    \label{cor:quenched_absolute_continuity}
        Under the hypotheses of \Cref{prop:esp_stationary_state}, for every measurable random initial state \(\vartheta:\Omega\to\states\),
        \[
            \mbQ_{\vartheta(\omega);\omega}
            \ll\mbQ_{\s(\omega);\omega}
            \qquad\text{for \(\pr\)-almost every \(\omega\)}.
        \]
    \end{cor}
    \begin{proof}
        Fix a common event of full probability on which \(\s(\omega)>0\) and the POVM representation \(\mbQ_{\rho;\omega}(E)=\inner{\rho}{\mbQ^*_\omega(E)}\) holds for every \(\rho\in\states\) and every \(E\in\Sigma\).
        For \(\omega\) in this event, let \(\lambda_\omega:=\lambda_{\min}(\s(\omega))>0\), where $\lambda_{\min}(\,\cdot\,)$ denotes the smallest eigenvalue. 
        Every density matrix \(\rho\) satisfies
        \[
            0_d\le\rho\le\mbI_d\le\lambda_\omega^{-1}\s(\omega).
        \]
        Since \(\mbQ^*_\omega(E)\ge0\), it follows that
        \[
            \mbQ_{\rho;\omega}(E)
            \le\lambda_\omega^{-1}\mbQ_{\s(\omega);\omega}(E),
            \qquad E\in\Sigma.
        \]
        Taking \(\rho=\vartheta(\omega)\) proves the assertion.
    \end{proof}

    Condition~(\esp) is a sufficient criterion, not a standing assumption for the remainder of the paper. 
    Under \Cref{assumption1}, existence of at least one dynamically stationary state already follows from \Cite{Beck_1957}, without a contraction hypothesis.
    Here (\esp) supplies strict positivity and uniqueness, and the additional condition \(\rho_{\pr}(n)\to0\) supplies \Cref{assumption:forgetting}. 
    The instrument-mixing requirement in \Cref{assumption_mixing_base} remains separate. 
    The proofs of \Cref{thm:aclt,thm:universal-functional} use only \Cref{assumption1,assumption_mixing_base,assumption:forgetting} and do not require (\esp) itself.


\subsection{Augmented Doeblin and non-selective forgetting}
\label{subsec:doeblin_forgetting}


    We now give a sufficient condition for \Cref{assumption:forgetting} in the imperfect-measurement regime.
    It is a block minorization condition on the \emph{augmented} selective kernel, which records the posterior state after a block and the observed block of outcomes.
    For $L\in\mbN$ and $\pmb a=(a_1,\ldots,a_L)\in\mcA^L$, set
    \begin{equation}
    \label{eq:block-selective-map}
         \mcT^{(L)}_{\pmb a;\omega}
         :=\mcT_{a_L;\theta^L\omega}\circ\cdots\circ\mcT_{a_1;\theta\omega}.
    \end{equation}
    Write $\mcT^{(L)}_{\pmb a;\omega}\proj\rho$ for this output divided by its trace when the trace is positive, and a fixed reference state $\rho_\ast\in\states$ otherwise.
    Define the $L$-step augmented selective kernel by
    \begin{equation}
    \label{eq:augmented-L-step-kernel}
         K^{(L)}_\omega(\rho;B)
         :=\sum_{\pmb a\in\mcA^L}
         \tr{\mcT^{(L)}_{\pmb a;\omega}(\rho)}
         \mathbf1_B\!\left(\mcT^{(L)}_{\pmb a;\omega}\proj\rho,\pmb a\right),
        \qquad B\in\borel{\states\times\mcA^L}.
    \end{equation}
    Thus $K^{(L)}_\omega(\rho;\,\cdot\,)$ is the law of the posterior state and recorded word for a trajectory started from $\rho$ in environment $\omega$.
    Its weights sum to $\tr{\Phi^{(L)}_\omega(\rho)}=1$.
    Finite-dimensional composition and evaluation show that each word map and its trace are jointly measurable in $(\omega,\rho)$; division on the positive-trace set and the fixed value on its complement preserve measurability.
    For a Borel event $B$, each section $B_{\pmb a}=\{\eta:(\eta,\pmb a)\in B\}$ is Borel. 
    Formula~\eqref{eq:augmented-L-step-kernel} is therefore a finite sum of measurable weights times measurable indicators. This proves both normalization and joint measurability of the probability kernel.

    \begin{assumption}[Augmented block Doeblin condition]
    \label{assumption:block-doeblin}
        There exist $L\in\mbN$, $\varepsilon\in(0,1)$, and a probability kernel $\nu:\Omega\times\borel{\states\times\mcA^L}\to[0,1]$ such that, on a full-measure set $\Omega_{\mathrm D}$,
        \begin{equation}
        \label{eq:augmented-doeblin}
             K^{(L)}_\omega(\rho;B)\ge\varepsilon\nu_\omega(B)
             \qquad\text{for every }\rho\in\states
             \text{ and }B\in\borel{\states\times\mcA^L}.
    \end{equation}
    \end{assumption}
    
    \begin{remark}
    \label{rem:doeblin-perfect-obstruction}
        For $d\ge2$, no perfect instrument satisfies \Cref{assumption:block-doeblin}. 
        Indeed, fix an environment where the minorization holds and write the selective map of a word $u$ as $X\mapsto V_uXV_u^*$.
        Let $\nu_{\omega,u}(D)=\nu_\omega(D\times\{u\})$ for Borel $D\subseteq\states$.
        If $V_u$ is singular, a state supported in its kernel gives zero probability to $u$, forcing $\nu_{\omega,u}=0$.
        If $V_u$ is invertible, the normalized posterior map $\rho\mapsto V_u\rho V_u^*/\tr{V_u\rho V_u^*}$ is injective. Two distinct input states give distinct posterior atoms for the same word, so no positive measure can be dominated by both, again forcing $\nu_{\omega,u}=0$.
        There are finitely many words, contradicting the normalization of $\nu_\omega$.
        Thus imperfection is necessary for this augmented condition when $d\ge2$, but is not sufficient. The state space is a singleton when $d=1$.
    \end{remark}

    We show that \Cref{assumption:block-doeblin} is sufficient to get the existence of a unique dynamically stationary state, and we also get \Cref{assumption:forgetting}.

    \begin{restatable}[]{thm}{rhosexists}
    \label{thm:rho_ss_exists}
        Suppose \Cref{assumption1,assumption:block-doeblin} hold.
        Then there exists a dynamically stationary state $\s:\Omega\to\states$, unique up to $\pr$-null sets.
        Moreover, if $L$ and $\varepsilon$ are as in \Cref{assumption:block-doeblin}, then with
        \[
             \lambda:=(1-\varepsilon)^{1/L}\in(0,1),
             \qquad C:=2(1-\varepsilon)^{-1},
        \]
        one has
        \begin{equation}
        \label{eq:doeblin-prop-forgetting-pointwise}
             \unorm{\Phi^{(n)}_\omega(\rho)-\s(\theta^n\omega)}_1
             \le C\lambda^n,
             \qquad \rho\in\states,\ n\in\mbN,
        \end{equation}
        on a full-measure set independent of $\rho$ and $n$.
    \end{restatable}
    \begin{proof}
        Work on the invariant full-measure set $\bigcap_{j\in\mbZ}\theta^{-j}\Omega_{\mathrm D}$, also restricted so that the instrument identities hold at every shift.
        For each fixed \(\omega\), define
        \[
            \tau_\omega
            :=
            \int_{\states\times\mcA^L}
                \eta\,\nu_\omega(\dee\eta,\dee\pmb a).
        \]
        Here the integrand is the coordinate projection $(\eta,\pmb a)\mapsto\eta$ on $\states\times\mcA^L$.
        Since $\|\eta\|_1=1$ for every $\eta\in\states$, this matrix-valued integral is well-defined entrywise.
        It preserves Hermiticity, and for every $v\in\mbC^d$,
        \[
            v^*\tau_\omega v
            =
            \int_{\states\times\mcA^L}
                v^*\eta v\,\nu_\omega(\dee\eta,\dee\pmb a)
            \ge0.
        \]
        Moreover, since $\nu_\omega$ is a probability measure and $\operatorname{Tr}(\eta)=1$ for every $\eta\in\states$,
        \[
            \operatorname{Tr}(\tau_\omega)
            =
            \int_{\states\times\mcA^L}
                \operatorname{Tr}(\eta)\,
                \nu_\omega(\dee\eta,\dee\pmb a)
            =
            \nu_\omega(\states\times\mcA^L)
            =1.
        \]
        Thus $\tau_\omega\in\states$.
        Its measurability in $\omega$ follows by integrating each matrix entry against the probability kernel $\nu_\omega$.

        By the definition of the augmented kernel,
        \[
            \int_{\states\times\mcA^L}
                \eta\,K^{(L)}_\omega
                (\rho;\dee\eta,\dee\pmb a)
            =
            \sum_{\pmb a\in\mcA^L}
                \mcT^{(L)}_{\pmb a;\omega}(\rho)
            =
            \Phi^{(L)}_\omega(\rho),
        \]
        where words of zero probability contribute zero.
        For every \(\rho\in\states\), the Doeblin condition says that \(K^{(L)}_\omega(\rho;\cdot)-\varepsilon\nu_\omega\)  is a nonnegative measure.
        Since every \(\eta\in\states\) is positive semidefinite,
        \[
            \Phi^{(L)}_\omega(\rho)-\varepsilon\tau_\omega
            =
            \int_{\states\times\mcA^L}
                \eta\,
                \left[
                    K^{(L)}_\omega(\rho;\dee\eta,\dee\pmb a)
                    -\varepsilon\nu_\omega(\dee\eta,\dee\pmb a)
                \right]
            \ge0.
        \]
        Thus \(\Phi^{(L)}_\omega(\rho)\ge\varepsilon\tau_\omega\) for every \(\rho\in\states\).
        Consequently
        \[
             \Psi_\omega(X):=
             \frac{\Phi^{(L)}_\omega(X)-\varepsilon\tr{X}\tau_\omega}{1-\varepsilon}
        \]
        is positive and trace-preserving.

        Now we prove that $\Psi$ is trace-norm contracting.
        To see this, let \(H\) be Hermitian and write \(H=H_+-H_-\), where \(H_+\) and \(H_-\) are its positive and negative parts. 
        Thus
        \[
            H_+,H_-\ge0,
            \qquad
            \unorm{H}_1=\tr{H_+}+\tr{H_-}.
        \]
        Using positivity, trace preservation, and the triangle inequality, we obtain
        \begin{align*}
            \unorm{\Psi_\omega(H)}_1
            &\le
            \unorm{\Psi_\omega(H_+)}_1
            +
            \unorm{\Psi_\omega(H_-)}_1\\
            &=
            \tr{\Psi_\omega(H_+)}
            +
            \tr{\Psi_\omega(H_-)}\\
            &=
            \tr{H_+}+\tr{H_-}
            =
            \unorm{H}_1.
        \end{align*}
        Here we used that the trace norm of a positive semidefinite matrix equals its trace.
        If, in addition, \(\tr{H}=0\), then the definition of \(\Psi_\omega\) gives
        \[
            \Phi^{(L)}_\omega(H)
            =
            \varepsilon\tr{H}\tau_\omega
            +(1-\varepsilon)\Psi_\omega(H)
            =
            (1-\varepsilon)\Psi_\omega(H).
        \]
        Consequently,
        \[
            \unorm{\Phi^{(L)}_\omega(H)}_1
            \le
            (1-\varepsilon)\unorm{H}_1.
        \]
        Taking the supremum over nonzero trace-zero Hermitian matrices
        \(H\) yields
        \[
            \kappa(\Phi^{(L)}_\omega)\le1-\varepsilon.
        \]

        We next extend this estimate to arbitrary product lengths.
        $\kappa$ is submultiplicative in the sense that \(\kappa(A\circ B)\le\kappa(A)\kappa(B)\) (see \cite{Replacement}).
        Write \(n=mL+r\), where \(m=\lfloor n/L\rfloor\) and \(0\le r<L\).
        Since we work on an invariant full-measure set, the block estimate holds at every shifted environment:
        \[
            \kappa\left(\Phi^{(L)}_{\theta^{jL}\omega}\right)
            \le1-\varepsilon,
            \qquad j\ge0.
        \]
        Splitting the product into \(m\) blocks of length \(L\) and a remaining block of length \(r\), submultiplicativity gives
        \[
            \kappa(\Phi^{(n)}_\omega)
            \le
            \kappa\left(\Phi^{(r)}_{\theta^{mL}\omega}\right)
            \prod_{j=0}^{m-1}
                \kappa\left(\Phi^{(L)}_{\theta^{jL}\omega}\right)
            \le
            (1-\varepsilon)^m
            =
            (1-\varepsilon)^{\lfloor n/L\rfloor}.
        \]
        Here \(\Phi^{(0)}\) denotes the identity map, and an empty product is interpreted as \(1\).
        In particular,
        \[
            \mbE_{\pr}\kappa(\Phi_\omega^{(n)})
            \le (1-\varepsilon)^{\lfloor n/L\rfloor}
            \longrightarrow0.
        \]
        By \Cref{prop:unified:common}, there exists a measurable dynamically stationary state $\s:\Omega\to\states$, unique up to $\pr$-null sets.
        Intersecting with all shifts of its stationarity set, for every $\rho$ and $n$ simultaneously we obtain
        \begin{align*}
             \unorm{\Phi^{(n)}_\omega(\rho)-\s(\theta^n\omega)}_1
             &=\unorm{\Phi^{(n)}_\omega(\rho)-\Phi^{(n)}_\omega(\s(\omega))}_1\\
             &\le2(1-\varepsilon)^{\lfloor n/L\rfloor}
             \le C\lambda^n.
        \end{align*}
        This proves \eqref{eq:doeblin-prop-forgetting-pointwise}.
    \end{proof}

    We now give a simple model where \Cref{assumption:block-doeblin} is verified. 

    \begin{remark}
        The fixed-posterior word maps in \Cref{subsec:unified:replacement} also verify \Cref{assumption:block-doeblin} when
        \[
             \mcT^{(L)}_{u;\omega}(X)=\tr{F_{u;\omega}X}\tau_{u;\omega},
             \quad F_{u;\omega}\ge c_{u;\omega}I,
             \quad c_\omega:=\sum_{u\in U}c_{u;\omega}\ge\varepsilon>0,
        \]
        on a full-measure set, for a finite set $U\subseteq\mcA^L$ and measurable states $\tau_{u;\omega}$ and nonnegative $c_{u;\omega}$.
        Indeed, retaining these words in \eqref{eq:augmented-L-step-kernel} gives a minorization by $\sum_{u\in U}c_{u;\omega}\delta_{(\tau_{u;\omega},u)}$; divide by $c_\omega$ to obtain $\nu_\omega$ on that set and define it as a fixed probability measure elsewhere.
        Exact state-independent replacement words correspond to $F_{u;\omega}=c_{u;\omega}I$.
    \end{remark}


\subsection{Mixing Properties of the Skew-Product}
\label{section:mixing_of_skew}


    The culmination of \Cref{assumption1,assumption_mixing_base,assumption:forgetting} is that these environmental and instrumental-level assumptions give the following mixing conditions for the annealed process. 

    The goal of this section is to obtain mixing bounds for the skew-product probability space $(\Omega\times\mcA^{\mbN},\mcF\otimes\Sigma, \overline{\mbQ}_{\s})$ under our standing assumptions. 
    The past and future \(\sigma\)-algebras \(\mcG_k\) and \(\mcG^k\), together with the corresponding
    mixing profile \(\overline{\alpha}_{\overline{\mbQ}_{\vartheta}}\),   were introduced above.
    We now establish the principal mixing estimate under the stationary annealed law \(\overline{\mbQ}_{\s}\).

    \begin{lemma}
    \label{lemma:alpha_mix_skew}
        Let \((\Omega,\mcF,\pr,\theta)\) satisfy \Cref{assumption1}, and suppose there exist a dynamically stationary state $\s$ and a deterministic bound $\beta_n(\eta)\le r_n\to0$ for every measurable state $\eta$. 
        Let \(\s\) be the unique dynamically stationary state , and let \(\overline\mbQ_{\s}\) be the corresponding annealed law on \(\Omega\times\mcA^{\mbN}\).
        Then, for every \(n\ge2\) and every \(m\in\{1,\ldots,n-1\}\),
        \[
            \overline\alpha_{\overline\mbQ_{\s}}(n)
            \le
            8\,\alpha_{\pr}(n-m)
            +
            r_{n-1}
            +
            2r_m.
        \]
    \end{lemma}
        \begin{proof}
            Fix \(n\ge2\), \(m\in\{1,\ldots,n-1\}\), and \(k\in\mbN\).
            Let \( D\in\mcG_k\) and \(H\in\mcG^{k+n}\) be arbitrary.
            For a measurable set \(K\subseteq\Omega\times\mcA^{\mbN}\) and for $\omega\in\Omega$ we denote by $K_\omega$ the $\omega$-section of $\mcA^\mbN$, i.e. \( K_\omega = \{\bar a\in\mcA^{\mbN}: (\omega,\bar a)\in K\}\). 
            For \(\pmb a=(a_1,\ldots,a_k)\in\mcA^k\), put
            \[
                [\pmb a]
                :=
                A^{-1}_{1:k}(\{\pmb a\}).
            \]
            Since \(\mcA^k\) is finite and \(D\in \mcG_k= \mcF_k\otimes\sigma(A_1,\ldots,A_k)\), there exist sets \(E_{\pmb a}\in\mcF_k\), \(\pmb a\in\mcA^k\), such that
            \begin{equation}
            \label{eq:alpha-past-decomposition}
                D
                =
                \bigsqcup_{\pmb a\in\mcA^k}
                \left(E_{\pmb a}\times[\pmb a]\right).
            \end{equation}
        
            Now, for
            \[
                q_{k+n-1}:
                \Omega\times\mcA^{\mbN}
                \longrightarrow
                \Omega\times\mcA^{\mbN},
                \qquad
                q_{k+n-1}(\omega,\bar a)
                :=
                (\omega,\varsigma^{k+n-1}\bar a),
            \]
            we have that
            \[
                \mcF^{k+n}\otimes
                \sigma(A_{k+n},A_{k+n+1},\ldots)
                =
                \left\{
                    q_{k+n-1}^{-1}(G):
                    G\in\mcF^{k+n}\otimes\Sigma
                \right\},
            \]
            where we have used that \(\sigma(A_{k+n},A_{k+n+1},\ldots) = \varsigma^{-(k+n-1)}(\Sigma)\). 
            Therefore, for the given $H\in\mcG^{k+n}$ there exists \(\widehat H\in\mcF^{k+n}\otimes\Sigma\) such that
            \begin{equation}
            \label{eq:alpha-future-pullback}
                H=q_{k+n-1}^{-1}(\widehat H),
                \qquad
                H_\omega
                =
                \varsigma^{-(k+n-1)}(\widehat H_\omega).
            \end{equation}
            Take
            \begin{equation}
            \label{eq:alpha-future-effect}
                M_H(\omega)
                :=
                \mbQ^*_{\theta^{(k+n-1)}\omega}
                \left(\widehat H_\omega\right).
            \end{equation}  
            We claim that \(M_H\) is \(\mcF^{k+n}\)-measurable. 
            Indeed, if \(C\) is a finite cylinder, then \(\omega \mapsto \mbQ^*_{\theta^{(k+n-1)}\omega}(C)\) 
            depends only on the instruments at the global times  \(k+n,k+n+1,\ldots\) appearing in that cylinder. Hence it is \(\mcF^{k+n}\)-measurable.
            Now if $\widehat H = J \times C$ for some $J\in \mcF^{k+n}$ and $C\in\Sigma$ then $M_H(\omega) = 1_J(\omega)\mbQ^*_{\theta^{k+n-1}\omega}(C)$.
            Therefore, a standard monotone class argument ($\pi-\lambda$ theorem) gives the desired measurability of $\omega\mapsto M_H(\omega)$. 
        
            Similarly, when $\widehat H = J \times C$ as above, by definition of $\mbQ^*_\omega$ we also have that \(0_d\le M_H(\omega)\le\mbI_d\) and from \Cref{lem:disordered_k_m_lemma_1_2} we obtain 
            \begin{align}
                \mbQ^*_\omega(H_\omega)
                &=
                \left(\Phi_\omega^{(k+n-1)}\right)\adj
                \left(M_H(\omega)\right),
                \label{eq:alpha-tail-full}\\
                \mbQ^*_\omega
                \left([\pmb a]\cap H_\omega\right)
                &=
                \left(\mcT_{\pmb a;\omega}^{(k)}\right)\adj
                \left(
                    \left(
                        \Phi_{\theta^k\omega}^{(n-1)}
                    \right)\adj
                    \left(M_H(\omega)\right)
                \right),
                \qquad
                \pmb a\in\mcA^k
                \label{eq:alpha-tail-prefix}
            \end{align}
            for \(\pr\)-almost every \(\omega\). 
            Another monotone class argument shows that the identities in \cref{eq:alpha-tail-full,eq:alpha-tail-prefix} hold for all such $\widehat H \in \mcF^{k+n}\otimes \Sigma$ in \eqref{eq:alpha-future-pullback}.
        
            \medskip
        
            Define the positive semidefinite matrix
            \[
                X_D(\omega)
                :=
                \sum_{\pmb a\in\mcA^k}
                1_{E_{\pmb a}}(\omega)\,
                \mcT_{\pmb a;\omega}^{(k)}
                \left(\s(\omega)\right),
            \]
            and let
            \[
                U(\omega)
                :=
                \tr{X_D(\omega)}
                =
                \mbQ_{\s(\omega);\omega}(D_\omega).
            \]
            Then \(0\le U\le1\). Define
            \[
                \rho_D(\omega)
                :=
                \begin{cases}
                    \dfrac{X_D(\omega)}{U(\omega)},
                        & U(\omega)>0,\\[1em]
                    \rho_\ast,
                        & U(\omega)=0.
                \end{cases}
            \]
            The map \(\rho_D:\Omega\to\states\) is measurable because  \eqref{eq:alpha-past-decomposition} is a finite decomposition.
            Moreover, if \(U(\omega)=0\), then positivity implies \(X_D(\omega)=0_d\).
            Set
            \[
                V(\omega)
                :=
                \mbQ_{\s(\omega);\omega}(H_\omega).
            \]
            By \eqref{eq:alpha-tail-full} and dynamical stationarity,
            \begin{equation}
            \label{eq:alpha-V}
                V(\omega)
                =
                \inner{\s(\omega)}{\mbQ^*_{\omega}(H_\omega)}
                =
                \inner{\s(\omega)}{\left(\Phi^{(k+n-1)}_\omega\right)\adj (M_H(\omega))}
                =
                \inner{
                    \s(\theta^{k+n-1}\omega)
                }{
                    M_H(\omega)
                }.
            \end{equation}
            By \eqref{eq:alpha-past-decomposition} and \eqref{eq:alpha-tail-prefix},
            \begin{align}
                W(\omega)
                :=
                \mbQ_{\s(\omega);\omega}
                (D_\omega\cap H_\omega)
                =
                \inner{
                    \Phi_{\theta^k\omega}^{(n-1)}
                    \left(X_D(\omega)\right)
                }{
                    M_H(\omega)
                }
                =
                U(\omega)
                \inner{
                    \Phi_{\theta^k\omega}^{(n-1)}
                    \left(\rho_D(\omega)\right)
                }{
                    M_H(\omega)
                }.
                \label{eq:alpha-W}
            \end{align}
            In each application of forgetting to $\rho_D(\omega)$ after time $k$, change variables $\eta=\theta^k\omega$ and use the measurable random initial state $\eta\mapsto\rho_D(\theta^{-k}\eta)$. This justifies the bound even when $\rho_D$ depends on the environment.
            The last identity also holds on \(\{U=0\}\), because both sides
            vanish there.
        
            Since $\overline\mbQ_{\s}(\dee\omega, d\bar{a}) = \pr(\dee\omega)\mbQ_{\s(\omega);\omega}(\dee\bar{a})$ we get
            \[
                \overline\mbQ_{\s}(D)
                =
                \mbE_{\pr}[U],
                \qquad
                \overline\mbQ_{\s}(H)
                =
                \mbE_{\pr}[V],
                \qquad
                \overline\mbQ_{\s}(D\cap H)
                =
                \mbE_{\pr}[W].
            \]
            Consequently,
            \begin{align}
                \left|
                    \overline\mbQ_{\s}(D\cap H)
                    -
                    \overline\mbQ_{\s}(D)
                    \overline\mbQ_{\s}(H)
                \right|
                &=
                \left|
                    \mbE_{\pr}[W]
                    -
                    \mbE_{\pr}[U]
                    \mbE_{\pr}[V]
                \right|
                \notag
                \\
                &\le
                \left|
                    \mbE_{\pr}[W-UV]
                \right|
                +
                \left|
                    \mathrm{Cov}_{\pr}(U,V)
                \right|.
                \label{eq:alpha-two-terms}
            \end{align}
        
            Now we estimate the two summands appearing in \eqref{eq:alpha-two-terms}.
            For the first one, since \(0_d\le M_H(\omega)\le\mbI_d\), \eqref{eq:alpha-V}--\eqref{eq:alpha-W} give
            \begin{align*}
                \left|
                    \mbE_{\pr}[W-UV]
                \right|
                &\le
                \int_\Omega
                U(\omega)
                \left\|
                    \Phi_{\theta^k\omega}^{(n-1)}
                    \left(\rho_D(\omega)\right)
                    -
                    \s(\theta^{k+n-1}\omega)
                \right\|_1
                \,\dee\pr(\omega)
                \\
                &\le
                \int_\Omega
                \left\|
                    \Phi_{\theta^k\omega}^{(n-1)}
                    \left(\rho_D(\omega)\right)
                    -
                    \s(\theta^{k+n-1}\omega)
                \right\|_1
                \,\dee\pr(\omega)
                \le r_{n-1}
            \end{align*}
            where the last inequality uses the uniform forgetting hypothesis of the lemma, applied to the measurable state $\eta(\xi)=\rho_D(\theta^{-k}\xi)$ after the change of variables $\xi=\theta^k\omega$, and invariance of $\pr$. 
        
           Now for the second summand in \eqref{eq:alpha-two-terms}, consider the state \(\Phi_{\theta^{-m}\omega}^{(m)}(\rho_\ast)\). 
           This depends only on the instruments at times $-m+1, \ldots , 0$.
           Set
            \[
                N_D(\omega)
                :=
                \sum_{\pmb a\in\mcA^k}
                1_{E_{\pmb a}}(\omega)\,
                \left(
                    \mcT_{\pmb a;\omega}^{(k)}
                \right)\adj(\mbI_d).
            \]
            Since the sum selects, for each \(\omega\), a subset of the \(k\)-step outcomes, we have that
            \[
                0_d
                \le
                N_D(\omega)
                \le
                \sum_{\pmb a\in\mcA^k}
                \left(
                    \mcT_{\pmb a;\omega}^{(k)}
                \right)\adj(\mbI_d)
                =
                \left(
                    \Phi_\omega^{(k)}
                \right)\adj(\mbI_d)
                =
                \mbI_d.
            \]
            Here positivity follows from the fact that each map $\mcT^{(k)}_{\pmb a;\omega}$ is completely positive, and thus the adjoint is positive. 
            For the upper bound, we have used that the total collection of $k$-step effects sums to the identity, because $\Phi^{(k)}_\omega$ is trace preserving, and thus the adjoint is unital. 
            
            Moreover, by definition of $U(\omega)$ we also have that 
            \[
                U(\omega)
                =
                \inner{\s(\omega)}{N_D(\omega)}.
            \]
            Now define
            \[
                U'(\omega)
                :=
                \inner{\Phi_{\theta^{-m}\omega}^{(m)}(\rho_\ast)}{N_D(\omega)},
                \qquad
                V'(\omega)
                :=
                \inner{
                   \Phi_{\theta^{k+n-1-m}\omega}^{(m)}(\rho_\ast)
                }{
                    M_H(\omega)
                }.
            \]
            Then \(0\le U',V'\le1\).
        
            We have that \( U'\in\mcF_k\) since \(N_D\) depends only on the sets \(E_{\pmb a}\in\mcF_k\) and the instruments at times \(1,\ldots,k\) and as \( \Phi_{\theta^{-m}\omega}^{(m)}(\rho_\ast)\)\ depends only on the instruments at times at most \(0\).
            On the other hand, as \(M_H\) is \(\mcF^{k+n}\)-measurable, and as \( \Phi_{\theta^{k+n-1-m}\omega}^{(m)}(\rho_\ast)\) depends only on the instruments at times \(k+n-m, \ldots, k+n-1\) we see that $V'\in \mcF^{k+n-m}$. 
        
            By bilinearity of covariance,
            \begin{align*}
                \mathrm{Cov}_{\pr}(U,V)
                &=
                \mathrm{Cov}_{\pr}\left(U'+(U-U'),V\right)\\
                &=
                \mathrm{Cov}_{\pr}(U-U',V)
                +
                \mathrm{Cov}_{\pr}(U',V)\\
                &=
                \mathrm{Cov}_{\pr}(U-U',V)
                +
                \mathrm{Cov}_{\pr}(U',V')
                +
                \mathrm{Cov}_{\pr}(U',V-V').
                \end{align*}
            Thus,
            \[
                \left|\mathrm{Cov}_{\pr}(U,V)\right|
                \le
                \left|\mathrm{Cov}_{\pr}(U',V')\right|
                +
                \left|\mathrm{Cov}_{\pr}(U-U',V)\right|
                +
                \left|\mathrm{Cov}_{\pr}(U',V-V')\right|.
            \] 
            We now estimate these three terms. 
            First 
            \[
                \mathrm{Cov}_{\pr}(U-U',V)
                =
                \mbE_{\pr}\!\left[(U-U')V\right]
                -
                \mbE_{\pr}\!\left[U-U'\right]
                \mbE_{\pr}\!\left[V\right]
                =
                \mbE_{\pr}\!\left[
                    (U-U')\left(V-\mbE_{\pr}[V]\right)
                \right].
            \]
            Since \(0\le V\le1\), we have \( 0\le\mbE_{\pr}[V]\le1\) and \(\left|V-\mbE_{\pr}[V]\right|\le1\).
            Thus, using the uniform forgetting estimate, we get
            \begin{equation*}
                \left|
                    \mathrm{Cov}_{\pr}(U-U',V)
                \right|
                \le
                \mbE_{\pr}\!\left[
                    |U-U'|\,
                    \left|V-\mbE_{\pr}[V]\right|
                \right]
                \le
                \mbE_{\pr}[|U-U'|]
                \le r_m.
            \end{equation*}
            Similarly, we also see that 
            \[
                |\mathrm{Cov}_{\pr}(U',V-V')| \le r_m.
            \]
            Finally, since
            \[
                U'\in\mcF_k,
                \qquad
                V'\in\mcF^{k+n-m},
                \qquad
                \norm{U'}_{L^\infty(\pr)},
                \norm{V'}_{L^\infty(\pr)}
                \le1,
            \]
            \Cref{lemma:correlation_bounds}, applied with \(p=q=\infty\), gives
            \begin{align*}
                \left|\mathrm{Cov}_{\pr}(U',V')\right|
                \le
                8\,
                \alpha_{\pr}\!\left(\mcF_k,\mcF^{k+n-m}\right)
                \norm{U'}_{L^\infty(\pr)}
                \norm{V'}_{L^\infty(\pr)}
                \le
                8\,\alpha_{\pr}(n-m).
            \end{align*}
            
            Therefore, using all of these upper bounds for the summands in \Cref{eq:alpha-two-terms} we obtain that 
            \[
                 \left|
                    \overline\mbQ_{\s}(D\cap H)
                    -
                    \overline\mbQ_{\s}(D)
                    \overline\mbQ_{\s}(H)
                \right|
                \le
                r_{n-1} + 2r_m + 8\alpha_{\pr}(n-m)
            \]
            Taking the supremum over \(D\) and \(H\), and then over \(k\in\mbN\), proves
            \[
                \overline\alpha_{\overline\mbQ_{\s}}(n)
                \le
                8\,\alpha_{\pr}(n-m)
                +
                r_{n-1}
                +
                2r_m.
            \]
        \end{proof}

    The mixing bounds from \Cref{lemma:alpha_mix_skew}, together with \Cref{assumption_mixing_base,assumption:forgetting}, give summable mixing for bounded block observables under the stationary annealed law. We use this estimate in the functional proof in \Cref{section:universal_proof}. 


\section{Proofs of the stationary functional CLTs}
\label{section:universal_proof}


Having established the results above, we now prove \Cref{thm:aclt} under the dynamically stationary initial state $\s$, that is, under the law $\overline\mbQ_{\s}$.

Let $q,\ell$ and $f$ be as in \Cref{thm:aclt}, and write
\[
 \mu:=\mbE_{\overline\mbQ_{\s}}f(A_{1:\ell}),
 \qquad X_j:=f(A_{j:j+\ell-1})-\mu,
 \qquad S_n:=\sum_{j=1}^nX_j,
 \qquad S_0:=0,
\]
for simplicity. 

We first record the variance calculation used for the second-moment conclusion.

    \begin{lemma}
    \label{lemma:variance_limit}
        Let \((X_n)_{n\ge1}\) be a strictly stationary, centered sequence of real-valued random variables on a probability space \((\Omega,\mcF,\mu)\), with \(X_1\in L^2(\mu)\). Write
        \[
            \gamma(k)
            :=
            \mathrm{Cov}_\mu(X_1,X_{1+k})
            =
            \mbE_\mu[X_1X_{1+k}],
            \qquad k\ge0,
        \]
        and assume that
        \[
            \sum_{k=1}^\infty |\gamma(k)|<\infty.
        \]
        Define
        \[
            \sigma^2
            :=
            \gamma(0)+2\sum_{k=1}^\infty \gamma(k),
        \]
        which is well-defined by absolute summability. Then, 
        \[
            \mathrm{Var}_\mu\!\left(\frac{1}{\sqrt n}
                \sum_{i=1}^n X_i
            \right)
            =
            \frac1n\mathrm{Var}_\mu\left(\sum_{i=1}^n X_i\right)
            \longrightarrow
            \sigma^2
            \qquad\text{as }n\to\infty.
        \]
        In particular, \(\sigma^2\ge0\) and whenever \(\sigma^2=0\), we have
        \[
            \frac{1}{\sqrt n}\left(\sum_{i=1}^n X_i\right)\to0
            \qquad\text{in }L^2(\mu),
        \]
        and hence also in probability.
    \end{lemma}
        \begin{proof}
            Let 
            \[
                S_n:=\sum_{i=1}^n X_i.
            \]
            Since $(X_n)$ is centered and strictly stationary,
            \[
                \mathrm{Var}_\mu(S_n)
                = \mbE_\mu\left[S_n^2\right]
                = \sum_{i=1}^n\sum_{j=1}^n \mbE_\mu[X_i X_j]
                = \sum_{i=1}^n\sum_{j=1}^n \gamma(|i-j|).
            \]
            Grouping terms by lag $k:=|i-j|$, we obtain the standard variance decomposition
            \[
                \mathrm{Var}_\mu(S_n)
                = n\,\gamma(0)
                  + 2\sum_{k=1}^{n-1} (n-k)\,\gamma(k),
            \]
            so that
            \begin{equation}
            \label{eq:variance_scaled}
                \mathrm{Var}_\mu\!\left(\frac{S_n}{\sqrt{n}}\right)
                = \frac{1}{n}\mathrm{Var}_\mu(S_n)
                = \gamma(0)
                  + 2\sum_{k=1}^{n-1} \left(1 - \frac{k}{n}\right)\gamma(k).
            \end{equation}
            
            We now pass to the limit $n\to\infty$.
            Fix $\varepsilon>0$. By absolute summability of $(\gamma(k))_{k\ge1}$, there exists $K\in\mbN$ such that
            \[
                \sum_{k>K} |\gamma(k)| < \varepsilon.
            \]

            For $n\ge K+1$, split the sum in \eqref{eq:variance_scaled} as
            \[
                \sum_{k=1}^{n-1} \left(1 - \frac{k}{n}\right)\gamma(k)
                =
                \sum_{k=1}^{K} \left(1 - \frac{k}{n}\right)\gamma(k)
                +
                \sum_{k=K+1}^{n-1} \left(1 - \frac{k}{n}\right)\gamma(k).
            \]
            For the finite part,
            \[
                \lim_{n\to\infty}\sum_{k=1}^{K} \left(1 - \frac{k}{n}\right)\gamma(k)
                = \sum_{k=1}^{K} \gamma(k),
            \]
            since $(1-k/n)\to1$ for each fixed $k$ and the sum is finite.
            
            For the tail, we use $0\le 1 - k/n \le 1$ to obtain the bound
            \[
                \left|
                   \sum_{k=K+1}^{n-1} \left(1 - \frac{k}{n}\right)\gamma(k)
                \right|
                \le
                \sum_{k=K+1}^{\infty} |\gamma(k)|
                < \varepsilon
                \qquad\text{for all }n\ge K+1.
            \]
            Combining these estimates, we find
            \[
                \limsup_{n\to\infty}
                \left|
                   \sum_{k=1}^{n-1} \left(1 - \frac{k}{n}\right)\gamma(k)
                   - \sum_{k=1}^\infty \gamma(k)
                \right|
                \le 2\varepsilon.
            \]
            Since $\varepsilon>0$ was arbitrary, it follows that
            \[
                \lim_{n\to\infty}
                \sum_{k=1}^{n-1} \left(1 - \frac{k}{n}\right)\gamma(k)
                = \sum_{k=1}^\infty \gamma(k).
            \]
            
            Returning to \eqref{eq:variance_scaled}, we conclude that
            \[
                \lim_{n\to\infty}
                \mathrm{Var}_\mu\!\left(\frac{S_n}{\sqrt{n}}\right)
                = \gamma(0) + 2\sum_{k=1}^\infty \gamma(k)
                = \sigma^2,
            \]
            as claimed.
            If $\sigma^2=0$, then
            \[
                \mbE_\mu\!\left[\left|\frac{S_n}{\sqrt{n}}\right|^2\right]
                = \mathrm{Var}_\mu\!\left(\frac{S_n}{\sqrt{n}}\right)
                \longrightarrow 0,
            \]
            so $\frac{S_n}{\sqrt{n}}\to0$ in $L^2(\mu)$ and hence in probability.
        \end{proof}

    Below, we recall the functional central limit theorem of \cite[Theorem~1]{DMR94} with its vector extension expressed in our notation.
    
    \begin{theorem}[{ \cite[Theorem~1]{DMR94}}]
    \label{thm:DMR_FCLT}
        Let $(Y_j)_{j\in\mbZ}$ be a strictly stationary sequence of $\mbR^q$-valued random variables on a probability space $(\Xi,\mcH,\nu)$, with $\mbE_\nu Y_0=0$ and $\mbE_\nu\|Y_0\|^2<\infty$. For $n\ge0$, define
        \[
         a_n:=\sup_{\substack{D\in\sigma(Y_j:j\le0)\\
                               H\in\sigma(Y_j:j\ge n)}}
               |\nu(D\cap H)-\nu(D)\nu(H)|,
        \]
        and suppose $a_n\to0$.
        Set $a(t):=a_{\lfloor t\rfloor}$ for $t\ge0$, and for $0<u<1$ put
        \[
         a^{-1}(u):=\inf\{t\ge0:a(t)\le u\},
         \qquad
         Q(u):=\inf\{t\ge0:\nu(\|Y_0\|>t)\le u\}.
        \]
        If
        \begin{equation}
        \label{eq:DMR-quantile-condition}
         \int_0^1 a^{-1}(u)Q(u)^2\,\dee u<\infty,
        \end{equation}
        then
        \[
         \Gamma:=\sum_{h\in\mbZ}\mbE_\nu[Y_0Y_h^{\mathsf T}]
        \]
        converges absolutely entry by entry and defines a positive semidefinite matrix. Moreover,
        \[
         \left(\frac1{\sqrt n}\sum_{j=1}^{\lfloor nt\rfloor}Y_j\right)_{0\le t\le1}
         \underset{\nu}{\overset{\mathrm d}{\longrightarrow}}
         \left(\Gamma^{1/2}W(t)\right)_{0\le t\le1}
        \]
        in $D([0,1],\mbR^q)$ with the $J_1$ topology, where $W$ is standard $q$-dimensional Brownian motion.
        Singular and zero matrices $\Gamma$ are also allowed.
    \end{theorem}

    For bounded observables, summability of the mixing coefficients implies \eqref{eq:DMR-quantile-condition}. We record this consequence for sequences indexed by $\mbN$, together with the second-moment limit needed below.

    \begin{cor}
    \label{thm:FCLT_alpha}
        Let $(Y_j)_{j\ge1}$ be a strictly stationary sequence of centered $\mbR^q$-valued random variables on a probability space $(\Xi,\mcH,\nu)$. 
        Suppose that $\|Y_1\|\le M$ almost surely for some finite $M$, and define
        \[
             \alpha_{Y,\nu}(n)
             :=\sup_{k\ge1}\sup_{\substack{D\in\sigma(Y_1,\ldots,Y_k)\\
                                          H\in\sigma(Y_j:j\ge k+n)}}
                   |\nu(D\cap H)-\nu(D)\nu(H)|.
        \]
        If $\sum_{n\ge1}\alpha_{Y,\nu}(n)<\infty$, then the series
        \[
             \Gamma=\mbE_\nu[Y_1Y_1^{\mathsf T}]
             +\sum_{h\ge1}\left(
               \mbE_\nu[Y_1Y_{1+h}^{\mathsf T}]
              +\mbE_\nu[Y_{1+h}Y_1^{\mathsf T}]\right)
        \]
        converges absolutely entry by entry and is positive semidefinite. 
        With $T_n=\sum_{j=1}^nY_j$ and $T_0=0$,
        \[
             \left(n^{-1/2}T_{\lfloor nt\rfloor}\right)_{0\le t\le1}
             \underset{\nu}{\overset{\mathrm d}{\longrightarrow}}
             \left(\Gamma^{1/2}W(t)\right)_{0\le t\le1}
        \]
        in $D([0,1],\mbR^q)$ with the $J_1$ topology, where $W$ is standard $q$-dimensional Brownian motion. Singular and zero covariance matrices are allowed. Moreover,
        \[
            \frac1n\mbE_\nu[T_nT_n^{\mathsf T}]\longrightarrow\Gamma.
        \]
    \end{cor}
    \begin{proof}
        We first construct a two-sided stationary process with the same law on the positive indices. 
        This is the standard two-sided extension of a stationary sequence; see \cite{Kallenberg_2021}; hence we do not present all details. 
        For integers $i_1<\cdots<i_r$, choose an integer  $s\ge0$ with $i_1+s\ge1$, and prescribe
        \[
             \operatorname{Law}_{\widetilde\nu}
               (\widetilde Y_{i_1},\ldots,\widetilde Y_{i_r})
             :=\operatorname{Law}_{\nu}(Y_{i_1+s},\ldots,Y_{i_r+s}).
        \]
        Strict stationarity makes the right side independent of the choice of $s$. 
        These finite-dimensional laws are consistent: for any finite set of indices and any subset, one may use the same shift $s$ and then take the corresponding marginal of the original joint law.
        Kolmogorov's extension theorem therefore gives a probability measure $\widetilde\nu$ on $(\mbR^q)^{\mbZ}$ with its product Borel $\sigma$-algebra. 
        Its coordinate process $(\widetilde Y_j)_{j\in\mbZ}$ is strictly stationary, and
        \begin{equation}
        \label{eq:FCLT-extension-law}
         \operatorname{Law}_{\widetilde\nu}
               ((\widetilde Y_j)_{j\ge1})
         =\operatorname{Law}_{\nu}((Y_j)_{j\ge1}).
        \end{equation}
        
        We next verify equality of the mixing coefficients. For $m,n\ge1$, put
        \[
             \mcF_m^-:=\sigma(\widetilde Y_{-m+1},\ldots,\widetilde Y_0),
             \qquad
             \mcF^-:=\sigma(\widetilde Y_j:j\le0),
             \qquad
             \mcF_n^+:=\sigma(\widetilde Y_j:j\ge n).
        \]
        Write $\alpha_\lambda(\mcA,\mcB)$ for the strong mixing coefficient of two $\sigma$-algebras under the probability measure $\lambda$. 
        Since $\mcF_m^-\uparrow\mcF^-$, every $D\in\mcF^-$ admits $D_m\in\mcF_m^-$ with  $\widetilde\nu(D\mathbin\triangle D_m)\to0$.
        For example, take  $D_m=\{\mbE_{\widetilde\nu}[\mathbf1_D\mid\mcF_m^-]>1/2\}$ and use martingale convergence.
        For every $H\in\mcF_n^+$, replacing $D$ by $D_m$ changes $\widetilde\nu(D\cap H)-\widetilde\nu(D)\widetilde\nu(H)$ by at most $2\widetilde\nu(D\mathbin\triangle D_m)$.
        Consequently,
        \[
             \alpha_{\widetilde\nu}(\mcF^-,\mcF_n^+)
             =\sup_{m\ge1}
                   \alpha_{\widetilde\nu}(\mcF_m^-,\mcF_n^+).
        \]
        Shifting both blocks forward by $m$ and using \eqref{eq:FCLT-extension-law} gives
        \begin{align*}
             a_n
             :=\alpha_{\widetilde\nu}(\mcF^-,\mcF_n^+)
             =\sup_{m\ge1}
                \alpha_\nu\left(\sigma(Y_1,\ldots,Y_m),
                                 \sigma(Y_j:j\ge m+n)\right)
             =\alpha_{Y,\nu}(n).
        \end{align*}
        In particular, the two-sided extension is strongly mixing.
        
        For the two-sided extension, let $a_h$, $a^{-1}$ and $Q$ be as in \Cref{thm:DMR_FCLT}. 
        Then $a_h=\alpha_{Y,\nu}(h)$ for $h\ge1$,  $a_0\le1$, and $Q(u)\le M$. Since $(a_h)$ is nonincreasing,
        \[
             a^{-1}(u)=\sum_{h\ge0}\mathbf1_{\{u<a_h\}},
             \qquad 0<u<1.
        \]
        Tonelli's theorem therefore gives
        \[
             \int_0^1 a^{-1}(u)Q(u)^2\,\dee u
             =\sum_{h\ge0}\int_0^{a_h}Q(u)^2\,\dee u
             \le M^2\left(1+\sum_{h\ge1}\alpha_{Y,\nu}(h)\right)<\infty.
        \]
        Thus \Cref{thm:DMR_FCLT} gives the functional limit for $(\widetilde Y_j)$ under $\widetilde\nu$.
        By \eqref{eq:FCLT-extension-law}, each normalized partial-sum path has the same distribution as the corresponding path for $(Y_j)$  under $\nu$, so this is also the asserted functional limit.
        By stationarity, the negative-lag terms in its covariance series are the transposes of the corresponding positive-lag terms, so that its matrix  $\Gamma$ is precisely the one displayed here.
        
        Now by \Cref{lemma:correlation_bounds} with $p=q=\infty$, every entry of $C_h=\mbE_\nu[Y_1Y_{1+h}^{\mathsf T}]$ is bounded in absolute value by $8M^2\alpha_{Y,\nu}(h)$ for $h\ge1$.
        The covariance series is therefore absolutely convergent. Applying \Cref{lemma:variance_limit} to each coordinate and to sums of two coordinates, then polarizing, gives
        \[
            \frac1n\mbE_\nu[T_nT_n^{\mathsf T}]\longrightarrow\Gamma.
        \]
        In particular, $\Gamma$ is positive semidefinite, since each matrix on the left is positive semidefinite.
    \end{proof}
        
    We are now ready to prove \Cref{thm:aclt}, which we restate below. 

    \annealedclt*

    \begin{proof}
        We apply \Cref{thm:FCLT_alpha} to the process $X_j=f(A_{j:j+\ell-1})-\mu$ under the stationary annealed law $\overline\mbQ_{\s}$, in the probability space $(\Omega\times\mcA^\mbN,\mcF\otimes\Sigma,\overline\mbQ_{\s})$.  
        It is centered by the definition of $\mu$.
        By the first assertion of \Cref{thm:LLN}, the skew shift $\tau$ preserves $\overline\mbQ_{\s}$. 
        Since
        \[
            X_j=X_1\circ\tau^{j-1},
            \qquad j\ge1,
        \]
        the sequence $(X_j)_{j\ge1}$ is strictly stationary under $\overline\mbQ_{\s}$.
        
        Let $\alpha_{X,\overline\mbQ_{\s}}(n)$ denote the coefficient defined in \Cref{thm:FCLT_alpha} for this sequence. For $n\ge\ell$ and every $k\ge1$,
        \[
         \sigma(X_1,\ldots,X_k)\subseteq\mcG_{k+\ell-1},
         \qquad
         \sigma(X_j:j\ge k+n)\subseteq\mcG^{k+n}.
        \]
        Here the sigma-algebras labeled with $\mcG$ are the joint instrument-outcomes sigma-algebras from \ref{eq:mixing+joint}, whereas the sigma-algebras on the left contain only recorded-outcome information. 
        Consequently,
        \[
         \alpha_{X,\overline\mbQ_{\s}}(n)
         \le\overline\alpha_{\overline\mbQ_{\s}}(n-\ell+1).
        \]
        This coefficient is unchanged if $X$ is regarded as a process on the annealed outcome space alone.
        
        By \Cref{lemma:alpha_mix_skew}, for $n\ge2$ and $m=\lfloor n/2\rfloor$,
        \[
         \overline\alpha_{\overline\mbQ_{\s}}(n)
         \le8\alpha_{\pr}(\lceil n/2\rceil)
               +r_{n-1}+2r_{\lfloor n/2\rfloor}.
        \]
        Here $\alpha_{\pr}$ is the mixing coefficient of the instrument-map process defined in \eqref{eq:rho_n__and_alpha_n_disorder}. 
        \Cref{assumption_mixing_base,assumption:forgetting} make the right side summable. 
        Each term of the original rate sequence occurs at most twice in the half-indexed sum, so monotonicity of $r_n$ is unnecessary. 
        It follows that $\sum_n\alpha_{X,\overline\mbQ_{\s}}(n)<\infty$.
        
        We may therefore apply \Cref{thm:FCLT_alpha} with $\nu=\overline\mbQ_{\s}$, $Y_j=X_j$, and $M=\max_w\|f(w)-\mu\|$.
        For $h\ge0$, write
        \[
            C_h:=\mbE_{\overline\mbQ_{\s}}
                     [X_1X_{1+h}^{\mathsf T}].
        \]
        The covariance matrix supplied by that corollary is
        \[
            C_0+\sum_{h\ge1}(C_h+C_h^{\mathsf T})
            =\Sigma_f,
        \]
        which is precisely \eqref{eq:universal-covariance}.
        Thus the corollary gives the asserted functional CLT, absolute convergence of this series, and positive semidefiniteness of $\Sigma_f$.

        To relate its second-moment conclusion to the expressions in the
        theorem, expand $S_nS_n^{\mathsf T}$ and group the terms by their lag.
        By stationarity, for $j=i+h>i$,
        \[
            \mbE_{\overline\mbQ_{\s}}[X_iX_j^{\mathsf T}]=C_h,
            \qquad
            \mbE_{\overline\mbQ_{\s}}[X_jX_i^{\mathsf T}]
            =C_h^{\mathsf T}.
        \]
        There are $n$ diagonal terms and, for each $1\le h<n$,
        exactly $n-h$ pairs with $j-i=h$. Hence
        \[
            \begin{aligned}
                \frac1n\mbE_{\overline\mbQ_{\s}}[S_nS_n^{\mathsf T}]
                =
                \frac1n\sum_{i,j=1}^n
                    \mbE_{\overline\mbQ_{\s}}[X_iX_j^{\mathsf T}]
                =
                C_0+\sum_{h=1}^{n-1}
                     \left(1-\frac hn\right)(C_h+C_h^{\mathsf T}).
            \end{aligned}
        \]
        Absolute summability of the covariance series allows us to pass to the limit in this expression, giving
        \begin{equation}
        \label{eq:universal-stationary-moments}
            \frac1n\mbE_{\overline\mbQ_{\s}}[S_nS_n^{\mathsf T}]
            \longrightarrow\Sigma_f,
        \end{equation}
        in agreement with the final conclusion of \Cref{thm:FCLT_alpha}.
        Since $\mbE_{\overline\mbQ_{\s}}S_n=0$, we also have
        \[
            \Cov_{\overline\mbQ_{\s}}(S_n)
            =
            \mbE_{\overline\mbQ_{\s}}[S_nS_n^{\mathsf T}],
        \]
        so both remaining limits in \Cref{thm:aclt} follow.
        
        If $\Sigma_f=0$, then
        \[
            \mbE_{\overline\mbQ_{\s}}
                \left[\left\|\frac{S_n}{\sqrt n}\right\|^2\right]
            =
            \str\left(
                \frac1n\mbE_{\overline\mbQ_{\s}}
                    [S_nS_n^{\mathsf T}]
            \right)
            \longrightarrow0,
        \]
        which also proves the $L^2$ conclusion in the remark following that theorem.
    \end{proof}


\section{Universal CLT}


Having established the stationary functional CLT in \Cref{thm:aclt} we now extend it to any initial state.
For a random initial state $\vartheta:\Omega\to\states$, write
\[
\overline\mbQ^{\mathrm{out}}_\vartheta
=\overline\mbQ_\vartheta\circ\pi_{\mcA^\mbN}^{-1}
\]
for its annealed outcome law.
We retain $\mbQ_{\rho;\omega}$ for the quenched outcome law at the fixed environment $\omega$.
For a word $\pmb a=(a_i,\ldots,a_j)\in\mcA^{j-i+1}$, write
\[
    \mcT^{[i:j]}_{\pmb a;\omega}
    :=
    \mcT_{a_j;\theta^j\omega}\circ\cdots\circ
    \mcT_{a_i;\theta^i\omega}.
\]
An empty product is treated as the identity. 
In particular, put $\Phi^{(0)}_\omega=\mathrm{id}$, $\beta_0(\vartheta)=\mbE_{\pr}\unorm{\vartheta-\s}_1$ and $r_0=2$.

For a measurable space $(E,\mcE)$, let $\mcP(E,\mcE)$ denote the set of probability measures on $(E,\mcE)$.
The total variation distance is the function
\[
    d_{\mathrm{TV}}:
    \mcP(E,\mcE)\times\mcP(E,\mcE)\longrightarrow[0,1],
    \qquad
    d_{\mathrm{TV}}(P,Q)
    :=\sup_{B\in\mcE}|P(B)-Q(B)|.
\]
Thus $P$ and $Q$ are probability measures on the same measurable space, and the supremum is over its measurable events.

The following lemma bounds the total variation distance between the quenched laws $\mbQ_{\s(\omega);\omega}$ and $\mbQ_{\vartheta(\omega);\omega}$, after projection onto the first $N$ outcomes, for any measurable random initial state $\vartheta:\Omega\to\states$.
Both laws use the same environment realization $\omega$.
The bound holds more generally for any two initial states. 

\begin{lemma}
\label{univ:lem:povm}
    There is a $\theta$-invariant set $\Omega_*$ of full $\pr$-measure such that, for every $\omega\in\Omega_*$,  every integer $N\ge1$, and all $\rho,\eta\in\states$,
    \[
        d_{\mathrm{TV}}\!\left(
            \mbQ_{\rho;\omega}\circ A_{1:N}^{-1},
            \mbQ_{\eta;\omega}\circ A_{1:N}^{-1}
        \right)
        \le \tfrac12\unorm{\rho-\eta}_1.
    \]
    Here both probability measures are defined using the same environment $\omega$.
\end{lemma}
\begin{proof}
    Let $\Omega_0$ be a measurable set of full $\pr$-measure on which the POVM $\mbQ^*_\omega$ introduced in \eqref{eq:duality} is defined and that identity holds for every $\rho\in\states$ and $S\in\Sigma$.
    Recall that
    \[
        \mbQ^*_\omega(\mcA^\mbN)=\mbI_d.
    \]
    Put
    \[
        \Omega_*:=\bigcap_{m\in\mbZ}\theta^{-m}\Omega_0.
    \]
    Since $\theta$ is invertible and preserves $\pr$, $\Omega_*$ is $\theta$-invariant and has full $\pr$-measure.

    Fix $\omega\in\Omega_*$, $N\ge1$, and $\rho,\eta\in\states$.
    For $B\subseteq\mcA^N$, put
    \[
        S_B:=A_{1:N}^{-1}(B)\in\Sigma.
    \]
    The positivity and normalization of the POVM give
    \[
        0_d\le\mbQ^*_\omega(S_B)
        \le\mbQ^*_\omega(\mcA^\mbN)=\mbI_d.
    \]
    By \eqref{eq:duality},
    \[
        \begin{aligned}
            \left(\mbQ_{\rho;\omega}\circ A_{1:N}^{-1}\right)(B)
            -
            \left(\mbQ_{\eta;\omega}\circ A_{1:N}^{-1}\right)(B)
            =\mbQ_{\rho;\omega}(S_B)-\mbQ_{\eta;\omega}(S_B)
            =\inner{\rho-\eta}{\mbQ^*_\omega(S_B)}.
        \end{aligned}
    \]

    Set $D:=\rho-\eta$, and write $D=D_+-D_-$ for its positive and negative parts. 
    Since $D$ is Hermitian with trace zero,
    \[
        D_\pm\ge0,
        \qquad
        \str D_+=\str D_-=\tfrac12\unorm{D}_1.
    \]
    Furthermore,
    \[
        \begin{aligned}
            0
            &\le\inner{D_\pm}{\mbQ^*_\omega(S_B)}\\
            &=\str\left(
                D_\pm^{1/2}\mbQ^*_\omega(S_B)D_\pm^{1/2}
              \right)
            \le\str D_\pm
            =\tfrac12\unorm{D}_1.
        \end{aligned}
    \]
    Both quantities therefore belong to $[0,\unorm{D}_1/2]$, so their difference satisfies
    \[
        \left|\inner{D}{\mbQ^*_\omega(S_B)}\right|
        \le\tfrac12\unorm{D}_1
        =\tfrac12\unorm{\rho-\eta}_1.
    \]
    Taking the supremum over $B\subseteq\mcA^N$ proves the claim.
    All these estimates hold for every $N,\rho,\eta$ on the same set $\Omega_*$.
\end{proof}

We now sum over the first $k$ outcomes and apply \Cref{univ:lem:povm} to the remaining outcome laws.
Averaging over the environment and using dynamical stationarity gives the following tail estimate.

\begin{prop}
\label{univ:prop:tail}
    Let $\s$ be a dynamically stationary state. 
    Then, for every measurable initial state $\vartheta:\Omega\to\states$ and all integers $k\ge0$ and $N\ge1$,
    \begin{equation}
    \label{univ:eq:tailtv}
        \TV\!\left(
            \overline\mbQ^{\mathrm{out}}_\vartheta
                \circ A_{k+1:k+N}^{-1},
            \overline\mbQ^{\mathrm{out}}_{\s}
                \circ A_{1:N}^{-1}
        \right)
        \le \tfrac12\beta_k(\vartheta).
    \end{equation}
    The same estimate holds for the re-indexed infinite tail:
    \begin{equation}
    \label{univ:eq:tailtv_infinite}
        \TV\!\left(
            \overline\mbQ^{\mathrm{out}}_\vartheta
                \circ(\varsigma^k)^{-1},
            \overline\mbQ^{\mathrm{out}}_{\s}
        \right)
        \le \tfrac12\beta_k(\vartheta),
    \end{equation}
    where $\varsigma$ is the left shift on $\mcA^\mbN$.
\end{prop}
\begin{proof}
    Work on the invariant full-measure set from \Cref{univ:lem:povm}, intersected, if necessary, with the invariant full-measure set where the instrument identities hold at every shift.
    For $k=0$, take $\mcA^0$ to consist of the empty word and put $\mcT^{(0)}_{\emptyset;\omega}=\mathrm{id}$.
    
    Fix $k\ge0$ and $N\ge1$.
    For a retained word $\pmb b\in\mcA^N$, summing over all  prefixes $\pmb a\in\mcA^k$ gives
    \begin{align*}
        \mbQ_{\vartheta(\omega);\omega}
            (A_{k+1:k+N}=\pmb b)
        &\quad=
        \sum_{\pmb a\in\mcA^k}
        \str\!\left(
                \mcT^{(N)}_{\pmb b;\theta^k\omega}
        \left(
            \mcT^{(k)}_{\pmb a;\omega}
                (\vartheta(\omega))
        \right)
            \right)\\
        &\quad=
        \str\!\left(
            \mcT^{(N)}_{\pmb b;\theta^k\omega}
            \left(
                \Phi^{(k)}_\omega(\vartheta(\omega))
            \right)
        \right)\\
        &\quad=
        \mbQ_{\Phi^{(k)}_\omega(\vartheta(\omega));\theta^k\omega}
            (A_{1:N}=\pmb b).
    \end{align*}
    Here we used
    \[
        \sum_{\pmb a\in\mcA^k}
            \mcT^{(k)}_{\pmb a;\omega}
        =\Phi^{(k)}_\omega.
    \]
    Thus summing over the first $k$ outcomes leaves the non-selective state $\Phi^{(k)}_\omega(\vartheta(\omega))$, with the remaining instruments starting at $\theta^{k+1}\omega$.
    By summing over $\pmb b\in B$, the same identity holds for every event $B\subseteq\mcA^N$.

    Define probability measures on $\mcA^N$ by
    \[
        P_\omega = \mbQ_{\Phi^{(k)}_\omega(\vartheta(\omega));\theta^k\omega}
            \circ A_{1:N}^{-1}
        \qquad\text{ and }\qquad 
        R_\omega =
        \mbQ_{\s(\theta^k\omega);\theta^k\omega}
            \circ A_{1:N}^{-1}
    \]
    Both laws use the same environment $\theta^k\omega$.
    Consequently, \Cref{univ:lem:povm} gives, for every $B\subseteq\mcA^N$,
    \[
        |P_\omega(B)-R_\omega(B)|
        \le
        \tfrac12
        \unorm{
            \Phi^{(k)}_\omega(\vartheta(\omega))
            -\s(\theta^k\omega)
        }_1.
    \]
    But $P_\omega = \mbQ_{\Phi^{(k)}_\omega(\vartheta(\omega));\theta^k\omega}\circ A_{1:N}^{-1}$ yields that 
    \[
        P_\omega(B)
        =
        \mbQ_{\Phi^{(k)}_\omega(\vartheta(\omega));\theta^k\omega}
            (A_{1:N}\in B)
        =
        \mbQ_{\vartheta(\omega);\omega}
            (A_{k+1:k+N}\in B).
    \]
    Integrating with respect to $\pr$ therefore gives
    \[
        \int_\Omega P_\omega(B)\,\dee\pr(\omega)
        =
        \overline\mbQ^{\mathrm{out}}_\vartheta
            (A_{k+1:k+N}\in B).
    \]
    Moreover, invariance of $\pr$ under $\theta^k$ gives
    \begin{align*}
        \int_\Omega R_\omega(B)\,\dee\pr(\omega)
        =
        \int_\Omega
            \mbQ_{\s(\omega');\omega'}
                (A_{1:N}\in B)\,\dee\pr(\omega')
        =
        \overline\mbQ^{\mathrm{out}}_{\s}
            (A_{1:N}\in B).
    \end{align*}
    It follows that
    \begin{align*}
        \left|
            \overline\mbQ^{\mathrm{out}}_\vartheta
                (A_{k+1:k+N}\in B)
            -
            \overline\mbQ^{\mathrm{out}}_{\s}
                (A_{1:N}\in B)
        \right|
        &\quad\le
        \int_\Omega
            |P_\omega(B)-R_\omega(B)|\,\dee\pr(\omega)\\
        &\quad\le
        \tfrac12
        \int_\Omega
            \unorm{
                \Phi^{(k)}_\omega(\vartheta(\omega))
                -\s(\theta^k\omega)
            }_1\,\dee\pr(\omega)
        =\tfrac12\beta_k(\vartheta).
    \end{align*}
    Taking the supremum over $B\subseteq\mcA^N$ proves
    \eqref{univ:eq:tailtv}.
    
    Now to prove the assertion in \ref{univ:eq:tailtv_infinite}, take
    \[
        P:=
        \overline\mbQ^{\mathrm{out}}_\vartheta
            \circ(\varsigma^k)^{-1},
        \qquad
        \text{ and }
        \qquad
        R:=\overline\mbQ^{\mathrm{out}}_{\s}.
    \]
    Then let 
    \[
        m:=\tfrac12(P+R).
    \]
    These are probability measures on $(\mcA^\mbN,\Sigma)$.
    Let $\Sigma_N:=\sigma(A_1,\ldots,A_N)$.
    Since $A_{1:N}\circ\varsigma^k=A_{k+1:k+N}$, the finite-record estimate \eqref{univ:eq:tailtv} gives
    \[
        |P(B)-R(B)|
        \le\tfrac12\beta_k(\vartheta),
        \qquad B\in\Sigma_N,\quad N\ge1.
    \]

   Fix an arbitrary event $S\in\Sigma$. 
   On the probability space $(\mcA^\mbN,\Sigma,m)$, define
    \[
        g_N:=\mbE_m[\mathbf1_S\mid\Sigma_N],
        \qquad N\ge1.
    \]
    By the properties of conditional expectation, $g_N$ is $\Sigma_N$-measurable and satisfies $0\le g_N\le1$ $m$-almost surely. 
    In particular, $g_N$ is integrable. 
    Moreover, since
    $\Sigma_N\subseteq\Sigma_{N+1}$, the tower property of conditional expectation gives
    \[
        \begin{aligned}
            \mbE_m[g_{N+1}\mid\Sigma_N]
            &=
            \mbE_m\!\left[
                \mbE_m[\mathbf1_S\mid\Sigma_{N+1}]
                \,\middle|\,\Sigma_N
            \right]\\
            &=
            \mbE_m[\mathbf1_S\mid\Sigma_N]
            =g_N
            \qquad m\text{-almost surely}.
        \end{aligned}
    \]
    Thus $(g_N)_{N\ge1}$ is a bounded, and hence uniformly integrable, martingale with respect to the filtration $(\Sigma_N)_{N\ge1}$ under $m$.

    By L\'evy's upward theorem (or Doob’s martingale convergence theorem) we get the following convergence:
    \[
        g_N
        \longrightarrow
        \mbE_m\!\left[
            \mathbf1_S
            \,\middle|\,
            \sigma\!\left(\bigcup_{N\ge1}\Sigma_N\right)
        \right]
        =
        \mbE_m[\mathbf1_S\mid\Sigma]
        =
        \mathbf1_S, \qquad N \to \infty,
    \]
    $m$-almost surely and in $L^1(m)$, where the last  equality follows from $S\in\Sigma$.

    Choose a $\Sigma_N$-measurable version of each $g_N$ and set
    \[
        S_N:=\{g_N>1/2\}\in\Sigma_N.
    \]
    On $S\setminus S_N$, we have $\mathbf1_S=1$ and $g_N\le1/2$, whereas on $S_N\setminus S$, we have $\mathbf1_S=0$ and $g_N>1/2$. 
    Hence
    \[
        \mathbf1_{S\mathbin\triangle S_N}
        \le2|g_N-\mathbf1_S|.
    \]
    Integrating with respect to the measure $m$ now yields
    \[
        m(S\mathbin\triangle S_N)
        \le2\mbE_m|g_N-\mathbf1_S|
        \longrightarrow0.
    \]
    Consequently, the finite-record estimate and the identity $P+R=2m$ imply
    \[
        \begin{aligned}
            |P(S)-R(S)|
            &\le
            |P(S_N)-R(S_N)|
            +|P(S)-P(S_N)|
            +|R(S)-R(S_N)|\\
            &\le
            \tfrac12\beta_k(\vartheta)
            +P(S\mathbin\triangle S_N)
            +R(S\mathbin\triangle S_N)\\
            &=
            \tfrac12\beta_k(\vartheta)
            +2m(S\mathbin\triangle S_N).
        \end{aligned}
    \]
    There, we have used that \(|\mu(S)-\mu(S_N)| \le\mu(S\mathbin\triangle S_N)\) for any probability measure $\mu$.
    Letting $N\to\infty$ gives
    \[
        |P(S)-R(S)|
        \le\tfrac12\beta_k(\vartheta).
    \]
    Since $S\in\Sigma$ was arbitrary, taking the supremum
    over $S$ proves
    \[
        \TV(P,R)
        =
        \sup_{S\in\Sigma}|P(S)-R(S)|
        \le\tfrac12\beta_k(\vartheta).
    \]
    This concludes the proof of \eqref{univ:eq:tailtv_infinite}.
\end{proof}  

We now use the tail estimate to transfer the stationary functional CLT to arbitrary initial states.
If $\beta_k(\vartheta)\to0$, discarding a prefix of length $k_n\to\infty$ with $k_n=o(\sqrt n)$ makes the total-variation distance from the stationary outcome law tend to zero, while boundedness of the observable ensures that this removal changes the normalized partial-sum path by a uniformly vanishing amount.

\begin{prop}
\label{univ:prop:transfer}
    Let $q,\ell\in\mbN$ and let $f:\mcA^\ell\to\mbR^q$ be bounded and measurable.
    Define
    \[
        \mu:=\mbE_{\overline\mbQ^{\mathrm{out}}_{\s}}[f(A_{1:\ell})],
        \qquad
        X_j:=f(A_{j:j+\ell-1})-\mu,
        \qquad
        S_m:=\sum_{j=1}^{m}X_j,
        \qquad S_0:=0.
    \]
    Write
    \[
        \mcS_n(t)
        :=\frac{S_{\lfloor nt\rfloor}}{\sqrt n},
        \qquad t\in[0,1],
    \]
    for the normalized partial-sum path.
    Suppose that, under $\overline\mbQ^{\mathrm{out}}_{\s}$, the processes $\mcS_n$ converge in distribution to a Gaussian process $G$ in $D([0,1],\mbR^q)$ equipped with the $J_1$ topology.
    Then, for every measurable initial state $\vartheta:\Omega\to\states$ satisfying $\beta_k(\vartheta)\to0$, the same convergence holds under $\overline\mbQ^{\mathrm{out}}_\vartheta$.
    In particular, the centering $\mu$ and the covariance of the limiting Gaussian process are unchanged.
\end{prop}
\begin{proof}
    Since $f$ is bounded,
    \[
        B:=\sup_{w\in\mcA^\ell}
            \unorm{f(w)-\mu}<\infty.
    \]
    Thus $\unorm{X_j}\le B$ for every outcome record and every $j\ge1$.
    For integers $k,m\ge0$, define, on the same infinite outcome record,
    \[
        S_m^{[k]}
        :=\sum_{j=k+1}^{k+m}X_j,
        \qquad
        \mcS_n^{[k]}(t)
        :=\frac{S_{\lfloor nt\rfloor}^{[k]}}{\sqrt n}.
    \]
    Empty sums are automatically understood to be zero.

    The sets of starting indices $\{1,\ldots,m\}$ and $\{k+1,\ldots,k+m\}$ have symmetric difference of cardinality $2\min(k,m)$. 
    Since $S_m^{[k]}=S_{m+k}-S_k$, we have, for all integers $k,m\ge0$,
    \[
        S_m-S_m^{[k]}
        =S_k-(S_{m+k}-S_m)
        =\sum_{j=1}^{k}X_j
         -\sum_{j=m+1}^{m+k}X_j.
    \]
    Using $\unorm{X_j}\le B$ and the triangle inequality, we obtain
    \[
        \unorm{S_m-S_m^{[k]}}
        \le
        \sum_{j=1}^{k}\unorm{X_j}
        +\sum_{j=m+1}^{m+k}\unorm{X_j}
        \le2Bk.
    \]
    In particular, for every $n\ge1$,
        \begin{equation}
        \label{univ:eq:boundary}
            \sup_{t\in[0,1]}
                \unorm{
                    \mcS_n(t)-\mcS_n^{[k]}(t)
                }
            =
            \frac1{\sqrt n}
            \max_{0\le m\le n}
                \unorm{S_m-S_m^{[k]}}
            \le\frac{2Bk}{\sqrt n}.
        \end{equation}
    These bounds hold for every outcome record.

    Since $X_j\circ\varsigma^k=X_{j+k}$, we also have
    \[
        \mcS_n^{[k]}
        =\mcS_n\circ\varsigma^k.
    \]
    Both $\mcS_n$ and $\mcS_n^{[k]}$ are measurable maps from $(\mcA^\mbN,\Sigma)$ into $\mcD:= D([0,1],\mbR^q)$ equipped with the Borel
    $\sigma$-algebra of the $J_1$ topology, \(\mcB(\mcD)\).
    Define
    \[
        P:=
        \overline\mbQ^{\mathrm{out}}_\vartheta
            \circ(\varsigma^k)^{-1},\qquad\text{ and }\qquad
        \qquad
        R:=\overline\mbQ^{\mathrm{out}}_{\s}.
    \]
    These are probability measures on the outcome-record space $(\mcA^\mbN,\Sigma)$.
    The measurable map $\mcS_n$ assigns an entire normalized partial-sum path to each outcome record.
    Since $\mcS_n^{[k]}=\mcS_n\circ\varsigma^k$, its pushforward measures satisfy
    \[
        \begin{aligned}
            P\circ\mcS_n^{-1}
            =
            \overline\mbQ^{\mathrm{out}}_\vartheta
                \circ(\mcS_n^{[k]})^{-1},\quad
            R\circ\mcS_n^{-1}
            =
            \overline\mbQ^{\mathrm{out}}_{\s}
                \circ\mcS_n^{-1}.
        \end{aligned}
    \]
    Thus these are probability measures on $(\mcD,\mcB(\mcD))$: respectively, the law of the shifted path under $\vartheta$ and the law of the unshifted path under $\s$.

    For every $C\in\mcB(\mcD)$, measurability gives $\mcS_n^{-1}(C)\in\Sigma$. 
    Therefore,
    \[
        \begin{aligned}
            \TV\!\left(
                P\circ\mcS_n^{-1},
                R\circ\mcS_n^{-1}
            \right)
            &=
            \sup_{C\in\mcB(\mcD)}
            \left|
                P(\mcS_n^{-1}(C))
                -R(\mcS_n^{-1}(C))
            \right|\\
            &\le
            \sup_{E\in\Sigma}|P(E)-R(E)|\\
            &=\TV(P,R)
            \le\tfrac12\beta_k(\vartheta),
        \end{aligned}
    \]
    where the last inequality follows from \eqref{univ:eq:tailtv_infinite}.

    Let $\Lambda$ denote the set of strictly increasing continuous bijections of $[0,1]$ onto itself. 
    Set
    \[
        d_{J_1}(x,y)
        :=
        1\wedge
        \inf_{\lambda\in\Lambda}
        \max\left\{
            \sup_{t\in[0,1]}|\lambda(t)-t|,
            \sup_{t\in[0,1]}\unorm{x(t)-y(\lambda(t))}
        \right\}.
    \]
    The untruncated expression is the usual Skorokhod  metric; see \cite[Section~12]{Billingsley_1999}, with absolute value replaced by the Euclidean norm for $\mbR^q$-valued paths. 
    Truncation at one preserves the induced topology, so $d_{J_1}$ is a bounded metric generating the $J_1$ topology.
    Taking $\lambda(t)=t$ gives
    \[
        d_{J_1}(x,y)
        \le
        1\wedge\sup_{t\in[0,1]}\unorm{x(t)-y(t)}
        \le
        \sup_{t\in[0,1]}\unorm{x(t)-y(t)}.
    \]

    Now, let $h:D([0,1],\mbR^q)\to\mbR$ satisfy
    \[
        \sup_x|h(x)|\le1,
        \qquad
        |h(x)-h(y)|\le d_{J_1}(x,y).
    \]
    By \eqref{univ:eq:boundary},
    \[
        \begin{aligned}
            \left|
                \mbE_{\overline\mbQ^{\mathrm{out}}_\vartheta}
                    [h(\mcS_n)]
                -
                \mbE_{\overline\mbQ^{\mathrm{out}}_\vartheta}
                    [h(\mcS_n^{[k]})]
            \right|
            \le
            \mbE_{\overline\mbQ^{\mathrm{out}}_\vartheta}
                [d_{J_1}(\mcS_n,\mcS_n^{[k]})]
            \le\frac{2Bk}{\sqrt n}.
        \end{aligned}
    \]

    Moreover, $h\circ\mcS_n$ is a measurable function on the outcome-record space with absolute value at most one. 
    Hence
    \[
        \begin{aligned}
            \left|
                \mbE_{\overline\mbQ^{\mathrm{out}}_\vartheta}
                    h(\mcS_n^{[k]})
                -
                \mbE_{\overline\mbQ^{\mathrm{out}}_{\s}}
                    h(\mcS_n)
            \right|
            &=
            \left|
                \int h\circ\mcS_n\,\dee P
                -
                \int h\circ\mcS_n\,\dee R
            \right|\\
            &\le 2\TV(P,R)
            \le \beta_k(\vartheta).
        \end{aligned}
    \]
    Applying this inequality to the two path laws above yields
    \[
        \left|
            \mbE_{\overline\mbQ^{\mathrm{out}}_\vartheta} h(\mcS_n^{[k]})
            -
            \mbE_{\overline\mbQ^{\mathrm{out}}_{\s}}h(\mcS_n)
        \right|
        \le\beta_k(\vartheta).
    \]
    Combining the last two estimates gives
    \begin{equation}
    \label{univ:eq:BL}
        \left|
            \mbE_{\overline\mbQ^{\mathrm{out}}_\vartheta} h(\mcS_n)
            -
            \mbE_{\overline\mbQ^{\mathrm{out}}_{\s}}h(\mcS_n)
        \right|
        \le
        \frac{2Bk}{\sqrt n}+\beta_k(\vartheta).
    \end{equation}

    To express this estimate in terms of probability laws, define
    \[
        \mcH:=
        \left\{
            h:\mcD\to\mbR:
            \sup_{x\in\mcD}|h(x)|\le1,\quad
            |h(x)-h(y)|\le d_{J_1}(x,y)
            \text{ for all }x,y\in\mcD
        \right\}.
    \]
    For probability measures $\nu,\eta$ on $(\mcD,\mcB(\mcD))$, their bounded Lipschitz distance is
    \[
        d_{\mathrm{BL}}(\nu,\eta)
        :=
        \sup_{h\in\mcH}
        \left|
            \int_{\mcD}h\,\dee\nu
            -\int_{\mcD}h\,\dee\eta
        \right|.
    \]
    Thus this distance compares expectations uniformly over the bounded Lipschitz test functions used above.
    
    Let
    \[
        \nu_n^\vartheta
        :=
        \overline\mbQ^{\mathrm{out}}_\vartheta
            \circ\mcS_n^{-1},
        \qquad
        \nu_n^{\s}
        :=
        \overline\mbQ^{\mathrm{out}}_{\s}
            \circ\mcS_n^{-1}
    \]
    be the two laws of the normalized partial-sum path.
    Since the right-hand side of \eqref{univ:eq:BL} does not depend on $h$, taking the supremum over $h\in\mcH$ gives
    \[
        d_{\mathrm{BL}}(\nu_n^\vartheta,\nu_n^{\s})
        \le
        \frac{2Bk}{\sqrt n}+\beta_k(\vartheta).
    \]
    
    Choose integers $k_n$ such that
    \[
        k_n\longrightarrow\infty,
        \qquad
        \frac{k_n}{\sqrt n}\longrightarrow0;
    \]
    for example, $k_n=\lfloor n^{1/4}\rfloor$.
    The first condition ensures $\beta_{k_n}(\vartheta)\to0$, while the second ensures that the deterministic boundary error $2Bk_n/\sqrt n$ tends to zero. 
    Consequently,
    \[
        d_{\mathrm{BL}}(\nu_n^\vartheta,\nu_n^{\s})
        \longrightarrow0.
    \]
    
    On a separable metric space, the bounded Lipschitz distance metrizes weak convergence: for Borel probability measures $\nu_n,\nu$, $\nu_n$ converges weakly to $\nu$ if and only if $d_{\mathrm{BL}}(\nu_n,\nu)\to 0$; see \cite[Theorem~11.3.3]{Dudley_2002}.
    This applies to $\mcD$ equipped with the metric $d_{J_1}$.
    
    Let $\nu_G$ denote the law of the Gaussian limit $G$. 
    The assumed stationary functional CLT gives the weak convergence of $\nu_n^{\s}$ to $\nu_G$, and therefore
    \[
        d_{\mathrm{BL}}(\nu_n^{\s},\nu_G)
        \longrightarrow0.
    \]
    Thus,
    \[
        d_{\mathrm{BL}}(\nu_n^\vartheta,\nu_G)
        \le
        d_{\mathrm{BL}}(\nu_n^\vartheta,\nu_n^{\s})
        +d_{\mathrm{BL}}(\nu_n^{\s},\nu_G)\\
        \longrightarrow0.
    \]
    Applying the same characterization of weak convergence once more proves
    \[
        \left(
            \frac{S_{\lfloor nt\rfloor}}{\sqrt n}
        \right)_{0\le t\le1}
        \underset{\overline\mbQ^{\mathrm{out}}_\vartheta}
            {\overset{\mathrm d}{\longrightarrow}}
        G.
    \]
    The centering remains $\mu$, and the limiting Gaussian process is the same as in the stationary functional CLT.
\end{proof}

For the functional limit, the transfer proposition requires only $\beta_k(\vartheta)\to0$ in addition to the stationary CLT. The second-moment assertions in \ref{second_moment_universal_clt} require a separate comparison. 

\begin{prop}
\label{univ:prop:moments}
    Suppose \Cref{assumption1} holds and $\s$ is a dynamically stationary state. 
    Assume the uniform forgetting bound $\beta_n(\eta)\le r_n\to0$ for all measurable initial states $\eta:\Omega\to\states$, and $\sum_{n\ge0}\beta_n(\vartheta)<\infty$ for the given measurable initial state $\vartheta:\Omega\to\states$. 
    For every bounded scalar $g:\mcA^\ell\to\mbR$ with $\mbE_{\overline\mbQ^{\mathrm{out}}_{\s}} g(A_{1:\ell})=0$, put $X_i=g(A_{i:i+\ell-1})$ and $S_n=\sum_{i=1}^nX_i$. Then
    \[
     \mbE_{ \overline\mbQ^{\mathrm{out}}_\vartheta} S_n=O(1),
     \qquad \text{ and }\qquad
     \frac1n\left(\mbE_{ \overline\mbQ^{\mathrm{out}}_\vartheta} S_n^2-\mbE_{ \overline\mbQ^{\mathrm{out}}_{\s}} S_n^2\right)\longrightarrow0.
    \]
    Consequently, any stationary second-moment limit is also the second-moment and covariance limit under $\overline\mbQ_\vartheta$ after division by $n$.
\end{prop}
\begin{proof}
    Let $g:\mcA^\ell\to\mbR$ be a bounded scalar function with  $\mbE_{\s}g(A_{1:\ell})=0$, and define
    \[
        B:=\max_{w\in\mcA^\ell}|g(w)|.
    \]
    For the random initial state $\vartheta$, put $b_i:=\beta_{i-1}(\vartheta)$, so $\sum_{i\ge1}b_i<\infty$.
    Throughout this proof, we retain the shorthand $\mbE_\vartheta :=\mbE_{\overline\mbQ^{\mathrm{out}}_\vartheta}$ and $\mbE_{\s} :=\mbE_{\overline\mbQ^{\mathrm{out}}_{\s}}$ for expectations under the corresponding annealed outcome laws.
    Set
    \[
        \Delta_{i-1}(\omega)
        :=
        \Phi^{(i-1)}_\omega(\vartheta(\omega))
        -
        \s(\theta^{i-1}\omega).
    \]
    For $j\ge1$, let
    \[
        H_j(\omega)
        :=
        \sum_{w\in\mcA^\ell}
            g(w)
            \bigl(\mcT^{[j:j+\ell-1]}_{w;\omega}\bigr)^*
            (\mbI_d),
    \]
    where the star denotes the Hilbert--Schmidt adjoint.
    This is the Hermitian operator representing $g$ on the block starting at $j$.
    The block effects are positive and sum to $\mbI_d$, so
    \[
        -B\mbI_d\le H_j(\omega)\le B\mbI_d,
        \qquad
        \unorm{H_j(\omega)}_\infty\le B.
    \]
    Summing over the first $i-1$ outcomes gives
    \[
        \mbE_\vartheta X_i-\mbE_{\s}X_i
        =
        \mbE_{\pr}
            \str\bigl(H_i\Delta_{i-1}\bigr).
    \]
    Consequently,
    \[
        |\mbE_\vartheta X_i-\mbE_{\s}X_i|
        \le
        B\mbE_{\pr}\unorm{\Delta_{i-1}}_1
        =
        Bb_i.
    \]
    Stationarity of the annealed outcome law under $\s$ gives $\mbE_{\s}X_i=0$ for every $i$. 
    Hence
    \[
        |\mbE_\vartheta S_n|
        \le
        B\sum_{i=1}^n b_i
        \le
        B\sum_{i\ge1}b_i<\infty.
    \]
    The upper bound is finite and independent of $n$.
    Thus $\mbE_\vartheta S_n=O(1)$, proving the first assertion.

    For $0\le j-i<\ell$, the product $X_iX_j$ is a single finite-block observable bounded in absolute value by $B^2$.
    Applying the same argument to this observable shows that
    \[
        \left|
            \mbE_\vartheta[X_iX_j]
            -
            \mbE_{\s}[X_iX_j]
        \right|
        \le B^2b_i.
    \]

    Since $X_iX_j=X_jX_i$, the diagonal terms appear once in the second-moment difference, while the upper-triangular terms appear twice. For each fixed $i$, the overlapping upper-triangular terms have
    \[
        i+1\le j\le\min(n,i+\ell-1).
    \]
    Using the preceding bound, we get that the quantity 
    \[
        \frac1n\sum_{i=1}^n
            \left|
                \mbE_\vartheta[X_i^2]
                -
                \mbE_{\s}[X_i^2]
            \right|
        +
        \frac2n\sum_{i=1}^n
        \sum_{j=i+1}^{\min(n,i+\ell-1)}
            \left|
                \mbE_\vartheta[X_iX_j]
                -
                \mbE_{\s}[X_iX_j]
            \right|
    \]
    is bounded above by 
    \[
        \frac{B^2}{n}\sum_{i=1}^n b_i
        +
        \frac{2(\ell-1)B^2}{n}\sum_{i=1}^n b_i
        =
        \frac{(2\ell-1)B^2}{n}\sum_{i=1}^n b_i
        \le
        \frac{2\ell B^2}{n}\sum_{i\ge1}b_i.
    \]
    But $\frac{2\ell B^2}{n}\sum_{i\ge1}b_i\to 0$ as $n\to\infty$. 

    For the nonoverlapping terms, define at time $i+\ell-1$ the Hermitian matrix
    \[
        D_i(\omega)
        :=
        \sum_{w\in\mcA^\ell}
            g(w)
            \mcT^{[i:i+\ell-1]}_{w;\omega}
                (\Delta_{i-1}(\omega)),
    \]
    and take $u_i:=\str D_i$. 
    For any Hermitian $D$, the Jordan decomposition and trace preservation of the sum of the block operations give
    \[
        \begin{aligned}
            \sum_w\unorm{\mcT_w(D)}_1
            &\le
            \sum_w
                \left(
                    \str\mcT_w(D_+)
                    +
                    \str\mcT_w(D_-)
                \right)\\
            &=
            \sum_w\str\mcT_w(|D|)
            =
            \unorm D_1,
        \end{aligned}
    \]
    where $\mcT_w$ abbreviates  $\mcT^{[i:i+\ell-1]}_{w;\omega}$.
    Hence
    \[
        \unorm{D_i}_1
        \le B\unorm{\Delta_{i-1}}_1
        \le2B,
        \qquad
        \mbE_{\pr}\unorm{D_i}_1\le Bb_i.
    \]

    Fix $j\ge i+\ell$ and put $h:=j-i-\ell$.
    Define
    \[
        v_j(\omega)
        :=
        \str\bigl(
            H_j(\omega)\s(\theta^{j-1}\omega)
        \bigr)
    \]
    and
    \[
        R_{ij}
        :=
        \mbE_{\pr}\str\left[
            H_j
            \left(
                \Phi^{(h)}_{\theta^{i+\ell-1}\omega}(D_i)
                -
                u_i\s(\theta^{j-1}\omega)
            \right)
        \right].
    \]
    Born's rule now gives the exact difference
    \begin{align}
    \label{univ:eq:pairdifference}
        \mbE_\vartheta[X_iX_j]
        -
        \mbE_{\s}[X_iX_j]
        &=
        \mbE_{\pr}\str\bigl(
            H_j
            \Phi^{(h)}_{\theta^{i+\ell-1}\omega}(D_i)
        \bigr)
        \nonumber\\
        &=
        \mbE_{\pr}[u_i v_j]+R_{ij}.
    \end{align}
    The number of discarded channels is exactly $h$: their times are $i+\ell,\ldots,j-1$.
    The product is the identity when $h=0$.

    Since $|u_i|\le\unorm{D_i}_1$ and positive trace-preserving maps contract the trace norm of
    Hermitian matrices,
    \[
        \mbE_{\pr}
        \unorm{
            \Phi^{(h)}_{\theta^{i+\ell-1}\omega}(D_i)
            -
            u_i\s(\theta^{j-1}\omega)
        }_1
        \qquad\le
        2\mbE_{\pr}\unorm{D_i}_1
        \le2Bb_i.
    \]

    A second bound follows by writing $D_i=D_i^+-D_i^-$ and setting
    \[
        c_i^\pm:=\str D_i^\pm,
        \qquad
        \eta_i^\pm:=
        \begin{cases}
            D_i^\pm/c_i^\pm, & c_i^\pm>0,\\
            \rho_*, & c_i^\pm=0,
        \end{cases}
    \]
    where $\rho_*\in\states$ is a fixed state.
    Positive and negative parts are continuous functions of a Hermitian matrix, so $\eta_i^\pm$ are measurable states.
    We have
    \[
        \Phi^{(h)}_{\theta^{i+\ell-1}\omega}(D_i)
        -
        u_i\s(\theta^{j-1}\omega)
        =
        c_i^+
        \left(
            \Phi^{(h)}_{\theta^{i+\ell-1}\omega}(\eta_i^+)
            -
            \s(\theta^{j-1}\omega)
        \right)
        -
        c_i^-
        \left(
            \Phi^{(h)}_{\theta^{i+\ell-1}\omega}(\eta_i^-)
            -
            \s(\theta^{j-1}\omega)
        \right).
    \]
    
    Each random mass satisfies $c_i^\pm\le\unorm{D_i}_1\le2B$ pointwise.
    Change variables $\xi=\theta^{i+\ell-1}\omega$ and apply uniform forgetting to the measurable states
    \[
        \xi\longmapsto
        \eta_i^\pm\bigl(\theta^{-(i+\ell-1)}\xi\bigr).
    \]
    Since $\theta$ preserves $\pr$ and $j-1=i+\ell-1+h$, this gives
    \[
        \mbE_{\pr}
        \unorm{
            \Phi^{(h)}_{\theta^{i+\ell-1}\omega}(\eta_i^\pm)
            -
            \s(\theta^{j-1}\omega)
        }_1
        \le r_h.
    \]
    Bounding the random masses before taking expectations therefore yields
    \[
        \mbE_{\pr}
        \unorm{
            \Phi^{(h)}_{\theta^{i+\ell-1}\omega}(D_i)
            -
            u_i\s(\theta^{j-1}\omega)
        }_1
        \le4Br_h.
    \]
    Together with $\unorm{H_j}_\infty\le B$, the two bounds imply
    \begin{equation}
    \label{univ:eq:remainder}
        |R_{ij}|
        \le
        \min\bigl(2B^2b_i,\,4B^2r_h\bigr).
    \end{equation}
    These estimates remain valid at $h=0$ with $r_0=2$.

    For $n,i\ge1$, put
    \[
        a_{n,i}:=\frac1n\sum_{j=i+\ell}^n|R_{ij}|,
    \]
    with empty sums understood to be zero.
    By \eqref{univ:eq:remainder},
    \[
        0\le a_{n,i}
        \le
        \frac{4B^2}{n}\sum_{h=0}^{n-1}r_h
        \longrightarrow0,
    \]
    because $r_h\to0$.
    Moreover,
    \[
        a_{n,i}\le2B^2b_i,
        \qquad
        \sum_{i\ge1}b_i<\infty.
    \]
    Dominated convergence for the series over $i$ gives
    \[
        \frac1n
        \sum_{i=1}^n\sum_{j=i+\ell}^n|R_{ij}|
        =
        \sum_{i\ge1}a_{n,i}
        \longrightarrow0.
    \]

    It remains to handle $\mbE_{\pr}[u_i v_j]$, where $u_i$ may depend on the entire environment.
    Define the bounded measurable function
    \[
        v(\omega)
        :=
        \str\bigl(H_1(\omega)\s(\omega)\bigr).
    \]
    The definition of the block operators gives $H_j=H_1\circ\theta^{j-1}$, and hence
    \[
        v_j=v\circ\theta^{j-1},
        \qquad
        |v_j|\le B,
        \qquad
        \mbE_{\pr}v=\mbE_{\s}X_1=0.
    \]
    The mean ergodic theorem for the unitary operator $U\psi:=\psi\circ\theta$ on $L^2(\pr)$ gives convergence of the Cesaro averages to the orthogonal projection onto the invariant functions.
    Ergodicity makes this projection the constant $\mbE_{\pr}v=0$. 
    Thus
    \[
        \frac1n\sum_{j=1}^n v_j
        \longrightarrow0
        \qquad\text{in }L^2(\pr),
    \]
    and hence also in $L^1(\pr)$.

    For each fixed $i$, removing the first $i+\ell-1$ terms preserves this convergence.
    Since $|u_i|\le2B$, it follows that
    \[
            \left|
                \frac1n
                \sum_{j=i+\ell}^n
                    \mbE_{\pr}[u_i v_j]
            \right|
            =
            \left|
                \mbE_{\pr}\left[
                    u_i\frac1n\sum_{j=i+\ell}^n v_j
                \right]
            \right|
            \le
            2B\mbE_{\pr}
                \left|
                    \frac1n\sum_{j=i+\ell}^n v_j
                \right|
            \to0.
    \]
    Also, for every $n,i\ge1$,
    \[
        \left|
            \frac1n
            \sum_{j=i+\ell}^n\mbE_{\pr}[u_i v_j]
        \right|
        \le
        B\mbE_{\pr}|u_i|
        \le B^2b_i.
    \]
    Another dominated convergence argument over $i$ yields
    \[
        \frac1n
        \sum_{i=1}^n\sum_{j=i+\ell}^n
            \mbE_{\pr}[u_i v_j]
        \longrightarrow0.
    \]

    To combine these estimates, write
    \[
        \delta_{ij}
        :=
        \mbE_\vartheta[X_iX_j]
        -
        \mbE_{\s}[X_iX_j].
    \]
    Since the observables are scalar,
    \[
        \frac1n
        \bigl(
            \mbE_\vartheta S_n^2-\mbE_{\s}S_n^2
        \bigr)
        =
        \frac1n\sum_{i=1}^n\delta_{ii}
        +
        \frac2n\sum_{i=1}^n\sum_{j=i+1}^n\delta_{ij}.
    \]
    The diagonal and overlapping contributions tend to zero by the earlier estimate. 
    The nonoverlapping contribution tends to zero by \eqref{univ:eq:pairdifference} and the two dominated convergence arguments above.
    This proves the second-moment comparison.
    
    Finally,
    \[
            \frac1n
            \left(
                \operatorname{Var}_{\overline\mbQ_\vartheta}(S_n)
                -
                \mbE_{\s}S_n^2
            \right)
            =
            \frac1n
            \left(
                \mbE_\vartheta S_n^2-\mbE_{\s}S_n^2
            \right)
            -
            \frac{(\mbE_\vartheta S_n)^2}{n}
            \longrightarrow0,
    \]
    since $\mbE_\vartheta S_n=O(1)$.
    This proves the variance assertion whenever the stationary second-moment limit exists.
\end{proof}

Now we are ready to prove \Cref{thm:universal-functional}. 

\univannealedclt*

\begin{proof}
    Let $\vartheta:\Omega\to\states$ be any measurable random initial state.
    By \Cref{thm:aclt}, the functional CLT holds under $\overline\mbQ_{\s}$ with
    centering $\mu$ and covariance $\Sigma_f$.
    \Cref{assumption:forgetting} gives
    \[
     \beta_k(\vartheta)\le r_k\longrightarrow0,
     \qquad
     \sum_{k\ge0}\beta_k(\vartheta)
     \le2+\sum_{k\ge1}r_k<\infty.
    \]
    Thus \Cref{univ:prop:transfer} gives the asserted functional limit under
    $\overline\mbQ_\vartheta$, with the same $\mu$ and $\Sigma_f$.
    
    To prove the moment conclusions, apply \Cref{univ:prop:moments} to each
    coordinate of $f-\mu$ and to sums of two coordinates. Polarization yields
    \[
     \mbE_{\overline\mbQ_\vartheta}S_n=O(1),
     \qquad
     \frac1n\left(
       \mbE_{\overline\mbQ_\vartheta}[S_nS_n^{\mathsf T}]
       -\mbE_{\overline\mbQ_{\s}}[S_nS_n^{\mathsf T}]
     \right)\longrightarrow0.
    \]
    Together with \eqref{eq:universal-stationary-moments}, this gives
    \[
     \frac1n\sum_{i,j=1}^n
       \mbE_{\overline\mbQ_\vartheta}[X_iX_j^{\mathsf T}]
     =\frac1n\mbE_{\overline\mbQ_\vartheta}[S_nS_n^{\mathsf T}]
     \longrightarrow\Sigma_f.
    \]
    Moreover,
    \[
     \frac1n\Cov_{\overline\mbQ_\vartheta}(S_n)
     =\frac1n\mbE_{\overline\mbQ_\vartheta}[S_nS_n^{\mathsf T}]
      -\frac1n
        (\mbE_{\overline\mbQ_\vartheta}S_n)
        (\mbE_{\overline\mbQ_\vartheta}S_n)^{\mathsf T}
     \longrightarrow\Sigma_f,
    \]
    because the second term tends to zero. If $\Sigma_f=0$, taking the trace
    of the second-moment limit also gives $S_n/\sqrt n\to0$ in
    $L^2(\overline\mbQ_\vartheta)$.
\end{proof}

Then \Cref{thm:universal-clt-A} is immediate. 

\universalcltA*

\begin{proof}
    Apply \Cref{thm:universal-functional} with $q=1$, $\ell=m$ and $f(u)=\mathbf1_{\{u=\pmb b\}}$. Then $\mu=\mu_{\pmb b}$, $X_j=\mathbf1_{C_{\pmb b}}\circ\varsigma^{j-1}-\mu_{\pmb b}$, and the scalar covariance series \eqref{eq:universal-covariance} equals $\Sigma_{\pmb b}^2$ defined in the corollary.
    Evaluation at $t=1$ is continuous for the $J_1$ topology, since every allowed time change fixes $1$. 
    The continuous mapping theorem gives the scalar limit, because $W(1)$ is standard normal. 
    The zero-variance conclusion follows from the second-moment assertion of \Cref{thm:universal-functional}.
\end{proof}


\section*{Acknowledgments}
\addcontentsline{toc}{section}{Acknowledgments}


    The author acknowledges support from Villum Fonden Grant No. 25452 and Grant No. 60842, as well as QMATH Center of Excellence Grant No. 10059. 
    Part of this work was also supported by the Danish e-infrastructure Consortium (DeiC) grant 5260-00014B. 
    The author also thanks Prof. Albert H. Werner for valuable discussions. 
    We also thank the two anonymous reviewers of the initial submission for their input, which helped substantially improve the presentation of this paper. 


\section*{Data Availability Statement}
\noindent
No data was used in this research.

\section*{Competing Interests}
\noindent
The authors have no competing interests to declare that are relevant to the content of this article. 
\begin{appendix}


\section{Auxiliary Lemmas}
\label{section:appn}


    Throughout we work with the finite outcome alphabet $\mcA$, equipped with the discrete $\sigma$-algebra $2^\mcA$, and the state space $\states$ with its Borel $\sigma$-algebra $\borel{\states}$.
    The product space $\mcA\times\states$ is equipped with the product $\sigma$-algebra $2^\mcA\otimes\borel{\states}$, which coincides with the Borel $\sigma$-algebra of the product topology.

    We recall the projective action of an instrument component.
    For each $\omega\in\Omega$, $a\in\mcA$ and $\rho\in\states$ we set
    \begin{equation}
    \label{eq:appendix-projective-action}
        \mcT_{a;\omega}\proj\rho
        :=
        \begin{cases}
            \dfrac{\mcT_{a;\omega}(\rho)}{\tr{\mcT_{a;\omega}(\rho)}}, 
                & \text{if }\tr{\mcT_{a;\omega}(\rho)}>0,\\[1ex]
            \rho_\star, 
                & \text{if }\tr{\mcT_{a;\omega}(\rho)}=0,
        \end{cases}
    \end{equation}
    where $\rho_\star\in\states$ is some fixed reference state.
    The choice of $\rho_\star$ on the zero-trace set does not affect any of the outcome distributions.

    \begin{lemma}
    \label{lem:proj-action-jointly-measurable}
        Fix \(a\in\mcA\). Then the map
        \[
            \Omega\times\states \ni (\omega,\rho)
            \longmapsto \mcT_{a;\omega}\proj \rho \in \states
        \]
        is \((\mcF\otimes\borel{\states},\borel{\states})\)-measurable.
    \end{lemma}
        \begin{proof}
            Fix \(a\in\mcA\). Since \(\mcH=\mbC^d\) is finite-dimensional, the space \(\mcB(\mcH)\) is a finite-dimensional complex vector space. 
            Choose a basis \(E_1,\dots,E_N\) of \(\mcB(\mcH)\), where \(N=d^2\), and let \(\ell_1,\dots,\ell_N:\mcB(\mcH)\to\mbC\) be the corresponding coordinate
            functionals, so that
            \[
                \rho=\sum_{r=1}^N \ell_r(\rho)\,E_r,
                \qquad \rho\in\mcB(\mcH).
            \]
            Each \(\ell_r\) is linear and therefore continuous.
        
            For \((\omega,\rho)\in\Omega\times\states\), linearity of \(\mcT_{a;\omega}\) gives
            \[
                \mcT_{a;\omega}(\rho)
                =
                \sum_{r=1}^N \ell_r(\rho)\,\mcT_{a;\omega}(E_r).
            \]
            For each \(r\), the map \(\omega\mapsto \mcT_{a;\omega}(E_r)\) is \((\mcF,\borel{\mcB(\mcH)})\)-measurable by the standing measurability assumption on the instrument, and the map \(\rho\mapsto \ell_r(\rho)\) is \((\borel{\states},\borel{\mbC})\)-measurable. Hence
            \[
                (\omega,\rho)\longmapsto \ell_r(\rho)\,\mcT_{a;\omega}(E_r)
            \]
            is \((\mcF\otimes\borel{\states},\borel{\mcB(\mcH)})\)-measurable for each \(r\). Summing over finitely many \(r\), we conclude that
            \[
                X(\omega,\rho):=\mcT_{a;\omega}(\rho)
            \]
            is \((\mcF\otimes\borel{\states},\borel{\mcB(\mcH)})\)-measurable as a map \(\Omega\times\states\to\mcB(\mcH)\).
        
            Since the trace map \(\tr{\,\cdot\,}:\mcB(\mcH)\to\mbR\) is continuous, the function
            \[
                t(\omega,\rho):=\tr{X(\omega,\rho)}
            \]
            is \((\mcF\otimes\borel{\states},\borel{\mbR})\)-measurable. 
            Let
            \[
                A:=\{(\omega,\rho)\in\Omega\times\states:\ t(\omega,\rho)>0\}.
            \]
            Then \(A\in\mcF\otimes\borel{\states}\).
        
            On \(A\), define
            \[
                Y(\omega,\rho):=\frac{X(\omega,\rho)}{t(\omega,\rho)}.
            \]
            The map
            \[
                f:\mcB(\mcH)\times(0,\infty)\to\mcB(\mcH),
                \qquad
                f(Z,s):=\frac{Z}{s},
            \]
            is continuous, hence Borel measurable. 
            Since
            \[
                (\omega,\rho)\longmapsto (X(\omega,\rho),t(\omega,\rho))
            \]
            is measurable and takes values in \(\mcB(\mcH)\times(0,\infty)\) on \(A\), it follows that \(Y\) is \((\mcF\otimes\borel{\states})|_A\)-measurable.
        
            Now define \(\widetilde Y:\Omega\times\states\to\mcB(\mcH)\) by
            \[
                \widetilde Y(\omega,\rho)
                :=
                \begin{cases}
                    Y(\omega,\rho), & (\omega,\rho)\in A,\\[0.5ex]
                    \rho_\star, & (\omega,\rho)\notin A,
                \end{cases}
            \]
            where \(\rho_\star\in\states\) is the fixed reference state from \eqref{eq:appendix-projective-action}. Since \(\rho_\star\) is constant and \(A\) is measurable, \(\widetilde Y\) is \((\mcF\otimes\borel{\states},\borel{\mcB(\mcH)})\)-measurable.
        
            By construction,
            \[
                \widetilde Y(\omega,\rho)=\mcT_{a;\omega}\proj\rho
                \qquad\text{for all }(\omega,\rho)\in\Omega\times\states.
            \]
            Moreover, \(\widetilde Y(\omega,\rho)\in\states\) for every \((\omega,\rho)\): if \(t(\omega,\rho)>0\), then \(\mcT_{a;\omega}(\rho)\ge0\) and normalization gives a state; if \(t(\omega,\rho)=0\), then \(\widetilde Y(\omega,\rho)=\rho_\star\in\states\).
            Since \(\states\) carries the subspace Borel \(\sigma\)-algebra inherited from \(\mcB(\mcH)\), it follows that \((\omega,\rho)\mapsto \mcT_{a;\omega}\proj\rho\) is \((\mcF\otimes\borel{\states},\borel{\states})\)-measurable.
        \end{proof}

    Using \Cref{lem:proj-action-jointly-measurable}, we obtain joint measurability of the finite-step skew-product trajectory by induction.

    \begin{cor}
    \label{lem:rho-n-measurable}
        Let \(\rho_0:\Omega\to\states\) be an \((\mcF,\borel{\states})\)-measurable random initial state.
        For each \(n\ge1\), the map
        \[
            (\omega,\bar a)
            \longmapsto
            \rho_n^{\rho_0}(\omega,\bar a)
            :=
            \left(\mcT_{a_n;\theta^n(\omega)}\circ\cdots\circ
                  \mcT_{a_1;\theta(\omega)}\right)\proj
            \left(\rho_0(\omega)\right)
        \]
        from \(\Omega\times\mcA^\mbN\) to \(\states\) is
        \((\mcF\otimes\Sigma,\borel{\states})\)-measurable, where
        \(\Sigma:=(2^\mcA)^{\otimes\mbN}\).
        In fact, \(\rho_n^{\rho_0}(\omega,\bar a)\) depends only on the first \(n\)
        coordinates \((a_1,\dots,a_n)\) of \(\bar a\), and the induced map
        \[
            (\omega,a_1,\dots,a_n)\longmapsto
            \rho_n^{\rho_0}(\omega,a_1,\dots,a_n)
        \]
        from \(\Omega\times\mcA^n\) to \(\states\) is
        \((\mcF\otimes 2^{\mcA^n},\borel{\states})\)-measurable.
    \end{cor}
        \begin{proof}
            For each \(a\in\mcA\), set
            \[
                F_a(\omega,\rho):=\mcT_{a;\omega}\proj\rho.
            \]
            By \Cref{lem:proj-action-jointly-measurable}, each
            \(F_a:\Omega\times\states\to\states\) is
            \((\mcF\otimes\borel{\states},\borel{\states})\)-measurable.
        
            Define recursively
            \[
                G_0(\omega):=\rho_0(\omega),
            \]
            and for \(r\ge1\),
            \[
                G_r(\omega,a_1,\dots,a_r)
                :=
                F_{a_r}\!\left(\theta^r(\omega),\,G_{r-1}(\omega,a_1,\dots,a_{r-1})\right).
            \]
            Since \(\theta^r\) is \(\mcF\)-measurable and \(\mcA^r\) carries the discrete \(\sigma\)-algebra, an induction on \(r\) shows that each  \(G_r:\Omega\times\mcA^r\to\states\) is
            \((\mcF\otimes 2^{\mcA^r},\borel{\states})\)-measurable.
        
            By construction,
            \[
                \rho_n^{\rho_0}(\omega,\bar a)
                =
                G_n\left(\omega,\pi_1(\bar a),\dots,\pi_n(\bar a)\right),
            \]
            so \(\rho_n^{\rho_0}\) depends only on the first \(n\) coordinates of \(\bar a\).
            Since the map
            \[
                (\omega,\bar a)\longmapsto (\omega,\pi_1(\bar a),\dots,\pi_n(\bar a))
            \]
            from \(\Omega\times\mcA^\mbN\) to \(\Omega\times\mcA^n\) is  \((\mcF\otimes\Sigma,\mcF\otimes 2^{\mcA^n})\)-measurable, it follows that
            \[
                (\omega,\bar a)\longmapsto \rho_n^{\rho_0}(\omega,\bar a)
            \]
            is \((\mcF\otimes\Sigma,\borel{\states})\)-measurable.
        \end{proof}

    It was proved in \cite{traj} that there exists a family of outcome laws $\{\mbQ_{\rho,\omega}\}_{\rho\in\states,\;\omega\in\Omega}$ on $\mcA^{\mbN}$ such that the map
    \[
        (\rho,\omega) \longmapsto \mbQ_{\rho,\omega}
    \]
    from $(\states\times\Omega,\borel{\states}\otimes\mcF)$ to $(\mcP(\mcA^\mbN),\borel{\mcP(\mcA^\mbN)})$ is measurable, where $\mcP(\mcA^\mbN)$ denotes the space of Borel probability measures on $\mcA^\mbN$, equipped with the Prokhorov metric.
    In particular, for each deterministic initial state $\rho\in\states$ the map $\omega\mapsto\mbQ_{\rho,\omega}$ is $(\mcF,\borel{\mcP(\mcA^\mbN)})$-measurable.
    We now elevate this to random initial states.

    \begin{lemma}
    \label{lem:Q-random-initial-state}
        Let $\rho_0:\Omega\to\states$ be an $(\mcF,\borel{\states})$-measurable random initial state. 
        Then the map
        \[
            \omega \longmapsto \mbQ_{\rho_0(\omega);\omega}
        \]
        from $(\Omega,\mcF)$ to $(\mcP(\mcA^\mbN),\borel{\mcP(\mcA^\mbN)})$
        is measurable.
    \end{lemma}
        \begin{proof}
            Consider the map
            \[
                F : \Omega \longrightarrow \states\times\Omega,
                \qquad
                F(\omega) := \left(\rho_0(\omega),\omega\right).
            \]
            We first check that $F$ is $(\mcF,\borel{\states}\otimes\mcF)$-measurable.
            The product $\sigma$-algebra $\borel{\states}\otimes\mcF$ is generated by rectangles $B\times F_0$ with $B\in\borel{\states}$ and $F_0\in\mcF$, and for such a rectangle we have
            \[
                F^{-1}(B\times F_0)
                = \{\omega : \rho_0(\omega)\in B,\ \omega\in F_0\}
                = \rho_0^{-1}(B)\cap F_0 \in \mcF,
            \]
            since $\rho_0$ is $(\mcF,\borel{\states})$-measurable and $F_0\in\mcF$.
            Hence $F$ is measurable.
        
            By the result from \cite[Lemma~A.2]{traj}, the map
            \[
                B : \states\times\Omega \to \mcP(\mcA^\mbN),
                \qquad
                B(\rho,\omega) := \mbQ_{\rho,\omega},
            \]
            is $(\borel{\states}\otimes\mcF,\borel{\mcP(\mcA^\mbN)})$-measurable.
            Therefore the composition
            \[
                \Omega \xrightarrow{\ F\ } \states\times\Omega
                   \xrightarrow{\ B\ } \mcP(\mcA^\mbN)
            \]
            is $(\mcF,\borel{\mcP(\mcA^\mbN)})$-measurable.
            But
            \[
                (B\circ F)(\omega)
                = B\left(\rho_0(\omega),\omega\right)
                = \mbQ_{\rho_0(\omega);\omega},
            \]
            so the map $\omega\mapsto\mbQ_{\rho_0(\omega);\omega}$ is measurable, as claimed.
        \end{proof}


\section{Periodic examples without instrument mixing}
\label{app:periodic-examples}


    These examples show that the mixing assumption for the original instrument is not necessary for an outcome CLT. Nonconstant periodic instruments fail this assumption, but state-independent emissions admit a functional CLT by grouping complete periods.
    An identity-channel specialization also fails full-state forgetting.

\begin{example}[Periodic state-independent emissions.]
\label{eg:periodic-state-independent-emissions}
    Let $\Omega=\mbZ/M\mbZ$ carry the uniform measure, with $\theta j=j+1\pmod M$, and take
    \[
     \mcT_{a;j}=p_a^{(j)}\Psi_a^{(j)},
     \qquad a\in\mcA,
    \]
    where each $p^{(j)}$ is a probability vector and each $\Psi_a^{(j)}$ is a quantum channel.
    The base is invertible and ergodic. 
    If the instrument process is nonconstant, it is not mixing: a nontrivial event $E$ determined by the current instrument recurs every $M$ steps, giving $\alpha(kM)\ge\pr(E)(1-\pr(E))>0$.

    For every measurable initial state, the outcomes are independent conditional on the phase, and their law is independent of that state.
    They satisfy the functional CLT, jointly for finitely many bounded cylinder observables, together with convergence of normalized variances, as follows.

    Fix the initial phase and a bounded $f:\mcA^h\to\mbR$.
        Group outcomes into periods and put
    \[
     \begin{aligned}
     C_k&=(A_{kM+1},\ldots,A_{kM+M}),\\
     H_k&=\sum_{j=1}^M
           f(A_{kM+j},\ldots,A_{kM+j+h-1}),
     \qquad
     r=\left\lceil\frac{h-1}{M}\right\rceil.
     \end{aligned}
    \]
    Under this conditional law, the $C_k$ are i.i.d., and $H_k$ is a bounded function of $C_k,\ldots,C_{k+r}$.
    Let $p(c)$ be the conditional probability of $C_0=c$ and fix a state $\eta\in\mbS_2$.
    On a one-point environment, define an instrument on $\mcB(\mbC^2)$ by
    \[
        \widetilde{\mcT}_c(X):=p(c)\tr{X}\eta,
        \qquad c\in\mcA^M.
    \]
    Its outcome sequence has the law of $(C_k)_{k\ge0}$, and its non-selective channel is $X\mapsto\tr{X}\eta$.
    Thus $\eta$ is its unique stationary state, its forgetting rates can be taken to be zero for every $n\ge1$, and its instrument-mixing coefficients vanish.
    This realization satisfies the standing restriction $d\ge2$, so \Cref{thm:universal-functional} applies to the blocked observables.
    Since $(H_k)$ is stationary and $r$-dependent, the mean and limiting variance per original time step are
    \[
         \mu_f=\frac{\mbE H_0}{M},
         \qquad
         \sigma_f^2=\frac1M\left(
               \operatorname{Var}H_0
               +2\sum_{\ell=1}^r\operatorname{Cov}(H_0,H_\ell)
                                 \right),
    \]
    where expectations and covariances are conditional on the fixed phase.
    
    The centered partial sums differ from their completed-period sums by at most $2M\|f\|_\infty$, uniformly in the time endpoint.
    Coupling two phases by a shift of $\ell<M$ coordinates changes their partial sums by at most $2\ell\|f\|_\infty$.
    These bounded errors transfer the functional limit to original time and show that $\mu_f$ and $\sigma_f^2$ are phase independent.
    Finite dependence gives variance convergence, which is preserved under these bounded errors by Cauchy--Schwarz inequality.
    The same argument applies to vectors of observables and then to the uniform mixture of phases; zero variance directions are allowed.
    The variance is computed from period covariances: the ungrouped covariance series need not converge.
    
    Taking $d\ge2$, nonconstant probability vectors $p^{(j)}$, and $\Psi_a^{(j)}=\mathrm{id}$ gives $\Phi_j=\mathrm{id}$.
    The original instrument then fails both mixing and full-state  forgetting, while the same outcome CLTs hold.
\end{example}

\paragraph{Special case: Periodic reset instruments.}
\label{par:periodic-reset-instruments}
    At phase $0$, take
    \[
         \mcT_{r;0}=a\mcR_\tau,
         \qquad
         \mcT_{b;0}=(1-a)\Psi,
         \qquad a\in[0,1],
    \]
    and at phase $j\in\{1,\ldots,M-1\}$ take $\mcT_{(j,u);j}(X)=q_{j,u}U_{j,u}XU_{j,u}^*$, where $(q_{j,u})_u$ is a probability vector and the $U_{j,u}$ are unitaries. 
    All other phase probabilities are zero.
    These are periodic state-independent emissions, so the preceding outcome CLT holds for every $a$, including $a=0$.
    For $a>0$, every full period contains a reset of probability $a$, giving augmented minorization at length $M$ and full-state forgetting with
    \[
     r_n=2(1-a)^{\lfloor n/M\rfloor}.
    \]


\end{appendix}


\setlength{\bibitemsep}{2pt}
\printbibliography[heading=bibintoc,title={References}]


\end{document}